\documentclass[10pt,a4paper]{article}
\usepackage{amsmath, amssymb, color}
\usepackage{graphicx,eurosym} 

\usepackage[T1]{fontenc}
\usepackage[utf8]{inputenc}
\usepackage{lmodern}
\usepackage[english]{babel}

\usepackage{fancyhdr}
\usepackage[table, svgnames, dvipsnames]{xcolor}
\usepackage[title]{appendix}

\usepackage[normalem]{ulem} 
\usepackage{hyperref}
\hypersetup{hidelinks}
\usepackage{tabularx}
\usepackage{multirow}
\usepackage{booktabs}
\usepackage{amssymb}
\usepackage{amsmath,mathrsfs}
\usepackage{amsthm} 
\usepackage{graphicx}
\usepackage{empheq}
\usepackage{color}
\usepackage{empheq}
\usepackage{float}
\usepackage{caption}
\usepackage{subcaption}
\usepackage{xcolor}
\usepackage{adjustbox}
\usepackage{listings}
\usepackage{algorithm}
\usepackage{algpseudocode}
\usepackage{bbm}
\usepackage{enumitem}
\usepackage{stmaryrd}

\usepackage{algorithm}
\usepackage{algpseudocode}
\algdef{SE}[SUBALG]{Indent}{EndIndent}{}{\algorithmicend\ }%
\algtext*{Indent}
\algtext*{EndIndent}

\definecolor{gris}{gray}{0.8}
\definecolor{grissombre}{gray}{0.4}
\definecolor{grisclair}{gray}{0.7}

\definecolor{bleu}{RGB}{41, 128, 185}
\definecolor{rouge}{RGB}{192, 57, 43}
\definecolor{vert}{RGB}{39, 174, 96}

\theoremstyle{plain}
\newtheorem{theorem}{Theorem}[section]
\newtheorem{proposition}[theorem]{Proposition}

\theoremstyle{definition}

\newtheorem{assumption}{Assumption}
\newtheorem*{sbjts}{SBJTS problem}

\theoremstyle{remark}
\newtheorem{remark}{Remark}[section]

\numberwithin{equation}{section}
\numberwithin{figure}{section}
\numberwithin{table}{section}

\def \R{\mathbb{R}}

\def \E{\mathbb{E}}

\def \P{\mathbb{P}}

\def \d{\mathrm{d}}
\def \trans{^{\scriptscriptstyle{\intercal}}} 

\usepackage{todonotes}

\title{
Schrödinger bridges with jumps for time series generation
\thanks{This work is supported by the Chair Deep Learning and Statistics Qube RT, 
the BNP-PAR Chair ``Futures of Quantitative Finance", the Chair ``Risques Financiers", by FiME, Laboratoire de Finance des March\'es de l'Energie, and the ``Finance and Sustainable Development'' EDF - CACIB Chair. }
}
\date{\today}

\author{
Stefano De Marco
\thanks{CMAP, École Polytechnique, Institut Polytechnique de Paris, Email: stefano.de-marco@polytechnique.edu}
\quad 
Huyên Pham
\thanks{CMAP, École Polytechnique, Institut Polytechnique de Paris, Email: huyen.pham@polytechnique.edu}
\quad 
Davide Zanni
\thanks{CMAP, École Polytechnique, Institut Polytechnique de Paris, Email: davide.zanni@polytechnique.edu} 
}

\begin{document}

\maketitle


\begin{abstract}
\noindent We study generative modeling for time series using entropic optimal transport and the Schrödinger bridge (SB) framework, with a focus on applications in finance and energy modeling. 
Extending the diffusion-based approach of Hamdouche, Henry-Labordère, Pham, \cite{pham_generative}, we introduce a jump–diffusion Schrödinger bridge model that allows for discontinuities in the generative dynamics. 
Starting from a Schrödinger bridge entropy minimization problem, we reformulate the task as a stochastic control problem whose solution characterizes the optimal controlled jump–diffusion process. When sampled on a fixed time grid, this process generates synthetic time series matching the joint distributions of the observed data.

The model is fully data-driven, as both the drift and the jump intensity are learned directly from the data. We propose practical algorithms for training, sampling, and hyperparameter calibration. Numerical experiments on simulated and real datasets, including financial and energy time series, show that incorporating jumps substantially improves the realism of the generated data, in particular by capturing abrupt movements, heavy tails, and regime changes that diffusion-only models fail to reproduce. Comparisons with state-of-the-art generative models highlight the benefits and limitations of the proposed approach.
\end{abstract}

\vspace{3mm}


\vspace{3mm} 


\section{Introduction}

\paragraph{The generative modeling task.} Generative models have gained increasing attention in machine learning and applied probability. The task of learning the underlying data distribution to synthesize realistic and novel samples has become crucial across diverse domains, such as the generation of high-fidelity images, coherent text, complex videos or molecular structures. Contemporary architectures generally fall into three categories: flow-matching models \cite{lipman2022flow}, based on learning a velocity field to find a deterministic flow which moves a source distribution (usually Gaussian) to the target distribution; denoising diffusion models \cite{ho2020denoising, song2021scorebased}, built through the reversion of a Markov chain or a diffusion SDE, often relying on learning the score function of the data distribution; Generative Adversarial Networks (GANs) \cite{goodfellow2014generative}, written as a zero-sum game between a generator and a discriminator. In quantitative finance, the realistic simulation of financial time series has attracted increasing interest in risk management, portfolio construction, stress testing, and scenario generation. Traditional parametric models frequently fail to capture the complex and empirical stylized facts of financial data, thereby motivating the shift toward data-driven generative approaches. However, faithfully reproducing these temporal dependencies remains a significant challenge, and universally accepted evaluation metrics for synthetic financial data are still missing. See \cite{cetingoz2025synthetic} for a complete overview.

In this work, we focus on generative modeling for time series based on optimal transport (OT) techniques. OT-based approaches seek to identify a transport mechanism that optimally moves a reference distribution into the empirical data distribution, minimizing a transport cost while preserving the structure of the data. This framework enables generative modeling by directly learning an optimal transport map or coupling, bypassing the need for adversarial training or the inversion of complex stochastic dynamics. Early explorations of this idea include training generative models by minimizing OT-based metrics like Sinkhorn divergences \cite{genevay2018learning}. Subsequent works have aimed to directly learn the optimal transport maps, such as through continuous-time Neural ODE flows \cite{onken2021ot} or neural approximations of the Monge map \cite{korotin2022neural}. Specifically for time series, causal optimal transport has been introduced to reproduce temporal dependencies and causal structures rather than merely matching static marginal distributions \cite{xu2020cot}. Our work extends this domain by proposing a novel generative model built upon a specific, powerful optimal transport technique: the Schrödinger bridge (SB) problem.

\paragraph{The Schrödinger bridge problem.} The Schrödinger bridge problem, first introduced by E.\ Schrödinger in the 1930s, is an entropy minimization problem that seeks to find the probability measure that is "as close as possible" to another reference probability, while satisfying prescribed marginal distributions. Formally, consider a reference path measure $\mathbb{Q}$ on the space $\Omega=C([0,1], \mathbb{R}^d)$, typically fixed as the Wiener measure, under which the canonical process $X=(X_t)_{t\geq 0}$ is a Brownian motion. The \textit{dynamic Schrödinger bridge problem} is then formulated as
\begin{equation*}
    \inf_{\mathbb{P}\in \mathcal{P}(\Omega): \,P_0=\mu_0,\, P_1=\mu_1} H( \mathbb{P} | \mathbb{Q}),
\end{equation*}
where $H(\mathbb{P} | \mathbb{Q})$ is the relative entropy (Kullback-Leibler divergence) between the two path measures, $P_0=\mathbb{P}\circ X_0^{-1}$, $P_1=\mathbb{P}\circ X_1^{-1}$, and $\mu_0$ and $\mu_1$ are the prescribed initial and terminal marginal distributions belonging to $\mathcal{P}(\mathbb{R}^d)$. The solution $\mathbb{P}^*$ is the probability law of a unique, constrained stochastic process that transports the initial distribution $\mu_0$ to the target distribution $\mu_1$ in an entropy-optimal way relative to the reference measure $\mathbb{Q}$. The \textit{static} counterpart of the problem is defined on the space of probabilities $\mathcal{P}(\mathbb{R}^d)$ as
\begin{equation*}
    \inf_{\pi\in \Pi(\mu_0,\mu_1)} H(\pi | Q_{01}),
\end{equation*}
where $Q_{01}(\d x\,\d y)=\mathbb{Q}((X_0,X_1)\in \d x\,\d y)$ is the joint law of the initial and final position of the reference process, and $\Pi(\mu_0,\mu_1)$ denotes the set of couplings with marginals $\mu_0$ and $\mu_1$. The classical Schrödinger bridge problem has been extensively studied: we refer to Léonard \cite{leonard2013survey} for a comprehensive overview of the Schrödinger bridge problem, including classical solution methods, dual formulation, and the characterization of the solution in terms of Schrödinger potentials. In this paper, we are interested in the application of the Schrödinger bridge problem to financial modeling and its extension to processes with discontinuities. This extension has been recently explored in depth in \cite{zlotchevski2024schr, zlotchevski2025jump}, where the authors define the minimization problem over the space of path measures on càdlàg trajectories, choosing as reference measure the law of a Lévy-Itô process. They provide the explicit characterization of the solution via the associated Schrödinger system and its potentials.

An important connection is given by the interpretation of the Schrödinger bridge problem as entropic optimal transport (EOT) problem, which can be viewed as a computationally tractable relaxation of classical optimal transport. Indeed, consider the static Schrödinger bridge problem with a reference measure $R\in \mathcal{P}(\mathbb{R}^d \times\mathbb{R}^d)$ of the form
\begin{equation*}
    \d R(x,y) \propto \exp\left(-\frac{1}{\varepsilon} c(x,y)\right)\d(\mu_0\otimes \mu_1)(x,y).
\end{equation*}
Then minimizing the relative entropy with respect to $R$ over couplings $\Pi(\mu_0,\mu_1)$ is equivalent, up to an additive constant, to the problem
\begin{equation*}
    \inf_{\pi\in \Pi(\mu_0,\mu_1)} \int c(x,y) \pi(\d x, \d y) + \varepsilon H(\pi | \mu_0\otimes \mu_1),
\end{equation*}
where $c$ is a given cost function and $H$ the Kullback-Leibler divergence. This formulation, combining a transport cost with an entropic regularization term, is precisely the EOT problem; popularized by the work of Cuturi \cite{cuturi2013sinkhorn}, it provides the foundation for the well-known Sinkhorn algorithm, which allows to numerically compute the entropy-regularized optimal coupling. By leveraging this connection, various numerical methods for approximating the SB solution have been developed: the Iterative Proportional Fitting (IPF) \cite{nutz2021introduction}, that it is based on the Sinkhorn algorithm, the Iterative Markovian Fitting (IMF) \cite{shi2023diffusion}, and faster techniques like the Light SB algorithm \cite{gushchin2024light}. These numerical algorithms are essential for turning the Schrödinger bridge problem theory into a practical tool for generative modeling.

\paragraph{Generative models based on Schrödinger bridges.} The Schrödinger bridge technology provides a powerful foundation for the  design of generative models. De Bortoli et al.\ \cite{de2021diffusion} was among the first to establish a rigorous connection between the SB problem and score-based generative models, introducing the Diffusion SB as a novel implementation of the IPF algorithm using score-based diffusion techniques. This approach enables the construction of bridges between arbitrary distributions over a finite time horizon, and leads to improved sampling efficiency compared with  classical score-based methods. Shi et al.\ \cite{shi2023diffusion} further developed this line of work by proposing Diffusion Schrödinger Bridge Matching, an algorithm which refines the practical approximation of Schrödinger bridges and offers a more scalable and robust alternative to IPF in high-dimensional settings.

In the context of time series, the Schrödinger bridge problem naturally induces a generative model on the path space. The resulting optimal stochastic process interpolates the joint distribution at fixed dates and, when sampled, produces synthetic time series whose joint finite-dimensional distribution matches those of the observed data. This dynamic optimal transport perspective combines the interpretability of OT  with the flexibility of stochastic modeling, making it particularly well suited for generating realistic time series while preserving both distributional and temporal consistency. In financial applications, \cite{henry2019martingale} introduced a martingale formulation of the Schrödinger bridge problem within an entropic optimal transport framework, incorporating martingale constraints that are essential for asset price modeling. Building on this idea, Hamdouche et al.\ \cite{pham_generative} rigorously defined and implemented a numerical algorithm to reproduce the joint distribution of financial time series using Schrödinger bridges.

In this paper, we extend the work of Hamdouche et al.\ \cite{pham_generative} to the setting of jump-diffusion processes. Extensions of generative models to dynamics with jumps have been explored, notably by generalizing diffusion-based stochastic differential equations to Lévy-driven dynamics \cite{yoon2023score, baule2025generative}. Such extensions are crucial for capturing heavy-tailed distributions and abrupt movements commonly observed in financial and energy time series. We adopt a similar strategy within the Schrödinger bridge framework by moving beyond the standard Wiener reference and instead considering the law of a process composed of a Brownian motion and a compound Poisson process. This leads to a Schrödinger bridge problem defined on the space of càdlàg paths, and requires new theoretical results for the SB problem to discontinuous dynamics. Beyond its theoretical interest, this extension is strongly motivated by practical considerations: our numerical experiments demonstrate that incorporating jumps substantially improves the realism of the generated time series compared with diffusion-only Schr\"odinger bridge models, particularly in terms of tail behavior, abrupt variations, and regime changes.

\paragraph{Our contributions and organization of the paper.} This paper makes three main contributions:
\begin{itemize}
\item  First, we introduce a new Schr\"odinger bridge–based generative model for time series driven by jump–diffusion dynamics. Extending the framework of Hamdouche et al.\ \cite{pham_generative}, we formulate the Schr\"odinger bridge problem on the space of probability laws of càdlàg processes, allowing the reference measure to be the law of a jump–diffusion process rather than a purely diffusive one. 
\item Second, we propose a systematic calibration procedure for identifying the key hyperparameters of the Schrödinger bridge generative model, which is essential to ensure numerical stability and robustness across different financial datasets. 
\item Third, through extensive numerical experiments on both financial and energy time series, we show that incorporating jumps leads to a substantial improvement in generative performance compared with diffusion-only Schr\"odinger bridge models, notably by better capturing fat tails, abrupt variations, and regime changes observed in real data. 
\end{itemize} 
The paper is organized as follows: in Section \ref{section_formulation}, we formulate the Schrödinger bridge problem for time series, and extend it to the setting of jump-diffusion processes.  Section \ref{section_main_theorem} presents the main theoretical result: Theorem \ref{main_theorem} characterizes the solution of the problem and explicitly identifies the optimal controlled dynamics of the generative process used to simulate synthetic time series on a fixed time grid.  In Section \ref{section_estimators}, we develop the methodology for estimating the drift and jump intensity of the optimal dynamics, with particular attention to the case where the reference jump component is Gaussian.  Section \ref{section_simulation_schemes} introduces two simulation schemes designed for sampling trajectories of the optimal jump-diffusion process, and Section \ref{section_calibration} describes the calibration procedure used to select the model  hyperparameters. In Section \ref{section_numerical_test}, we report numerical experiments on both simulated and real-world datasets, including financial and energy time series, presenting qualitative simulations, quantitative performance  metrics, and comparisons with state-of-the-art generative models. Finally, Appendix \ref{app_A} recalls some notions on the predictability of stochastic processes, and Appendix \ref{app_B} reports some additional numerical tests.

\section{Setting and problem formulation} \label{section_formulation}
Let $d, N\in\mathbb{N}$, and let $\mu$ denote a probability distribution on the space $(\R^d)^N$. Assume that the distribution $\mu$ is only accessible through a finite set of samples $x^{\mathrm{(1)}}, \ldots, x^{\mathrm{(M)}}$, each observed on a discrete time grid $\mathcal{T}:=\{t_i\}_{i=1,\ldots,N}$, where we set $T:=t_N$ as the terminal observation horizon. Our aim is to develop a method that enables the generation of new samples from this unknown target distribution $\mu$. 

Consider the space $\Omega = D([0,T];\mathbb{R}^d)$ of the càdlàg functions $\omega : [0,T]\to \mathbb{R}^d$, endowed with the Skorokhod topology, and let $\mathcal{F}$ denote its Borel $\sigma$-algebra. Consider the canonical process $X=(X_t)_{t\in[0,T]}$ with càdlàg paths defined by $X_t(\omega)=\omega(t)$. 
Let $\mathbb{F}^0=(\mathcal{F}^0_t)_{t\in[0,T]}$ be the canonical filtration, that is, $\mathcal{F}^0_t = \sigma(X_s, 0\leq s\leq t)$. Let $\mathcal{P}(\Omega)$ denote the space of probability measures on $\Omega$. We let $\P^0\in \mathcal{P}(\Omega)$ denote the law of a process satisfying
\begin{equation} \label{dynamics_x_P0}
    X_t = X_0 + \sigma \, W^0_t + \int_{(0,t]\times\mathbb{R}^d} z\,N(\d s,\d z), \quad t\in [0,T],
\end{equation}
where $X_0=0$, $W^0$ is a $d$-dimensional Brownian motion, $\sigma \in \mathbb{R}^{d\times d}$ is an invertible matrix, and $ N(\d t,\d z)$ is a Poisson random measure on $[0,T]\times \mathbb{R}^d$ independent of $W^0$, and with intensity measure
\begin{equation*}
    \lambda^0 \nu^0(\d z) \d t,
\end{equation*}
where $\lambda^0 >0$ is a constant jump intensity and $\nu^0 \in \mathcal{P}(\mathbb{R}^d)$ is the jump size distribution satisfying  $\nu^0(\{0\})=0$. In particular, the jump component has finite activity. Let $\mathcal{N}^{\P^0}$ denotes the collection of all subsets of $\mathcal{F}$-measurable $\P^0$-null sets. We consider the augmented filtration $\mathbb{F}=(\mathcal{F}_t)_{t\in[0,T]}$ defined by 
$\mathcal{F}_t := \bigcap_{s > t} \big(\mathcal{F}^0_s \vee \mathcal{N}^{\P^0}\big),$ for $t \in [0,T)$, and $\mathcal{F}_T = \mathcal{F}^0_T \vee \mathcal{N}^{\P^0}$. Hence, the filtration $\mathbb{F}$ is complete and right-continuous.
Throughout the paper we work under the standing assumption that $\sigma$ is non-degenerate and $\lambda^0 \neq 0$, unless explicitly stated otherwise. Our goal is to solve the following Schrödinger bridge problem with jumps for time series, which will be denoted by SBJTS. 

\begin{sbjts}
Find $\mathbb{P}^*\in \mathcal{P}(\Omega)$ solution to
\begin{equation*}
    \mathbb{P}^* \in \text{arg} \min_{\mathbb{P}\in\mathcal{P}^\mu_\mathcal{T}(\Omega)} H(\mathbb{P}|\mathbb{P}^0)
\end{equation*}
where $\mathcal{P}^\mu_\mathcal{T}(\Omega) = \{\mathbb{P} \in \mathcal{P}(\Omega) : \mathbb{P} \circ (X_{t_1}, \ldots, X_{t_N})^{-1} = \mu \}$, and $H(\cdot|\cdot)$ is the Kullback-Leibler divergence between two probability measures defined by 
\begin{equation*}
    H(\mathbb{P}|\mathbb{P}^0) = 
    \begin{cases}
        \E^{\P}\left[\ln \frac{\d\P}{\d\P^0}\right] &\text{if } \P \ll \P^0 ,\\
        + \infty &\text{otherwise.}
    \end{cases}
\end{equation*}
\end{sbjts}

\vspace{2mm}

We remark that the structure of our problem coincides with the one introduced in \cite{pham_generative}, except that we include an additional compound Poisson term in the dynamics of the process $X$ under the reference measure $\mathbb{P}^0$. This modification provides a natural extension of the original framework, allowing for the presence of jumps. 

To solve the SBJTS problem, we consider a stochastic control formulation, following the classical approach of \cite{dai1991stochastic}. Indeed, given $\P\in\mathcal{P}(\Omega)$ with finite relative entropy $H(\P|\P^0)<\infty$, by Girsanov's theorem (see \cite{léonard, jacod2013limit}) we can associate to $\P$ an $\R^d$ valued adapted process $\alpha=(\alpha_t)_{t\in[0,T]}$ and an $\R^+$ valued predictable process $\lambda = (\lambda_t)_{t\in[0,T]}$ with finite energy 
\begin{equation} \label{finite_energy_ass}
    \E^{\P}\left[\frac{1}{2}\int_0^T \|\sigma^{-1} \alpha_t\|^2 \d t +\int_{(0,T]\times\mathbb{R}^d} \left( \lambda_t(z)\ln\left(\frac{\lambda_t(z)}{\lambda^0}\right) + \lambda^0 - \lambda_t(z) \right)\nu^0(\d z)\d t \right] < \infty,    
\end{equation}
such that the Radon–Nikodym derivative corresponding to this change of measure satisfies
\begin{multline}\label{Z_T}
    \frac{\d\P}{\d\P^0} = \exp\Bigg(\int_0^T (\sigma^{-1} \alpha_t)\trans \, \d W^0_t - \frac{1}{2}\int_0^T \|\sigma^{-1}\alpha_t\|^2 \d t \\+ \int_{(0,T]\times\mathbb{R}^d} \ln\left(\frac{\lambda_t(z)}{\lambda^0}\right) N(\d t,\d z) 
    - \int_{(0,T]\times\mathbb{R}^d} (\lambda_t(z)-\lambda^0) \nu^0(\d z)\d t\Bigg) \mathbbm{1}_{\{\frac{\d\P}{\d\P^0}>0 \}} \, .
\end{multline}
Moreover, the process defined by 
\begin{equation}\label{new BM}
    W_t : = W^0_t - \int_0^t \sigma^{-1}\alpha_s \d s, \quad t\in [0,T]
\end{equation}
is a $\P$-Brownian motion, and the process
\begin{equation}\label{new N}
    \int_{(0,t]\times\mathbb{R}^d} \ln\left( \frac{\lambda_s(z)}{\lambda^0}\right) \Big(N(\d s,\d z) -  \lambda_s(z) \nu^0(\d z)\d s\Big), \quad t\in [0,T] 
\end{equation}
is a $\P$-martingale.

\begin{remark}
Expression \eqref{Z_T} is obtained by the combination of the results of C. Léonard for the continuous case (Theorem 2.3 in  \cite{leonard2013survey}) and the jump case (Theorem 2.9 in  \cite{leonard2013survey}). In particular, under the reference measure $\mathbb{P}^0$, the Brownian motion $W^0$ and the Poisson random measure $N(\d t, \d z)$ are independent. Consequently, the continuous martingale $M^c_t = \int_0^t (\sigma^{-1} \alpha_s)\trans \d W^0_s$ and the purely discontinuous martingale $M^d_t = \int_{(0,t]\times\mathbb{R}^d} \left(\frac{\lambda_s(z)}{\lambda^0}-1\right)(N(\d s, \d z) - \lambda^0\nu^0(\d z)\d s)$ have quadratic covariation satisfying $[M^c, M^d]_t = 0$, $\mathbb{P}^0$-a.s. By the multiplicative property of the Doléans-Dade exponential for orthogonal semimartingales, the joint density process is the product of the two exponential martingales.
\end{remark}

\vspace{5mm}

\noindent It follows that we can write the entropy $H(\P|\P^0)$ as 
\begin{equation*}
    H(\P|\P^0) = \mathbb{E}^\mathbb{P} \left[ \ln{\left(\mathbbm{1}_{\{\frac{\d\P}{\d\P^0}>0 \}} \right)} + \frac{1}{2}\int_0^T \|\sigma^{-1}\alpha_t\|^2 \d t + \int_{(0,T]\times\mathbb{R}^d} \left(\lambda_t(z)\ln \left(\frac{\lambda_t(z)}{\lambda^0}\right) - \lambda_t(z) +\lambda^0\right) \nu^0(\d z)\d t \right]
\end{equation*}
and, using that $\frac{\d\P}{\d\P^0}>0$, $\P-a.s.$, we get $H(\mathbb{P}|\P^0) = J(\alpha,\lambda)$, where $J(\alpha,\lambda)$ is the functional
\begin{equation*}
    J(\alpha,\lambda) = \mathbb{E}^\mathbb{P} \left[ \frac{1}{2}\int_0^T \|\sigma^{-1}\alpha_t\|^2 \d t + \int_{(0,T]\times\mathbb{R}^d} \left(\lambda_t(z)\ln \left(\frac{\lambda_t(z)}{\lambda^0}\right) - \lambda_t(z) +\lambda^0\right) \nu^0(\d z)\d t \right].
\end{equation*}
Hence the SBJTS problem is equivalent to the following formulation
\begin{equation*}
    (\textbf{P})\,\begin{cases}
        \text{minimize over $(\alpha,\lambda) \in \mathcal{A}$ the functional $J(\alpha,\lambda)$} \\
        \text{subject to: } X_t = X_0 + \int_0^t \alpha_s\,\d s + \sigma \, W_t + \int_{(0,t]\times\mathbb{R}^d} zN(\d s,\d z),\, X_0=0, \quad \P-a.s. \\
        \quad \quad \quad\quad \quad \mathbb{P}\circ(X_{t_1},\ldots, X_{t_N})^{-1}=\mu \\
    \end{cases}
\end{equation*}
where 
\begin{equation*}
    \mathcal{A} = \left\{ (\alpha, \lambda) : \alpha: \Omega\times [0,T] \to \mathbb{R}^d \,\text{ adapted, } \lambda: \Omega\times [0,T] \times \mathbb{R}^d \to \mathbb{R}^+ \,\text{ predictable, } J(\alpha,\lambda)<\infty \right\},
\end{equation*}
$W$ is $\mathbb{P}$-Brownian motion and $N(\d t,\d z)$ is the random counting measure of a marked point process with predictable compensator $\lambda_t(z)\nu^0(\d z)\d t$ under $\mathbb{P}$.
We denote by $V_{SBJTS}$ the infimum of this stochastic control problem:
\begin{equation*}
    V_{SBJTS} = \inf_{\alpha,\lambda \in \mathcal{A}^\mu_\mathcal{T}} J(\alpha,\lambda)  
\end{equation*}
where $\mathcal{A}^\mu_\mathcal{T}$ is the set of controls $(\alpha,\lambda) \in\mathcal{A}$ such that $\mathbb{P}\circ(X_{t_1},\ldots, X_{t_N})^{-1}=\mu$ and $\d X_t = \alpha_t\,\d t + \sigma \,\d W_t + \int_{\mathbb{R}^d} zN(\d t,\d z),$ for $t\leq T$, $X_0=0$.

Our goal is to prove the existence of an optimal pair $(\alpha^*, \lambda^*) \in \mathcal{A}^\mu_\mathcal{T}$ with an explicit representation that can eventually be estimated from empirical observations.
Once identified, these optimal drift and intensity can be used to simulate the corresponding controlled jump-diffusion process $X$, thereby generating samples from the time series distribution $\mu$ from the probability measure $\P^*$ on the path space $\Omega$.

\section{Solution of the Schrödinger bridge problem with jumps for time series} \label{section_main_theorem}
Let $\mu^0_\mathcal{T} = \mathbb{P}^0 \circ(X_{t_1},\ldots, X_{t_N})^{-1}$ be the distribution of the time series under the reference measure $\mathbb{P}^0$. We introduce the following assumptions on the target distribution $\mu$ and $\mu^0_\mathcal{T}$.

\begin{assumption} \label{assumption_1}
    Assume that $H(\mathbb{\mu}|\mathbb{\mu}^0_\mathcal{T}) < \infty$. 
\end{assumption}

\noindent Observe that Assumption \ref{assumption_1} implies that $\mu \ll \mu^0_\mathcal{T}$: we denote the corresponding Radon–Nikodym derivative as $\frac{\d\mu}{\d\mu^{0}_\mathcal{T}}$. Assumption \ref{assumption_1} is sufficient to guarantee that the SBJTS problem is well posed and that a minimizer exists. 
We set $\mathbf{X}_{t_i}:=(X_{t_1}, \ldots, X_{t_i})$, $\mathbf{x}_i := (x_1, \ldots, x_i)$ for all $(x_1, \ldots, x_N) \in (\mathbb{R}^d)^N$, $t_0 = 0$, and we denote by $X_{t^-}$ the left limit of the process $X$ in $t$. For $i=0,\ldots,N-1$, $t\in [t_i,t_{i+1}]$, $x_1, \ldots, x_i, x\in \mathbb{R}^d$ consider the function
\begin{equation} \label{h_i}
    h_i(t, x; \mathbf{x}_i) = \mathbb{E}^{\P^0}\left[ \frac{\d\mu}{\d\mu^0_\mathcal{T}}\left(X_{t_1}, \ldots, X_{t_N} \right) \; \big\vert \; \mathbf{X}_{t_i} = \mathbf{x}_i, \ X_{t} = x \right].
\end{equation}
The solution of the SBJTS problem is provided in the following theorem. 

\begin{theorem} \label{main_theorem}
The probability measure $\mathbb{P}^* \in \mathcal{P}(\Omega)$ defined by 
\begin{equation}\label{def_solution_SB}
    \frac{\d\mathbb{P}^*}{\d\P^0} = \frac{\d\mu}{\d\mu^{0}_\mathcal{T}}(X_{t_1}, \ldots, X_{t_N})
\end{equation}
is the unique solution of the SBJTS problem and it is the law of the optimal process
\begin{equation} \label{SDE_under P*}
    X_t = X_0 + \int_0^t \alpha^*_s \d s + \sigma \, W_t +\int_{(0,t]\times\R^d} z N(\d s,\d z), \quad t\in [0,T],
\end{equation}
with $X_0=0$, where $(W_t)_{t\in[0,T]}$ is a $\mathbb{P}^*$-Brownian motion, and $N(\d t,\d z)$ is an integer-valued random measure with predictable compensator $\lambda^*_t(z)\nu^0(dz)dt$ under $\mathbb{P}^*$. 
For $i=0,\ldots,N-1$, $t\in (t_i,t_{i+1}]$, the drift $\alpha^*_t$ and the intensity $\lambda^*_t(z)$ are given by 
\begin{align}
    \alpha^*_t &= a(t, X_t; \mathbf{X}_{t_i}), \label{def_of_drift} \\
    \lambda^*_t(z) &= \Lambda(t,X_{t^-},z;\mathbf{X}_{t_i}), \label{def_of_int}
\end{align} 
where, for $x_1, \ldots, x_i, x\in \mathbb{R}^d$, and $z\in\R^d$, we set 
\begin{align*}
    a(t,x;\mathbf{x}_i) &= \sigma \sigma\trans \nabla_x \ln h_i(t, x; \mathbf{x}_i), \\
    \Lambda(t,x,z;\mathbf{x}_i) &= \lambda^0 \frac{h_i(t,x+z;\mathbf{x}_i)}{h_i(t,x;\mathbf{x}_i)}. 
\end{align*}
The pair $(\alpha^*, \lambda^*)\in \mathcal{A}$ achieves the minimum of the problem (\textbf{P}), which satisfies
\begin{equation*}
    V_{SBJTS} = H(\P^* \vert \P^0) = H(\mu \vert \mu^0_\mathcal{T}).
\end{equation*}
\end{theorem}

\vspace{2mm}

\begin{proof}
\textit{Step 1: definition of the candidate optimal measure $\mathbb{P}^*$.} Following the approach of \cite{pham_generative}, we begin by introducing the definition of the probability measure $\P^*$ such that $\P^{*}\ll\P^0$. Observe that
\begin{align*}
    \E^{\P^0}\left[\frac{\d\mu}{\d\mu^0_\mathcal{T}}(X_{t_1},\ldots,X_{t_N}) \right] &= \int_{(\R^d)^N} \frac{\d\mu}{\d\mu^0_\mathcal{T}}(x_1,\ldots,x_N) \mu^0_\mathcal{T} (\d x_1,\ldots,\d x_N) =1. 
\end{align*}
This suggests the existence of a probability $\P^* \ll \P^0$ defined by
\begin{equation*}
    \frac{\d\P^*}{\d\P^0} = \frac{\d\mu}{\d\mu^0_\mathcal{T}}(X_{t_1},\ldots,X_{t_N})  
\end{equation*}
on $\mathcal{F}_T$, and we can define the martingale given by
\begin{equation}\label{def_martingale_Z}
    Z_t = \E^{\P^0} \left[\frac{\d\P^*}{\d\P^0} \; \big\vert \; \mathcal{F}_t \right], \; \; 0 \leq t \leq T,
\end{equation}
with $Z_0 = 1$. 

\textit{Step 2: expression of the martingale $Z$ as a Doléans-Dade exponential martingale.}
For any fixed index $i\in \{0,\ldots,N-1\}$, and $t\in [t_i,t_{i+1}]$, using the tower property of the conditional expectation and the Markovianity of the process $X$ under $\P^0$, we can write 
\begin{equation*}
    h_i(t,x;\,\mathbf{x}_{i}) = \E^{\P^0} \left[h_{i+1}(t_{i+1},X_{t_{i+1}};\, \mathbf{x}_{i},X_{t_{i+1}}) \, \big\vert \, X_t=x   \right],
\end{equation*}
for $x_1, \ldots, x_i, x \in \mathbb{R}^d$, and so, in particular, 
\begin{equation} \label{terminal-cond}
    h_i(t_{i+1},x; \mathbf{x}_i) =h_{i+1}(t_{i+1},x; \mathbf{x}_i,x),
\end{equation}
with the convention that $h_N(t_N,x;\,\mathbf{x}_{N-1},x) = \frac{\d\mu}{\d\mu^0_\mathcal{T}} (x_1, \ldots, x_{N-1},x)$. By Markov property, we get that
\begin{equation*}
    Z_t = h_i(t,X_t;\mathbf{X}_{t_i}), 
\end{equation*}
and, by the independence and stationarity of the increments of the Lévy process under $\mathbb{P}^0$, it holds
\begin{equation} \label{h_as _convolution}
    h_i(t,x;\mathbf{x}_{i}) = \E_{\xi_{i+1}, \ldots, \xi_N} \left[ \frac{\d\mu}{\d\mu^0_\mathcal{T}}(x_1,\ldots,x_i, x + \xi_{i+1}, \ldots, \; x +  \sum_{j=i+1}^N \xi_j) \right],
\end{equation}
where 
\begin{align*}
    \xi_{i+1} &= \sigma\sqrt{t_{i+1} - t}\, Y_{i+1} + \sum_{n= 1}^{\bar{N}_{t_{i+1}-t}} J^{(i+1)}_n, \\
    \xi_j &= \sigma\sqrt{t_j - t_{j-1}}\, Y_j + \sum_{n= 1}^{\bar{N}_{t_{j}-t_{j-1}}} J^{(j)}_n, \quad \text{for all } j=i+2,\ldots,N,
\end{align*}
where $Y_{j}, \, j=i+2,\ldots,N$, are i.i.d.\ $\mathcal{N}(0,I_d)$, $\bar{N}$ is a Poisson process with intensity $\lambda^0$, $(J^{(j)}_n)_{n\geq 1, j=i+1,\ldots,N}$ are i.i.d.\ random variables distributed according to $\nu^0$, independent of $Y_j$ and $\bar{N}$. 
Thanks to \eqref{h_as _convolution}, the function $h_i(t,x;\mathbf{x}_{i})$ is given by an integral with respect to a smooth density, that is the convolution of the Gaussian density with the jump distribution $\nu^0$. Therefore, the function $(t,x)\mapsto h_i(t, x; \mathbf{x}_i)$ belongs to $C^{1,2}([t_i,t_{i+1}]\times\R^d)$, and we can apply Itô's formula under the measure $\P^0$ to $Z_t = h_i(t,X_t; \mathbf{X}_{t_i})$. Since $Z$ is a $\P^0$-martingale, its predictable finite-variation part must vanish. Moreover, consider the stopping times 
\begin{equation*}
    \tau_i = \inf \{t\in (t_i,t_{i+1}]: Z_t=0 \text{ or } Z_{t^-}=0 \}, \quad i = 0, \ldots, N-1
\end{equation*}
with the standard convention that $\inf \emptyset = +\infty$.
Then, for $t\in(t_i, t_{i+1}]$ we have:
\begin{equation}\label{eq_exp_mg}
    \d Z_t = \left(\nabla_x h_i(t,X_{t^-}; \mathbf{X}_{t_i})\trans \, \sigma \d W^0_t + \int_{\R^d} (h_i(t,X_{t^-}+z; \mathbf{X}_{t_i})-h_i(t,X_{t^-}; \mathbf{X}_{t_i})) \tilde{N}(\d t,\d z)\right)  \mathbbm{1}_{\{t\leq\tau_i\}}
\end{equation}
where $\tilde{N}(\d t,\d z) := N(\d t,\d z) - \lambda^0\nu^0(\d z)\d t$ denotes the compensated Poisson random measure. On the interval $(t_i, \tau_i ]$, the left limit $Z_{t^-} = h_i(t, X_{t^-}; \mathbf{X}_{t_i})$ is strictly positive $\mathbb{P}^0$-a.s., hence we have
\begin{equation*}
    \d Z_t = Z_{t^-} \Bigg( \nabla_x \ln h_i(t,X_{t^-}; \mathbf{X}_{t_i})\trans \sigma \d W^0_t +\int_{\R^d} \Bigg(\lambda^0 \frac{h_i(t,X_{t^-}+z; \mathbf{X}_{t_i})}{h_i(t,X_{t^-}; \mathbf{X}_{t_i})}-\lambda^0\Bigg) \left(\frac{N(\d t,\d z)}{\lambda^0}- \nu^0(\d z)\d t\right) \Bigg)  \mathbbm{1}_{\{t\leq\tau_i\}} \, .
\end{equation*}
By defining the processes $\alpha^*_t=\nabla_x \ln h_i(t,X_{t}; \mathbf{X}_{t_i})$ and $\lambda^*_t(z)=\lambda^0 \frac{h_i(t,X_{t^-}+z; \mathbf{X}_{t_i})}{h_i(t,X_{t^-}; \mathbf{X}_{t_i})}$, we get
\begin{equation*} 
    \d Z_t = Z_{t^-} \Bigg( (\sigma^{-1} \alpha^*_t)\trans \d W^0_t +\int_{\R^d} (\lambda_t^*(z) -\lambda^0) \left(\frac{N(\d t,\d z)}{\lambda^0}- \nu^0(\d z)\d t\right) \Bigg) \mathbbm{1}_{\{t \le \tau_i\}}
\end{equation*}
whose solution is given by the stochastic exponential
\begin{multline*}
Z_{t} = Z_{t_i} \exp\Bigg(\int_{t_i}^{t} (\sigma^{-1} \alpha_s^*)^{\trans} \d W^0_s - \frac{1}{2}\int_{t_i}^{t } \|\sigma^{-1} \alpha_s^*\|^2  \d s \\
+\int_{({t_i},t ]\times\mathbb{R}^d} \ln\left(\frac{\lambda_s^*(z)}{\lambda^0}\right) N(\d s,\d z)-\int_{({t_i},t ]\times\mathbb{R}^d} (\lambda^*_s(z)-\lambda^0) \nu^0(\d z)\d s\Bigg) \mathbbm{1}_{\{t \le \tau_i\}}\,.
\end{multline*}
In order now to recover the expression of $\frac{\d\P^*}{\d\P^0}=Z_T$, we must concatenate these piecewise solutions over the grid $0 = t_0 < t_1 < \ldots < t_N = T$. This step relies on the following standard properties of non-negative martingales (see, e.g., \cite{revuz2013continuous}):
\begin{enumerate}
    \item The density process $Z=(Z_t)_{t\in[0,T]}$ is a non-negative $\P^0$-martingale. Because we only assume absolute continuity $\P^* \ll \P^0$ (and not strict equivalence), the event $\{Z_T = 0\}$ may occur with positive $\P^0$-probability.
    \item If $\tau_i(\omega) = t^* < +\infty$, for $i\in \{0,\ldots, N-1\}$, $\omega\in \Omega$, then $Z_t(\omega) = 0$ for all $t \in [t^*, T]$, $\P^0$-a.s. Consequently, $Z_T > 0$ almost surely if and only if $\tau_i = \infty$ for all $i \in \{0, \ldots, N-1\}$.    
\end{enumerate}
By concatenating the intervals, the exact global representation over $[0, T]$ is obtained by multiplying the standard stochastic exponential by the indicator function $\mathbbm{1}_{\{Z_T > 0\}} = \mathbbm{1}_{\{\frac{\d\P^*}{\d\P^0}>0 \}}$.
Eventually, the Radon-Nikodym derivative of $\mathbb{P}^*$ with respect to $\mathbb{P}^0$ is given by
\begin{multline*} 
    \frac{\d\P^*}{\d\P^0} =Z_T= \exp\Bigg(\int_0^T (\sigma^{-1} \alpha^*_s)\trans \d W^0_s - \frac{1}{2}\int_0^T \|\sigma^{-1} \alpha^*_s\|^2 \d s \\
    +\int_{(0,T] \times \mathbb{R}^d} \ln\left(\frac{\lambda^*_s(z)}{\lambda^0}\right) N(\d s,\d z)-\int_{(0,T] \times\mathbb{R}^d} (\lambda^*_s(z)-\lambda^0) \nu^0(\d z)\d s\Bigg)\mathbbm{1}_{\{\frac{\d\P^*}{\d\P^0}>0 \}} \, .
\end{multline*}  

\textit{Step 3: application of Girsanov's theorem}. 
Under the assumption that $H(\P^*\,\vert\,\P^0)<\infty$, by Girsanov's theorem the processes
\begin{align}
    &W^0_t - \int_0^t \sigma^{-1}\alpha^*_s \d s, \label{new_BM}  \\
    &\int_{(0,t]\times\mathbb{R}^d} \ln\left( \frac{\lambda^*_s(z)}{\lambda^0}\right) \Big( N(\d s,\d z) - \lambda^*_s(z) \nu^0(\d z)\d s \Big), \label{new_int}
\end{align}
for $t\in [0,T]$, are $\P^*$-martingales: \eqref{new_BM} defines a new $\P^*$-Brownian motion that we denote by $W$, while from \eqref{new_int} we get that $\lambda^*_t(z)\nu^0(\d z)\d t$ is the new compensator of the random counting measure $N(\d t,\d z)$ under $\P^*$. Hence, we can derive the dynamics of the process $X=(X_t)_{t\in[0,T]}$ under the probability $\P^*$: 
\begin{equation*}
    X_t = X_0 + \int_0^t \alpha^*_s\,\d s + \sigma \, W_t + \int_{(0,t]\times\mathbb{R}^d} zN(\d s,\d z), \quad t\in [0,T],
\end{equation*}
where $N(\d t,\d z)$ has compensator $\lambda^*_t(z)\nu^0(\d z)\d t$. 

\textit{Step 4: optimality of $\P^*$}. 
We show now that the processes $\alpha^*$ and $\lambda^*(z)$ are optimal for the minimization of $J(\alpha,\lambda)$. By definition we have that $J(\alpha^*,\lambda^*)=H(\P^*\vert\P^0) = H(\mu\vert\mu^0_\mathcal{T})$. Let $\alpha,\lambda$ be the drift and the intensity associated to another probability $\P$, with $H(\P|\P^0)<\infty$, and let $W^\P = (W^\P_t)_{t\in [0,T]}$ be a $\P$-Brownian motion. Then using Jensen's inequality and that $\frac{\d\P}{\d\P^0}$ is strictly positive $\P-a.s.$, we have: 
\begin{align*}
    1 &= \E^{\P^0} \left[ \frac{\d\mu}{\d\mu^0_\mathcal{T}}(X_{t_1}, \ldots,X_{t_N}) \right] = \E^{\P} \left[ \frac{\d\mu}{\d\mu^0_\mathcal{T}}(X_{t_1}, \ldots,X_{t_N}) \cdot \Big(\frac{\d\P}{\d\P^0}\Big)^{-1} \right]  \\
    &= \E^\P \Bigg[\exp\Bigg( \ln\frac{\d\mu}{\d\mu^0_\mathcal{T}}\left(X_{t_1}, \ldots, X_{t_N}\right) - \int_0^T (\sigma^{-1} \alpha_t)\trans \d W^0_t + \frac{1}{2}\int_0^T \|\sigma^{-1} \alpha_t \|^2 \d t  -\int_{(0,T]\times\mathbb{R}^d} \ln\frac{\lambda_t(z)}{\lambda^0} N(\d t,\d z)\\
    & \quad \quad + \int_{(0,T]\times\mathbb{R}^d} (\lambda_t(z)-\lambda^0) \nu^0(\d z)\d t\Bigg)\Bigg] \\
    &\geq \exp\Bigg(\E^\P \Bigg[\ln\frac{\d\mu}{\d\mu^0_\mathcal{T}}\left(X_{t_1}, \ldots, X_{t_N}\right) - \int_0^T (\sigma^{-1} \alpha_t)\trans \d W^\P_t - \frac{1}{2}\int_0^T \|\sigma^{-1} \alpha_t\|^2 \d t + \int_{(0,T]\times\mathbb{R}^d} (\lambda_t(z)-\lambda^0) \nu^0(\d z)\d t \\
    &\quad \quad - \int_{(0,T]\times\mathbb{R}^d} \ln\frac{\lambda_t(z)}{\lambda^0} \lambda_t(z) \nu^0(\d z)\d t - \int_{(0,T]\times\mathbb{R}^d} \ln\frac{\lambda_t(z)}{\lambda^0} (N(\d t,\d z) -\lambda_t(z) \nu^0(\d z)\d t) \Bigg]\Bigg)\\
    &= \exp \Bigg(H(\mu\vert\mu^0_\mathcal{T})- \E^\P\Bigg[ \int_{(0,T]\times\mathbb{R}^d} \Bigg(\lambda_t(z)\ln\frac{\lambda_t(z)}{\lambda^0}- \lambda_t(z)+\lambda^0\Bigg) \nu^0(\d z)\d t + \frac{1}{2}\int_0^T \|\sigma^{-1} \alpha_t\|^2 \d t \Bigg]\Bigg) \\
    &= \exp \Bigg(H(\mu\vert\mu^0_\mathcal{T})- J(\alpha,\lambda)\Bigg).
\end{align*}
Hence
\begin{equation*}
    J(\alpha^*,\lambda^*) =H(\mu\vert\mu^0_\mathcal{T}) \leq J(\alpha,\lambda),
\end{equation*}
and therefore also 
\begin{equation*}
    H(\P^* \vert \P^0) \leq H(\P \vert \P^0), \quad \forall \P \in \mathcal{P}^\mu_\mathcal{T}(\Omega).
\end{equation*}
This shows that $\P^*$ realizes the minimum in the SBJTS problem. The uniqueness of the minimizer $\P^*$ follows by the strict convexity of $H(\,\cdot\,|\P^0)$ and the convexity of the set of admissible measures $\mathcal{P}^\mu_\mathcal{T}$.

\textit{Step 5: check the constraint on the joint law.} 
Finally we check that the constraint on the joint distribution is satisfied. For any function $\phi$ bounded measurable on $(\R^d)^N$:
\begin{align*}
    \E^{\P^*}[\phi(X_{t_1},\ldots,X_{t_N})] &= \E^{\P^0}\left[ \frac{\d\mu}{\d\mu^0_\mathcal{T}}\left(X_{t_1}, \ldots, X_{t_N}\right) \phi(X_{t_1},\ldots,X_{t_N})\right]\\
    &= \int_{(\R^d)^N} \frac{\d\mu}{\d\mu^0_\mathcal{T}}\left(x_1, \ldots, x_N\right)\phi(x_1, \ldots, x_N) \mu^0_\mathcal{T}(\d x_1, \ldots,\d x_N)  \\
    &=\int_{(\R^d)^N}\phi(x_1, \ldots, x_N)\mu\left(\d x_1, \ldots, \d x_N\right).
\end{align*}
This proves that $\P^*\circ(X_{t_1}, \ldots, X_{t_N})^{-1}=\mu$ and completes the proof.
\end{proof}

\vspace{2mm}

\begin{remark} 
We justify here the decomposition of the stochastic integral with respect to the compensated Poisson random measure $\tilde{N}(\d t,\d z)$ appearing in expression \eqref{eq_exp_mg} into the difference of an integral with respect to the counting measure $N(\d t, \d z)$ and an integral with respect to its compensator $\lambda^0\nu^0(\d z)\d t$. More precisely, we want to show that, on the event $\{t\leq \tau_i\}$: 
\begin{equation*}
    \int_{t_i}^t \int_{\R^d} \Big|\frac{\lambda^*_s(z)}{\lambda^0}-1 \Big| \lambda^0\nu^0(\d z)\d s < \infty, \quad \P^0-a.s.
\end{equation*}
First, using the inequality $|y-1|\leq C(f(y)+1)$ for a sufficiently large constant $C > 0$, where $f(y)=y \ln y -y+1$, we get 
\begin{equation*}
    \int_{t_i}^t \int_{\R^d} \Big|\frac{\lambda^*_s(z)}{\lambda^0}-1 \Big| \lambda^0\nu^0(\d z)\d s \leq C I_t + C \lambda^0 (t - t_i), \quad \P^0-a.s.
\end{equation*}
where $I_t = \int_{t_i}^t \int_{\R^d} f\left(\frac{\lambda^*_s(z)}{\lambda^0} \right)\lambda^0\nu^0(\d z)\d s$. Since $(I_t)_t$ is non-decreasing, for $s<t$ we have $I_t \geq I_s$, and thus
\begin{align*}
    \mathbb{E}^{\P^0}[ Z_t I_t \mid \mathcal{F}_s ] \ge \mathbb{E}^{\P^0} [ Z_t I_s \mid \mathcal{F}_s ] = I_s \mathbb{E}^{\P^0}  [ Z_t \mid \mathcal{F}_s ] = Z_s I_s
\end{align*}
where we used that $(Z_t)_t$ is a $\mathbb P^0$-martingale. Hence $(Z_t I_t)_t$ is a $\P^0$-submartingale. Applying the submartingale property in the expectation yields:
\begin{align*}
    \mathbb{E}^{\P^0} [ Z_t I_t ] \le \mathbb{E}^{\P^0} [ Z_T I_T ] = \mathbb{E}^{\P^*} [ I_T ] \le H(\mathbb{P}^* \mid \mathbb{P}^0) < \infty\, .
\end{align*}
Consequently, the random variable $Z_tI_t$ is finite $\P^0-a.s$. Finally, since $Z_t>0$ on the event $\{t\leq \tau_i\}$, we obtain that $I_t<\infty$, $\P^0-a.s.$ in $\{t\leq \tau_i\}$, which allows to conclude.
\end{remark}

\vspace{2mm}

\paragraph{Pure jump case.}
Theorem \ref{main_theorem} holds also in the pure jump case. Indeed, if we assume that $\sigma$ is null, i.e.\ that $\mathbb{P}^0$ is the law of a compound Poisson process with jumps distributed according to $\nu^0$, then we can still write the SBJTS problem, define the solution as in \eqref{def_solution_SB} and the martingale
\begin{equation*}
    Z_t = \mathbb{E}^{\P^0}\left[ \frac{\d\mu}{\d\mu^0_\mathcal{T}}\left(X_{t_1}, \ldots, X_{t_N} \right) \; \big\vert \; \mathcal{F}_t\right], \quad t\leq T.
\end{equation*}
By Markov property we have $Z_t = h_i(t,X_t;\mathbf{X}_{t_i})$, for $t\in (t_i, t_{i+1}],\, i=0,\ldots,N-1$, and since $X=(X_t)_{t\geq 0}$ is a pure jump process, it holds
\begin{equation*}
    h_i(t,X_{t}; \mathbf{X}_{t_i}) = h_i(t_i,X_{t_i};\mathbf{X}_{t_i}) + \sum_{t_i<s\leq t} h_i(s,X_{s}; \mathbf{X}_{t_i}) - h_i(s,X_{s^-}; \mathbf{X}_{t_i})
\end{equation*}
which gives 
\begin{equation*}
    \d Z_t = \mathbbm{1}_{\{t \le \tau_i\}}\int_{\R^d} \Big(h_i(t,X_{t^-}+z; \mathbf{X}_{t_i})-h_i(t,X_{t^-}; \mathbf{X}_{t_i})\Big) \Big(N(\d t,\d z)-\lambda^0 \nu^0(\d z)\d t \Big)
\end{equation*}
as $Z$ is martingale, where the stopping time $\tau_i$ is defined as in the proof of Theorem \ref{main_theorem}. Hence we get that the new compensator of the random measure $N(\d t,\d z)$ under the measure $\mathbb{P}^*$ is $\lambda^*_t(z)\nu^0(\d z)\d t$, where
\begin{equation*}
    \lambda^*_t(z) = \lambda^0 \frac{h_i(t,X_{t^-}+z; \mathbf{X}_{t_i})}{h_i(t,X_{t^-}; \mathbf{X}_{t_i})},
\end{equation*}
and the dynamics of the optimal process under $\mathbb{P}^*$ is
\begin{equation*} 
    X_t = X_0 + \int_{(0,t]\times \R^d} z N(\d s,\d z), \quad t\in [0,T].
\end{equation*}

\vspace{5mm}
Thanks to Theorem \ref{main_theorem}, we obtain the dynamics \eqref{SDE_under P*} of the optimal process $X$ under the probability measure $\mathbb{P}^*$ which solves the SBJTS problem. Our objective is to simulate this optimal dynamics. As follows from \eqref{def_of_drift} and \eqref{def_of_int}, both the drift and the jump intensity are time, state and path dependent. This path dependence is essential, as it requires the construction of estimators to approximate the optimal drift and jump intensity that have to be evaluated along each simulated trajectory.

\section{Approximation of the optimal drift and intensity} \label{section_estimators}

In this section, we focus on the expressions of the optimal drift and intensity of jumps to highlight how we can get fully data driven estimators that allow the simulation of the generative process \eqref{SDE_under P*}. We first derive an explicit expression for the function $h_i(t,x; \mathbf{x}_{i})$, and then we present the estimators that are used in practice.

\subsection{Explicit formula for the drift and intensity of jumps}
First, using the notation introduced in the proof of Theorem \ref{main_theorem}, we have that the increments of the process $X$ under $\mathbb{P}^0$ satisfy
\begin{equation} \label{decomp_incr}
    X_{t_{i+1}} - X_{t} \stackrel{law}{=} \sigma\sqrt{t_{i+1}-t}\, Y + \sum_{n= 1}^{\bar{N}_{t_{i+1}-t}} J_n
\end{equation}
for $t\in [t_i,t_{i+1})$, $i=0,\ldots, N-1$, where $Y$ is a standard Gaussian random variable $\mathcal{N}(0,I_d)$, $\bar{N}$ is a Poisson process with intensity $\lambda^0$, $(J_n)_{n\geq 1}$ are i.i.d.\ random variables distributed according to $\nu^0$, independent of $Y$ and $\bar{N}$. Conditioning on the values of $\bar{N}$, we can explicitly write the density function of the increment $X_{t_{i+1}} - X_{t}$ as 
\begin{equation} \label{density_increments}
    f^0_{t_{i+1}-t}(z) = \frac{e^{-\lambda^0 (t_{i+1}-t)}}{\left(2\pi (t_{i+1}-t) \right)^\frac{d}{2} |\det \sigma|}\sum_{k\geq 0} \frac{(\lambda^0 (t_{i+1}-t))^k}{k!} \int_{\R^d} \exp{\left(-\frac{\|  \sigma^{-1} (z-y) \|^2}{2(t_{i+1}-t)} \right)} (\nu^0)^{*k}(\d y),
\end{equation}
for $z\in\R^d$, where we denote by $(\nu^0)^{*k}$ the $k$-times convolution product of the measure $\nu^0$.

We state the following proposition using that $\mu^0_\mathcal{T}$ admits a density with respect to the Lebesgue measure, as by definition $\mu^0_\mathcal{T} = \mathbb{P}^0 \circ(X_{t_1},\ldots, X_{t_N})^{-1}$, where $X$ follows the dynamics \eqref{dynamics_x_P0}. Moreover, as $\mu \ll \mu^0_\mathcal{T}$ by hypothesis, $\mu$ admits also a density with respect to the Lebesgue measure. By a slight abuse of notation, we denote these densities respectively by $(x_1,\ldots,x_N) \mapsto \mu^0_\mathcal{T}(x_1,\ldots,x_N)$ and $(x_1,\ldots,x_N) \mapsto \mu(x_1,\ldots,x_N)$. We introduce a further assumption on the measure $\mu$ to derive expressions of the optimal drift and intensity of jumps involving conditional expectations under $\mu$. 
\begin{assumption}\label{assumption_2}
    Assume that $\E_\mu \left[ \frac{1}{f^0_{t_{i+1}-t_i} (X_{t_{i+1}}-x_i)} \right]<\infty$, for $i=0,\ldots, N-1$, and $x_i\in \R^d$.
\end{assumption}

\begin{remark}[Probabilistic interpretation and validity of Assumption \ref{assumption_2}]
\label{rem:assumption2}
Assumption \ref{assumption_2} is a tail-compatibility requirement on the target distribution $\mu$ relative to the reference transition density $f^0_{\Delta t_i}$, ensuring the applicability of the Dominated Convergence Theorem in Proposition \ref{prop_estimators}. We give some examples to provide concrete probabilistic intuition of this assumption. 
\begin{enumerate}
    \item \textbf{Pure diffusion case:} If  $\lambda^0 = 0$, the transition density $f^0_{\Delta t_i}$ is purely Gaussian. Up to a constant, Assumption \ref{assumption_2} reduces to the exponential-quadratic moment condition
    \[
    \mathbb{E}_\mu \left[ \exp\left( \frac{\| \sigma^{-1}(X_{t_{i+1}} - x_i) \|^2}{2 \Delta t_i} \right) \right] < \infty.
    \]
    This is strictly stronger than sub-Gaussianity; for instance, in dimension $d=1$ with Gaussian target $X_{t_{i+1}} \sim \mathcal{N}(0, v)$, it requires $v < \sigma^2 \Delta t_i$, meaning the target marginal must have lighter tails than the reference transition.
    \item \textbf{Gaussian jump distribution.} If $\lambda^0 > 0, \nu^0 \sim \mathcal{N}(0, \Gamma)$, by lower-bounding $f^0_{\Delta t_i}$ by the $k$-th term of the sum \eqref{density_increments}, it suffices that $\mathbb{E}_\mu \left[ \exp\left( \frac{1}{2} (X_{t_{i+1}} - x_i)\trans \Sigma_k^{-1} (X_{t_{i+1}} - x_i) \right) \right] < \infty$ holds for \emph{at least one} $k \ge 0$, where $\Sigma_k = \Delta t_i \sigma \sigma\trans + k \Gamma$. For example, in dimension $d=1$ with Gaussian target $X_{t_{i+1}} \sim \mathcal{N}(0, v)$, this requires $v < \sigma^2 \Delta t_i + k \gamma^2$. Since $k \in \mathbb{N}$ can be chosen arbitrarily large, the term $\sigma^2 \Delta t_i + k \gamma^2$ can always exceed any finite target variance $v$, guaranteeing that Assumption \ref{assumption_2} holds. 
    \item \textbf{Heavy-tailed jumps:} Suppose that the jump size distribution $\nu^0$ admits a density satisfying a  polynomial lower bound of order $\gamma>1$, namely
    $p_{\nu^0}(z)\geq \frac{c}{1+|z|^\gamma},\, z\in\mathbb R,$
    for some $c>0$. Then the convolution structure of the transition density implies that
    $\frac{1}{f^0_{\Delta t_i}(z)}\leq C(1+|z|^\gamma)$
    for some constant $C>0$. 
    Consequently, Assumption \ref{assumption_2} is satisfied once that the moment condition $\mathbb{E}_\mu [|X_{t_{i+1}}|^\gamma] < \infty$ holds, significantly weakening the integrability requirement on $\mu$.
    \item \textbf{Compact support:} If $\mu$ has bounded support, $z \mapsto \frac{1}{f^0_{\Delta t_i}(z)} $ attains a finite supremum on the support domain, making Assumption \ref{assumption_2} trivially satisfied.
\end{enumerate}
\end{remark}

\begin{proposition}\label{prop_estimators}
For $i=0,\ldots,N-1$, $\mathbf{x}_i\in(\R^d)^i$, $x\in \mathbb{R}^d$ and $z\in\R^d$, the functions $a(t,x;\mathbf{x}_i)$ and $\Lambda(t,x,z;\mathbf{x}_{i})$ defined in Theorem \ref{main_theorem} are given by
\begin{align}
     a(t,x;\mathbf{x}_i) &= \sigma \sigma\trans \frac{\E_\mu [\nabla_x F_i(t,x_i,x,X_{t_{i+1}}) \vert \mathbf{X}_{t_i} = \mathbf{x}_{i}]}{\E_\mu [F_i(t,x_i,x,X_{t_{i+1}}) \vert \mathbf{X}_{t_i} = \mathbf{x}_{i}]}, \quad t\in [t_i,t_{i+1}), \label{a general} \\
     \Lambda(t,x,z;\mathbf{x}_{i}) &= \lambda^0 \frac{\E_\mu [ F_i(t,x_i,x+z,X_{t_{i+1}}) \vert \mathbf{X}_{t_i} = \mathbf{x}_{i}]}{\E_\mu [F_i(t,x_i,x,X_{t_{i+1}}) \vert \mathbf{X}_{t_i} = \mathbf{x}_{i}]}, \quad t\in (t_i,t_{i+1}), \label{lambda general}
\end{align}
where 
\begin{equation} \label{funct_Fi}
    F_i(t,x_i,x,x_{i+1})=\frac{f^0_{t_{i+1}-t}(x_{i+1}-x)} {f^0_{t_{i+1}-t_i}(x_{i+1}-x_i)},
\end{equation}
the function $f^0_{t_{i+1}-t}$ defined in \eqref{density_increments} is the density of the increment $X_{t_{i+1}}-X_t$, for $t\in[t_i,t_{i+1})$, under the reference measure $\mathbb{P}^0$, and $\mathbb{E}_\mu$ denotes the expectation under $\mu$. 
\end{proposition}
 
\vspace{2mm}

\begin{proof}
Consider $\mu^0_\mathcal{T}(x_1,\ldots,x_N)$ the density with respect to the Lebesgue measure of the vector $(X_{t_1},\ldots,X_{t_N})$ under $\mathbb{P}^0$. By independence of the increments, we can write 
\begin{equation} \label{decomp_density}
    \mu^0_\mathcal{T}(x_1,\ldots,x_N) = \prod_{j=0}^{N-1} f^0_{t_{j+1}-t_{j}}(x_{j+1}-x_{j}),
\end{equation}
where we set $x_0=0$ and $t_0=0$. For fixed $i=0,\ldots,N-1$, $t\in (t_i,t_{i+1})$, $\mathbf{x}_i\in(\R^d)^i$, and $x\in \mathbb{R}^d$, we can rewrite the function $h_i(t,x;\mathbf{x}_{i})$ using \eqref{decomp_density} as
\begin{align*}
    h_i(t,x;\mathbf{x}_{i}) &= \E^{\P^0} \left[\frac{\d \mu}{\d \mu^0_\mathcal{T}}(x_1,\ldots,x_i, X_{t_{i+1}}, \ldots,X_{t_N}) \; \big\vert \; X_t=x\right] \\
    &= \int_{(\R^d)^{N-i}} \frac{\mu(x_1,\ldots,x_N)}{\mu^0_i(x_1,\ldots,x_i)} \frac{f^0_{t_{i+1}-t}(x_{i+1}-x)} {f^0_{t_{i+1}-t_i}(x_{i+1}-x_i)} \, \d x_{i+1} \cdots \d x_N
\end{align*}
where we call $\mu^0_i(x_1,\ldots,x_i)$ the density of $(X_{t_1},\ldots,X_{t_i})$ under $\mu^0_\mathcal{T}$. Hence, if we introduce the function
\begin{equation*} 
    F_i(t,x_i,x,x_{i+1})=\frac{f^0_{t_{i+1}-t}(x_{i+1}-x)} {f^0_{t_{i+1}-t_i}(x_{i+1}-x_i)},
\end{equation*}
and we denote by $\mu_i(x_1,\ldots,x_i)$ the density of $(X_{t_1},\ldots,X_{t_i})$ under $\mu$ with respect to the Lebesgue measure, i.e. 
\begin{equation*}
    \mu_i(x_1,\ldots,x_i) = \int_{(\R^d)^{N-i}} \mu(x_1,\ldots,x_N)\, \d x_{i+1} \cdots \d x_N,
\end{equation*} 
we obtain 
\begin{align*}
    h_i(t,x;\mathbf{x}_{i}) &= \int_{(\R^d)^{N-i}} F_i(t,x_i,x,x_{i+1}) \frac{\mu_i(x_1,\ldots,x_i)}{\mu^0_i(x_1,\ldots,x_i)} \frac{\mu(x_1,\ldots,x_N)}{\mu_i(x_1,\ldots,x_i)}\, \d x_{i+1} \cdots \d x_N \\
    &= C\int_{(\R^d)^{N-i}} F_i(t,x_i,x,x_{i+1})  \frac{\mu(x_1,\ldots,x_N)}{\mu_i(x_1,\ldots,x_i)}\, \d x_{i+1} \cdots \d x_N,
\end{align*}
with $C$ a constant depending only on $(x_1, \ldots, x_i)$. Therefore,
\begin{equation} \label{h_i as mean in mu}
    h_i(t,x;\mathbf{x}_{i}) = C\, \E_\mu [F_i(t,x_i,x,X_{t_{i+1}}) \vert \mathbf{X}_{t_i} = \mathbf{x}_{i}]. 
\end{equation}
This allows to get the expression
\begin{equation} \label{computation_drift}
    a(t,x;\mathbf{x}_i) = \sigma \sigma\trans \frac{\nabla_x h_i(t, x; \mathbf{x}_i)}{h_i(t, x; \mathbf{x}_i)} = \sigma \sigma\trans \frac{\E_\mu [\nabla_x F_i(t,x_i,x,X_{t_{i+1}}) \vert \mathbf{X}_{t_i} = \mathbf{x}_{i}]}{\E_\mu [F_i(t,x_i,x,X_{t_{i+1}}) \vert \mathbf{X}_{t_i} = \mathbf{x}_{i}]}, 
\end{equation}
which is also well defined at $t=t_i$ thanks to \eqref{terminal-cond}, and 
\begin{align} \label{computation_intensity}
    \Lambda(t,x,z;\mathbf{x}_{i}) =\lambda^0 \frac{h_i(t,x+z; \mathbf{x}_{i})}{h_i(t,x; \mathbf{x}_{i})} = \lambda^0 \frac{\E_\mu [ F_i(t,x_i,x+z,X_{t_{i+1}}) \vert \mathbf{X}_{t_i} = \mathbf{x}_{i}]}{\E_\mu [F_i(t,x_i,x,X_{t_{i+1}}) \vert \mathbf{X}_{t_i} = \mathbf{x}_{i}]}.
\end{align}
We conclude the proof showing that we can interchange the gradient operator $\nabla_x$ and the conditional expectation $\mathbb{E}_\mu$ in \eqref{computation_drift}, as a consequence of the Dominated Convergence Theorem. First, the differentiability of the map $x \mapsto F_i(t,x_i,x,X_{t_{i+1}})$ follows directly from the smoothness of the reference transition density function $f^0_{t_{i+1}-t}(\cdot)$ for all $t < t_{i+1}$. Moreover, for fixed $t \in [t_i, t_{i+1})$ and $x \in \mathbb{R}^d$, relation \eqref{h_i as mean in mu} guarantees that the conditional expectation is finite: 
$$\mathbb{E}_\mu [F_i(t,x_i,x,X_{t_{i+1}}) \mid \mathbf{X}_{t_i} = \mathbf{x}_{i}] = \frac{1}{C} h_i(t,x;\mathbf{x}_{i}) < \infty\, .$$
Next, we must find an integrable dominating function $b(x_i, X_{t_{i+1}})$, independent of $x$, such that
\begin{equation*}
    |\nabla_x \,F_i(t,x_i,x,X_{t_{i+1}})| \leq b(x_i, X_{t_{i+1}}), \quad \mu\text{-a.s.}
\end{equation*}
Using the derivative properties of the jump-diffusion convolution structure, the gradient of the numerator with respect to $x$ yields:
\begin{align*}
    \nabla_xf^0_{t_{i+1}-t}(X_{t_{i+1}}-x) = \sum_{k\geq 0} w_k \int_{\R^d} \sigma^{-1}\frac{X_{t_{i+1}}-x-y}{t_{i+1}-t} g_\sigma(X_{t_{i+1}}-x-y) (\nu_0)^{*k}(\d y)
\end{align*}
where $w_k=e^{-\lambda^0 (t_{i+1}-t)}\frac{(\lambda^0 (t_{i+1}-t))^k}{k!}$ and $g_\sigma$ denotes the Gaussian density. Taking the absolute value and applying the triangle inequality, we obtain:
\begin{align*}
    |\nabla_x \,F_i(t,x_i,x,X_{t_{i+1}})| &\leq \frac{1}{f^0_{t_{i+1}-t_i}(X_{t_{i+1}}-x_i)} \sum_{k\geq 0} w_k \int_{\R^d} \|\sigma^{-1}\|\, \frac{|X_{t_{i+1}}-x-y|}{t_{i+1}-t} g_\sigma(X_{t_{i+1}}-x-y) (\nu_0)^{*k}(\d y) \\
     &\leq \frac{1}{f^0_{t_{i+1}-t_i} (X_{t_{i+1}}-x_i)} \sup_{z\in \R^d} \Big\{\|\sigma^{-1}\|\, \frac{|z|}{t_{i+1}-t} g_\sigma(z) \Big\}
\end{align*}
as the function $z \mapsto \| \sigma^{-1} \| \frac{|z|}{t_{i+1}-t} g_\sigma(z)$ is continuous and decays exponentially fast at infinity, hence it has a finite global supremum over $\mathbb{R}^d$. Finally, we can use a global constant to bound the numerator and get 
\begin{equation*}
    |\nabla_x \,F_i(t,x_i,x,X_{t_{i+1}})|\leq \frac{\bar{C}}{f^0_{t_{i+1}-t_i} (X_{t_{i+1}}-x_i)} =: b(x_i, X_{t_{i+1}}). 
\end{equation*}
Thanks to Assumption \ref{assumption_2}, we conclude that $b(x_i, X_{t_{i+1}}) \in L^1(\mu)$.
\end{proof}

\vspace{2mm}

\begin{remark}
At the time values $\{t_i\}_{i=1, \ldots, N-1}$ the value of the intensity function $\lambda^*_{t_{i+1}}(z)= \Lambda(t_{i+1},X_{t_{i+1}^-},z;\mathbf{X}_{t_i})$ is given by following expression:
\begin{align*}
    \Lambda(t_{i+1},x,z;\mathbf{x}_{i}) &= \lambda^0 \frac{\E^{\P^0} \left[\frac{\d \mu}{\d \mu^0_\mathcal{T}}(x_1,\ldots,x_i, X_{t_{i+1}}, \ldots,X_{t_N}) \; \big\vert \; X_{t_{i+1}}=x+z\right]}{\E^{\P^0} \left[\frac{\d \mu}{\d \mu^0_\mathcal{T}}(x_1,\ldots,x_i, X_{t_{i+1}}, \ldots,X_{t_N}) \; \big\vert \; X_{t_{i+1}}=x\right]}\\
    &= \lambda^0 \frac{\mu_{i+1}(x_1, \ldots, x_{i}, x+z)}{\mu^0_{i+1}(x_1, \ldots, x_{i}, x+z)} \frac{\mu^0_{i+1}(x_1, \ldots, x_{i}, x)}{\mu_{i+1}(x_1, \ldots, x_{i}, x)} 
\end{align*}
hence we get
\begin{align} \label{int_at_grid_time}
    \Lambda(t_{i+1},x,z;\mathbf{x}_{i}) = \lambda^0 \frac{f^0_{t_{i+1}- t_{i}}(x-x_{i})}{f^0_{t_{i+1}-t_{i}}(x+z-x_{i})}\frac{\mu_{i+1}(x_1, \ldots, x_{i}, x+z)}{\mu_{i+1}(x_1, \ldots, x_{i}, x)}.
\end{align} 
Notice that by Theorem \ref{main_theorem} we get that $\lambda^*_t(z) = \Lambda(t,X_{t^-},z;\mathbf{X}_{t_i})$ is well-defined and predictable with respect to the natural filtration of the process, as required for the construction of the controlled jump–diffusion. This predictability property is explicit in \eqref{lambda general} and \eqref{int_at_grid_time}, where the dependence on the past is expressed exclusively through values of the time series observed at strictly earlier times.
\end{remark}

\vspace{2mm}

Proposition \ref{prop_estimators} gives expressions of the drift and intensity of jumps involving the expectation with respect to $\mu$. This is particularly convenient as this formulation allows to construct estimators in which the expectation under $\mu$ is approximated using the initial time series sampled from the target distribution. Plugging the density \eqref{density_increments} into the function $F_i(t,x_i,x,x_{t_{i+1}})$, we achieve the expressions for the drift and intensity of jumps. In the numerical simulations we will fix a particular choice for the distribution $\nu^0$ in order to have tractable estimators.

\subsection{Kernel regression estimators}
We derive the approximation through classical kernel methods of the conditional expectation under the target distribution $\mu$ that appears in the expressions of the optimal drift and intensity. Consider data samples $\mathbf{X}^{(m)} = (X^{(m)}_{t_1},\ldots,X^{(m)}_{t_N}), \, m=1,\ldots,M$ from $\mu$. Starting from the expressions \eqref{a general} and \eqref{lambda general}, we can estimate the conditional expectation under $\mu$ using the Nadaraya-Watson kernel estimator. For $i=0,\ldots,N-1$ fixed, $\mathbf{x}_i \in(\R^d)^i$, $x, z\in\R^d$ we get 
\begin{equation} \label{drift app}
    \hat{a}(t,x;\mathbf{x}_i) = \sigma \sigma\trans \frac{\displaystyle\sum_{m=1}^M \nabla_x F_i(t,x_i,x,X^{(m)}_{t_{i+1}}) \,\mathbf{K}^i_{\mathbf{h}_i}(\mathbf{x}_{i}-\mathbf{X}^{(m)}_{t_i}) }{\displaystyle\sum_{m=1}^M F_i(t,x_i,x,X^{(m)}_{t_{i+1}}) \,\mathbf{K}^i_{\mathbf{h}_i}(\mathbf{x}_{i}-\mathbf{X}^{(m)}_{t_i})}, \quad t\in [t_i, t_{i+1}),
\end{equation}
and 
\begin{align} \label{int app}
    \hat{\Lambda}(t,x,z;\mathbf{x}_{i}) = \lambda^0 \frac{\displaystyle\sum_{m=1}^M F_i(t,x_i,x+z,X^{(m)}_{t_{i+1}}) \, \mathbf{K}^i_{\mathbf{h}_i}(\mathbf{x}_{i}-\mathbf{X}^{(m)}_{t_i})}{\displaystyle\sum_{m=1}^M F_i(t,x_i,x,X^{(m)}_{t_{i+1}}) \, \mathbf{K}^i_{\mathbf{h}_i}(\mathbf{x}_{i}-\mathbf{X}^{(m)}_{t_i})}, \quad t\in (t_i, t_{i+1}).
\end{align}
Here $\mathbf{K}^i$ is a multivariate kernel function operating on $i$ arguments of dimension $d$. Given a bandwidth vector $\mathbf{h}_i=(h_1, \ldots, h_i)\in (0, \infty)^i$, we define the rescaled kernel by $\mathbf{K}^i_{\mathbf{h}_i}(\mathbf{x}_{i}) = \frac{1}{h_1^d \cdots h_i^d} \mathbf{K}^i(\frac{x_1}{h_1}, \ldots,\frac{x_i}{h_i})$. In our simulations, we take $h_j=h$ for all $j=1,\ldots,i$ and the standard multiplicative kernel
\begin{equation}\label{kernel_def}
    \mathbf{K}^i(x_1, \ldots, x_i)= \prod_{j=1}^i K(x_j)
\end{equation}
where $K: \R^d\to\R$ is the kernel function
\[
K(x) =  (1-\|x\|^2)^2 \mathbbm{1}_{\|x\|\leq 1} \,.
\]
This choice of kernel function is mainly motivated by considerations of computational convenience, as this simple polynomial form with compact support is less computationally demanding than, for instance, exponential kernels. Finally, in the same way we can derive the estimator for the intensity at times $t_i$, for $i=1, \ldots, N$, starting from \eqref{int_at_grid_time} and deriving the corresponding kernel estimator. This gives fully data-driven estimators that we can compute to simulate the generative process.

\section{Simulation techniques} \label{section_simulation_schemes}
In this section, we develop the main tools to perform the simulation of the optimal generative process \eqref{SDE_under P*}. We first focus on the jump term, which is no longer independent of the Brownian term under $\mathbb{P}^*$, hence it must be approximated iteratively along each simulated trajectory. We present a methodology to compute the instantaneous intensity, then we propose two algorithms for the SBJTS generative model, based on different simulation schemes for the jump-diffusion optimal process \eqref{SDE_under P*}. We refer to \cite{bremaud2020probability, Tankov_jump_proc} for the theory of the simulation of non-homogeneous jump processes.

\subsection{Gaussian jump distribution}
In order to derive tractable estimators of the drift and intensity of jumps for the numerical simulations, we want an easy expression for the function \eqref{funct_Fi}, and so in particular for the density of the increments \eqref{density_increments}. This means that we need a distribution $\nu^0$ whose convolutions admit explicit expressions. We could take for example a Dirac distribution, allowing for jumps with constant amplitude, but this choice is not very flexible. For this reason, in our simulations we work with Gaussian distributions.

Let $\nu^0$ be a multivariate Gaussian measure with mean $c\in\R^d$ and covariance matrix $\Gamma\in\R^{d\times d}$, i.e.
\begin{equation*}
    \nu^0(A) = \int_A \frac{1}{(2\pi)^{\frac{d}{2}} \det(\Gamma)^\frac{1}{2}} \exp{\left( -\frac{1}{2} (z-c)\trans \Gamma^{-1}(z-c) \right)} \d z, \quad A\in \mathcal{B}(\mathbb{R}^d).
\end{equation*}  
Then, in \eqref{decomp_incr} there are all independent Gaussian random variables, so using the properties of convolution we get an explicit formula for the density function \eqref{density_increments} of the increments of the process $X$ under the reference probability measure $\mathbb{P}^0$. Moreover, we make the following assumption:
\begin{itemize}
    \item $\sigma = diag(\sigma_1, \ldots, \sigma_d)$, with $\sigma_p>0$, $p=1,\ldots, d$;
    \item $\Gamma = diag(\gamma_1^2, \ldots, \gamma_d^2)$, with $\gamma_p>0$, $p=1,\ldots,d$.
\end{itemize}
Hence, for $t\in[t_i, t_{i+1})$, $i=0\ldots, N-1$, the expression of the density $f^0_{t_{i+1}-t}$ of the increment $X_{t_{i+1}}-X_{t}$ is 
\begin{equation} \label{density_explicit} 
    f^0_{t_{i+1}-t}(z) = \sum_{k\geq 0} \frac{(\lambda^0 (t_{i+1}-t))^k}{k!} \frac{e^{-\lambda^0 (t_{i+1}-t)}}{(2\pi)^\frac{d}{2} \prod_{p=1}^d \sqrt{\sigma_p^2 (t_{i+1}-t) +k\gamma_p^2} } \exp{\left(-\frac{1}{2} \sum_{p=1}^d \frac{(z_p - k c_{p})^2}{\sigma_p^2(t_{i+1}-t) +k\gamma_p^2}\right)}
\end{equation}
for $z\in\R^d$. For fixed $i=0,\ldots,N-1$, $t\in [t_i, t_{i+1})$, $\mathbf{x}_i \in(\R^d)^i$, $x\in\R^d$, we get an explicit expression for the function $F_i(t,x_i,x,x_{i+1})$ defined in \eqref{funct_Fi}, which allows to compute 
\begin{equation*} 
    a(t,x;\mathbf{x}_i) = \sigma \sigma\trans \frac{\E_\mu \left[\nabla_x F_i(t,x_i,x,X_{t_{i+1}})  \vert \mathbf{X}_{t_i} = \mathbf{x}_{i} \right]}{\E_\mu [F_i(t,x_i,x,X_{t_{i+1}}) \vert \mathbf{X}_{t_i} = \mathbf{x}_{i}]}, 
\end{equation*}
where the computation of $\nabla_x F_i(t,x_i,x,x_{i+1})$ is immediate as the only term to differentiate in the function $F_i(t,x_i,x,x_{i+1})$ is the exponential term at numerator, and for 
\begin{equation*} 
    \Lambda(t,x,z;\mathbf{x}_{i}) = \lambda^0 \,\frac{\E_\mu [ F_i(t,x_i,x+z,X_{t_{i+1}}) \vert \mathbf{X}_{t_i} = \mathbf{x}_{i}]}{\E_\mu [F_i(t,x_i,x,X_{t_{i+1}}) \vert \mathbf{X}_{t_i} = \mathbf{x}_{i}]}.
\end{equation*}
Then using estimators \eqref{drift app} and \eqref{int app}, we approximate the optimal drift $\alpha^*_t = a(t,X_t;\mathbf{X}_{t_i})$ and the optimal intensity $\lambda^*_t(z)= \Lambda(t,X_{t^-},z; \mathbf{X}_{t_i})$. 

In particular, this framework implies that the components of the Brownian motion and of the jump term of the process $X$ under $\P^0$ are assumed to be independent from each other, as we take diagonal matrices for $\sigma$ and $\Gamma$ to remove the correlation terms and have more tractable expressions. Nevertheless, we will show in the numerical simulations that this simplification is not restrictive, as the proposed generative model is still able to effectively capture the correlations among the different components of the datasets that we consider in Section \ref{section_numerical_test}. Furthermore, the marginal distribution of $X_t$ under $\mathbb{P}^0$, which now takes the form of an infinite Poisson mixture of Gaussians $f^0_t(x) = \sum_{k=0}^\infty e^{-\lambda^0 t} \frac{(\lambda^0 t)^k}{k!} \mathcal{N}(x; 0, \, \sigma\sigma\trans t + k \Sigma_J)$, allows the model to better capture fatter tails than in the pure diffusion case, where the marginal distribution is restricted to a single Gaussian $ \mathcal{N}(x; 0, \sigma\sigma\trans t)$.


\begin{remark}[Integrability under Assumption \ref{assumption_1}]
When $\nu^0$ is Gaussian, Assumption \ref{assumption_1} imposes a second-moment constraint on $\mu$ via the finiteness of the relative entropy $H(\mu |\mu^0) < \infty$. Lower-bounding the reference density yields $\mu^0(x) \ge C_0 \exp(-c|x|^2)$ for positive constants $C_0, c > 0$, which implies that $-\log \mu^0(x) \le c|x|^2 - \log C_0$. Decomposing the KL divergence gives:
$$H(\mu | \mu^0) = -h(\mu) - \int_{\mathbb{R}^d} \mu(x) \log \mu^0(x) \, \text{d}x \le -h(\mu) + c \, \mathbb{E}_\mu[|X|^2] - \log C_0$$
where $h(\mu) := -\int \mu(x)\log\mu(x)\text{d}x$ denotes the entropy of $\mu$. Thus, assuming $h(\mu) > -\infty$, Assumption \ref{assumption_1} is satisfied whenever the target distribution $\mu$ has a finite second moment $\mathbb{E}_\mu[|X|^2] < \infty$.
\end{remark}


\begin{remark}[Heavy-Tailed Targets]
In order to 
capture a target distribution $\mu$ with heavy tails, the reference jump distribution $\nu^0$ should be chosen such that its tail heaviness matches or exceeds that of the target. For example, suppose the target distribution $\mu$ has finite moments up to order $\beta > 1$, meaning:
$$\mathbb{E}_\mu[\vert{}X\vert{}^p] < \infty \quad \text{for all } p \le \beta,$$
and suppose we select a reference jump distribution $\nu^0$ (with density $p_{\nu^0}$) such that its convolution with the Gaussian component $g_{\sigma\sqrt{\Delta t_i}}$ has polynomial tail decay of order $\gamma>1$:
$(g_{\sigma\sqrt{\Delta t_i}} * p_{\nu^0})(z)
\sim \frac{c_i}{\vert{}z\vert{}^\gamma}
\, \text{as } \vert{}z\vert{} \to \infty,$
for some constant $c_i>0$.
To satisfy Assumption \ref{assumption_1}, we consider the cross-entropy term, and we use the fact that its negative logarithm grows at most logarithmically:
$-\log f^0_{\Delta t_i}(x)
\lesssim 1+\gamma\log(1+\vert{}x\vert{})$.
We get that a sufficient condition for Assumption \ref{assumption_1} is
$\mathbb{E}_\mu[\log(1+\vert{}X\vert{})]<\infty.$ Because $\mu$ has finite polynomial moments up to order $\beta$, the target trivially possesses a finite logarithmic moment. 
Regarding Assumption \ref{assumption_2}, the polynomial tail yields the upper bound $1/f^0_{\Delta t_i}(x) \le C(1 + \vert{}x\vert{}^\gamma)$ for some constant $C > 0$. Evaluating the expectation under $\mu$, the requirement that the reciprocal transition density be integrable reduces to the polynomial moment condition $\mathbb{E}_\mu[\vert{}X\vert{}^\gamma] < \infty$. By choosing a reference jump tail exponent $\gamma$ such that $1 < \gamma \le \beta$, both assumptions are satisfied simultaneously.
\end{remark}

\vspace{2mm}

In practice, since $\mu$ is the empirical distribution of the initial time series, we do not explicitly evaluate moment conditions or polynomial decays to ensure that the theoretical assumptions hold. We consistently choose $\nu^0$ to be Gaussian to simplify both the analytical expression of the density \eqref{density_increments} and the generation of synthetic time series, while still benefiting from the fatter tails provided by the jump mechanism.

\subsection{Simulation of the jump term}\label{section_jump_term}
Once that we have fixed the distribution $\nu^0$ to be a Gaussian distribution, we focus on the non-homogeneous jump process, whose instantaneous rate at any time $t$ is given by
\begin{equation*}
    L_t := \int_{\R^d} \lambda^*_t(z)\nu^0(\d z), 
\end{equation*}
and the size of the jumps follows the distribution
\begin{equation*}
    \frac{\lambda^*_t(z)\nu^0}{L_t}. 
\end{equation*}
In order to simulate this jump term, for $t\in (t_i, t_{i+1})$, $i=0, \ldots, N-1$, and given the value $X_{t^-} = x$, we approximate $\lambda^*_t(z)$ using the kernel estimator \eqref{int app} and $L_t$ with
\begin{equation} \label{estimator_intensity}
    \hat{L}(t, x; \mathbf{x}_{i}) := \int_{\mathbb{R}^d} \hat{\Lambda}(t, x, z; \mathbf{x}_{i}) \nu^0(\d z), 
\end{equation}
which for example can be computed using the Gauss-Hermite quadrature formula. Then we can sample the jump amplitude according to the distribution $\frac{\hat{\Lambda}(t, x, z; \mathbf{x}_{i}) \nu^0}{\hat{L}(t, x; \mathbf{x}_{i})}$ with different numerical methods. 

However, in our simulations we consider the following strategy to rewrite the term $\lambda^*_t(z)\nu^0$, which enables both a more efficient computation of $\hat{L}(t,x; \mathbf{x}_{i})$ and an easy way to sample the jump amplitude. Denoting by $g(z)$ the density of the Gaussian measure $\nu^0$ with respect to the Lebesgue measure, and using the estimator \eqref{int app}, we have
\begin{multline}\label{Gaussian_sum}
    \hat{\Lambda}(t, x, z; \mathbf{x}_{i}) g(z) = \frac{\lambda^0}{D} \sum_{m=1}^M \sum_{j\geq 0} \frac{w_{j,m} }{(2\pi)^{\frac{d}{2}} \prod_{p=1}^d \sqrt{\sigma_p^2(t_{i+1}-t) +j\gamma_p^2}} \exp{ \left(-\frac{1}{2} \sum_{p=1}^d \frac{(z_p-(X^{(m)}_{t_{i+1},p}-x_p- jc_p))^2}{\sigma_p^2 (t_{i+1}-t) +j\gamma_p^2}\right)}\\
    \cdot \frac{1}{(2\pi)^{\frac{d}{2}} \prod_{p=1}^d \gamma_p} \exp{\left(-\frac{1}{2} \sum_{p=1}^d \frac{(z_p-c_p)^2}{\gamma_p^2}\right)}
\end{multline} 
where $D$ is the denominator in the formula \eqref{int app}, and 
\begin{equation} \label{weights_mixture}
    w_{j,m} = \frac{\frac{(\lambda^0 (t_{i+1}-t))^j}{j!} e^{-\lambda^0 (t_{i+1}-t)} \, \mathbf{K}^i_{\mathbf{h}_i}(\mathbf{x}_{i}-\mathbf{X}^{(m)}_{t_i})} {\displaystyle\sum_{k\geq 0} \frac{(\lambda^0 (t_{i+1}-t_{i}))^k}{k!} \frac{e^{-\lambda^0 (t_{i+1}-t_{i})}}{(2\pi)^{\frac{d}{2}} \prod_{p=1}^d \sqrt{\sigma_p^2(t_{i+1}-t_i) +k\gamma_p^2}} \exp{\left(-\frac{1}{2} \sum_{p=1}^d \frac{(X^{(m)}_{t_{i+1},p}-x_{i,p}-kc_p)^2}{\sigma_p ^2 (t_{i+1}-t_{i}) +k \gamma_p^2}\right)}} 
\end{equation}
Using the properties of Gaussian densities, we can rewrite each product of the sum in \eqref{Gaussian_sum} to get 
\begin{equation} \label{Gaussian_mixture_model}
    \hat{\Lambda}(t, x, z; \mathbf{x}_{i}) g(z) = \frac{\lambda^0}{D} \sum_{m=1}^M \sum_{j\geq 0} w_{j,m} \, C_{j,m} \, g_{j,m}(z), 
\end{equation} 
where we denote by $g_{j,m}$ the density of the multivariate Gaussian 
\begin{equation*}
    \mathcal{N}_{j,m}:= \mathcal{N}\left( \xi_{j,m}, \Sigma_{j} \right), 
\end{equation*}
where $\xi_{j,m}\in \mathbb{R}^d$, $\Sigma_{j}\in \mathbb{R}^d \times \mathbb{R}^d$ are defined by
\begin{align*}
    \xi_{j,m} &= \left(\frac{(X^{(m)}_{t_{i+1}, p}-x_p-jc_{p})\gamma^2_p + c_{p}(\sigma_p^2 (t_{i+1}-t) +j\gamma^2_p)}{\sigma_p^2 (t_{i+1}-t) +(j+1)\gamma^2_p }\right)_{p=1, \ldots, d}\\
    \Sigma_{j} &= diag \left(\frac{\gamma^2_p(\sigma_p^2 (t_{i+1}-t) +j\gamma^2_p)}{\sigma_p^2 (t_{i+1}-t) +(j+1)\gamma^2_p}\right)_{p=1, \ldots, d}
\end{align*}
and by $C_{j,m}$ the normalizing constant
\begin{equation*}
    C_{j,m} = \prod_{p=1}^d \frac{1}{\sqrt{2\pi (\sigma_p^2(t_{i+1}-t) +(j+1)\gamma_p^2)}} \exp{\left( -\frac{(X^{(m)}_{t_{i+1},p}-x_p-(j+1)c_p)^2}{2(\sigma_p^2(t_{i+1}-t)+(j+1)\gamma_p^2)}  \right)},
\end{equation*}
for $m=1,\ldots,M$, $j\geq0$. Therefore we get the Gaussian mixture \eqref{Gaussian_mixture_model}. Substituting this mixture into \eqref{estimator_intensity}, the evaluation of the integral
\begin{equation*}
    \hat{L} (t, x; \mathbf{x}_{i}) = \int_{\R^d} \frac{\lambda^0}{D} \sum_{m=1}^M \sum_{j\geq 0} w_{j,m} \, C_{j,m} \, g_{j,m}(z) \d z  
\end{equation*}
leads to 
\begin{equation}\label{est_int_general_case}
    \hat{L} (t, x; \mathbf{x}_{i}) = \frac{\lambda^0}{D} \sum_{m=1}^M \sum_{j\geq 0} w_{j,m} \, C_{j,m}.
\end{equation}
Also the simulation of the jump size $J\sim \frac{\hat{\Lambda}(t, x, z; \mathbf{x}_{i}) \nu^0}{\hat{L}(t, x; \mathbf{x}_{i})}$ is now straightforward, as we choose a pair of indices $(j,m)$ according to the probabilities $\frac{w_{j,m} C_{j,m} }{\sum_{j,m} w_{j,m}C_{j,m}}$, and generate the random variable $J$ from the corresponding multivariate Gaussian distribution $\mathcal{N}_{j,m}$. Notice that this methodology yields a reduction in computational complexity while maintaining the accuracy, as it allows to avoid numerical quadrature in \eqref{estimator_intensity} to sample efficiently the jump sizes, hence we use it in all the numerical simulations.

\begin{remark}
Whenever numerical evaluation of the estimators is required in the simulations, the density of the increments \eqref{density_explicit} is approximated by truncating the infinite sum to a finite number of jumps $n_J\in \mathbb{N}$, i.e.\ we consider
\begin{equation*} 
    \hat{f^0}_{t_{i+1}-t}(z) = \sum_{k\geq 0}^{n_J} \frac{(\lambda^0 (t_{i+1}-t))^k}{k!} \frac{e^{-\lambda^0 (t_{i+1}-t)}}{(2\pi)^\frac{d}{2} \prod_{p=1}^d \sqrt{\sigma_p^2 (t_{i+1}-t) +k\gamma_p^2} } \exp{\left(-\frac{1}{2} \sum_{p=1}^d \frac{(z_p - k c_{p})^2}{\sigma_p^2(t_{i+1}-t) +k\gamma_p^2}\right)}.
\end{equation*}
Then we can derive the truncated versions of \eqref{drift app} and \eqref{est_int_general_case}, which are the final expressions used in our numerical simulations.
\end{remark}

\subsection{Euler scheme with Gaussian jumps}
\begin{algorithm}[t]
\caption{SBJTS simulation with Euler scheme}
\label{alg2}
\begin{algorithmic}[1]
\State \textbf{Input:} data samples of time series $(X^{(m)}_{t_1},\ldots,X^{(m)}_{t_N})$, $m=1,\ldots,M$
\State \textbf{Initialization:} initial state $x_0=0$
\For{$i = 0, \ldots, N-1$}
    \State Initialize state $x_{i,0}=x_i$
    \For{$k = 0, \ldots, N_\pi - 1$}
        \State Compute $\hat{a}(t_{i,k}, x_{i,k}; \mathbf{x}_{i})$ by kernel estimator \eqref{drift app}
        \State Sample $\varepsilon_k \sim \mathcal{N}(0,I_d)$
        \State Compute $\hat{L}(t_{i,k}, x_{i,k}; \mathbf{x}_{i})$ by kernel estimator \eqref{est_int_general_case} 
        \State Generate $\bar{N}_{i,k}\sim \text{Poi}((t_{i,k+1}-t_{i,k})\hat{L}(t_{i,k}, x_{i,k}; \mathbf{x}_{i}))$ 
        \State Generate $J_n \sim \frac{\hat{\Lambda}(t_{i,k}, x_{i,k}, z; \mathbf{x}_{i}) \nu^0(\d z)}{\hat{L}(t_{i,k}, x_{i,k}; \mathbf{x}_{i})}$, $n=1,\ldots,\bar{N}_{i,k}$
        \State Compute
        \begin{equation*}
            x_{i,k+1} = x_{i,k} + \frac{\Delta t_i}{N_\pi} \hat{a}(t_{i,k}, x_{i,k}; \mathbf{x}_{i}) + \sigma\sqrt{\frac{\Delta t_i}{N_\pi}} \varepsilon_k +\sum_{n= 1}^{\bar{N}_{i,k}} J_n
        \end{equation*}
    \EndFor  
    \State Set $x_{i+1}=x_{i,N_\pi}$
\EndFor
\State \textbf{Return:} ($x_1,\ldots,x_N$)
\end{algorithmic}
\end{algorithm}
The first simulation scheme that we propose is a classical Euler scheme for the simulation of the optimal process \eqref{SDE_under P*}: more details for the Euler scheme for jump-diffusion processes can be found for example in \cite{glasserman2004convergence, protter1997euler}. Given the set of observation times $\{t_0=0, t_1,\ldots,t_N\}$, we fix a discretization $\pi$ inside each interval $[t_i,t_{i+1}]$, $i=0,\ldots,N-1$, given by $t_{i,k}=t_i + \frac{k}{N_\pi} \Delta t_i$, for $k=0,\ldots,N_\pi-1$, where $\Delta t_i = t_{i+1}-t_i$ and $N_\pi$ is the number of uniform time steps. From equation \eqref{SDE_under P*} we have
\begin{equation*}
    X_{t_{i,k+1}}-X_{t_{i,k}} = \int_{t_{i,k}}^{t_{i,k+1}} \alpha^*_t \d t +\int_{t_{i,k}}^{t_{i,k+1}} \sigma \d W_t + \int_{t_{i,k}}^{t_{i,k+1}} \int_{\R} z N(\d t,\d z)
\end{equation*}
for $i=0,\ldots,N-1$ and $k=0,\ldots,N_\pi-1$. We denote by $x_{i,k}$ the Euler approximation of $X_{t_{i,k}}$, and by $\mathbf{x}_{i}= (x_{1},\ldots, x_{i})$ the values attained at the grid times $(t_1,\ldots,t_i)$. Notice that the approximation of $X_{t_{i,k+1}}$ is performed using the information available at $t_{i,k}$, in particular the value $x_{i,k}$ of $X_{t_{i,k}}$. In our setting this means that both the drift and the jump intensity are evaluated at $(t_{i,k}, x_{i,k})$ and assumed to be constant in the whole time interval $(t_{i,k},t_{i,k+1})$. The drift $\alpha^*_t$ is approximated by $\hat{a}(t_{i,k}, x_{i,k}; \mathbf{x}_{i})$ given by the kernel estimator \eqref{drift app}, while for the jump component the instantaneous rate $\hat{L}(t, x; \mathbf{x}_{i})$ is approximated by $\hat{L}(t_{i,k}, x_{i,k}; \mathbf{x}_{i})$ via \eqref{est_int_general_case}. This results in a constant intensity over each Euler time step, allowing to simulate the number of jumps in $(t_{i,k},t_{i,k+1})$ as
\begin{equation*}
    \bar{N}_{i,k} \sim \text{Poi}((t_{i,k+1}-t_{i,k})\hat{L}(t_{i,k}, x_{i,k}; \mathbf{x}_{i})).
\end{equation*}
Finally, for the size of the jumps, we simulate i.i.d.\ random variables $(J_n)_{n=1,\ldots,\bar{N}_{i,k}}$ with distribution
\begin{equation} \label{simulation}
    J_n \sim \frac{\hat{\Lambda}(t_{i,k}, x_{i,k}, z; \mathbf{x}_{i}) \nu^0}{\hat{L}(t_{i,k}, x_{i,k}; \mathbf{x}_{i})}, 
\end{equation}
using the model \eqref{Gaussian_mixture_model} with Gaussian densities $\mathcal{N}_{j,m}= \mathcal{N}\left( \xi_{j,m}, \Sigma_{j} \right)$ defined by 
\begin{align*}
    \xi_{j,m} &= \left(\frac{(X^{(m)}_{t_{i+1}, p}-x_{i,k,p}-jc_{p})\gamma^2_p + c_{p}(\sigma_p^2 (t_{i+1}-t_{i,k}) +j\gamma^2_p)}{\sigma_p^2 (t_{i+1}-t_{i,k}) +(j+1)\gamma^2_p }\right)_{p=1, \ldots, d}\\
    \Sigma_{j} &= diag \left(\frac{\gamma^2_p(\sigma_p^2 (t_{i+1}-t_{i,k}) +j\gamma^2_p)}{\sigma_p^2 (t_{i+1}-t_{i,k}) +(j+1)\gamma^2_p}\right)_{p=1, \ldots, d}
\end{align*}
where $x_{i,k,p}$ denotes the $p$-th component of $x_{i,k}\in \R^d$. The complete Euler scheme of the process $X$ is then given by
\begin{equation*}
    x_{i,k+1} = x_{i,k} + (t_{i,k+1}-t_{i,k})\, \hat{a}(t_{i,k}, x_{i,k}; \mathbf{x}_{i})+ \sigma (W_{t_{i,k+1}}-W_{t_{i,k}}) + \sum_{n=1}^{\bar{N}_{i,k}} J_n
\end{equation*}
The pseudo-code of the complete generative model is presented in Algorithm \ref{alg2}.

\subsection{Jump-adapted version of the Euler scheme}
We propose an alternative simulation method of the process \eqref{SDE_under P*} based on the so called jump-adapted Euler scheme for jump-diffusion SDEs (for details see \cite{bruti2007strong}). We consider the jump-adapted time discretization $\{\tau_0=0,\,\tau_1,\ldots,\tau_n=T\}$, which is constructed taking the equidistant time discretization $\pi$ given by $\{t_{i,k}\}_{i=0,\ldots,N-1, \, k=0\ldots,N_\pi-1}$ and adding the jump times given by the marked point process. In this way, we can simulate only the diffusion part of the process $X$ when there are no jumps, and add the jump increment only at the jump times. In particular, we use the following scheme:
\begin{align}
    \begin{cases}
        X_{\tau_{n+1}^-} = X_{\tau_{n}} + \alpha^*_{\tau_n} (\tau_{n+1}-\tau_n) + \sigma(W_{\tau_{n+1}}-W_{\tau_{n}}) \\
        X_{\tau_{n+1}} = X_{\tau_{n+1}^-} + J &\text{ if } \tau_{n+1} \text{ is a jump time} \\
        X_{\tau_{n+1}} = X_{\tau_{n+1}^-} &\text{ if } \tau_{n+1} \text{ is not a jump time}
    \end{cases}
\end{align}
where $J$ denotes the jump size at each jump time. Once again, the simulation of the jump times is performed during the simulation of the process as we work with time, state and path dependent intensity. Classical methods for simulating this type of jump-diffusion processes are based on the thinning algorithm (by Ogata \cite{ogata1981lewis, glasserman2004convergence}), which consists in taking an upper bound $M$ of the intensity function, simulating the jump times of a Poisson process with constant rate $M$, and then perform an acceptance-rejection step to select only the jump times of the time inhomogeneous Poisson process. Improvements of the Ogata's formulation in the case of a time and state dependent intensity allow to take a local bound which is an upper bound in space, but still a function of time (see \cite{daley2003introduction}). 
\begin{algorithm}[t]
\caption{SBJTS simulation with jump-adapted Euler scheme}
\label{alg3}
\begin{algorithmic}[1]
\State \textbf{Input:} data samples of time series $(X^{(m)}_{1},\ldots,X^{(m)}_{N})$, $m=1,\ldots,M$
\State \textbf{Initialization:} initial $x_0=0$, $t_{\text{current}}=0$
\State Generate the first jump time $t_{\text{jump}} \sim \text{Exp}( \hat{L}(t_{\text{current}}, x_0; \mathbf{x}_{0}))$
\For{$i = 0, \ldots, N-1$}
    \State Initialize state $x=x_i$
    \If{$t_{\text{jump}}=t_i$}
        \State $t_{\text{jump}} \sim t_{i} + \text{Exp} (\hat{L}(t_{i}, x_i; \mathbf{x}_{i}))$
    \EndIf
    \For{$k = 0, \ldots, N_\pi - 1$}
        \While{$t_{\text{jump}} \le t_{i,k+1}$ \textbf{and} $t_i < t_{\text{jump}} < t_{i+1}$}
            \State $\Delta t = t_{\text{jump}} - t_{\text{current}}$
            \State Compute $\hat{a}(t_{\text{current}}, x; \mathbf{x}_{i})$ by kernel estimator \eqref{drift app} and sample $\varepsilon \sim \mathcal{N}(0,I_d)$
            \State $x \leftarrow x + \Delta t\, \hat{a}(t_{\text{current}}, x; \mathbf{x}_{i}) + \sigma\sqrt{\Delta t} \,   \varepsilon $
            \State $t_{\text{current}} = t_{\text{jump}}$
            \State Generate $J \sim \hat{\Lambda}(t_{\text{current}}, x, z; \mathbf{x}_{i}) \nu^0(\d z)$
            \State Compute $\hat{L}(t_{\text{current}}, x; \mathbf{x}_{i})$ by kernel estimator \eqref{est_int_general_case} 
            \State Generate $t_{\text{jump}} \sim t_{\text{current}} + \text{Exp} (\hat{L}(t_{\text{current}}, x; \mathbf{x}_{i}))$
            \State $x \leftarrow x + J$
        \EndWhile
        \State $\Delta t = t_{i,k+1} - t_{\text{current}}$
        \State Compute $\hat{a}(t_{\text{current}}, x; \mathbf{x}_{i})$ by kernel estimator \eqref{drift app} and sample $\varepsilon \sim \mathcal{N}(0,I_d)$
        \State $x \leftarrow x + \Delta t\, \hat{a}(t_{\text{current}}, x; \mathbf{x}_{i}) + \sigma\sqrt{\Delta t} \,   \varepsilon $
        \State $t_{\text{current}} = t_{i,k+1}$
    \EndFor
    \State Set $x_{i+1}=x$
\EndFor
\State \textbf{Return:} $(x_{1}, \ldots, x_{N})$ 
\end{algorithmic}
\end{algorithm}

However, the main difficulty in our case lies in the fact that it is not straightforward to determine a uniform upper bound in both time and space for the function \eqref{est_int_general_case}. Moreover, working with state-dependent intensity, before performing the acceptance-rejection step we should simulate the continuous part of the process (which in particular involves the drift estimation) until each candidate jump time: to reduce at most the rejections, which generates additional computational time, we should find a strict upper bound of the intensity, which makes the problem still more complicated. For these reasons, we choose an approach that differs from the thinning method, and it is based on an approximation of the intensity function. Indeed, at each jump time $\tau_n$ we simulate the following jump time $\tau_{n+1}$ taking as time interval an exponential random variable with rate $\hat{L}$ evaluated at $\tau_n$. This means that we approximate the intensity function by a piecewise constant function, evaluating it only at the jump times. Consequently, the simulation of the jump component in the continuous-time generative process $X$ inevitably introduces a simulation error. However, rather than analyzing the error at the level of the SDE, we evaluate the performance of the generative model directly in terms of its ability to produce realistic synthetic time series, looking at specific metrics. In addition, this simulation strategy substantially reduces the computational cost of the generative model compared with the classical Euler scheme employed in Algorithm \ref{alg2}, making it particularly suitable for repeated generation tasks. We present the complete jump-adapted scheme in Algorithm \ref{alg3}.

This jump-adapted method has an average number of operations depending on the intensity of jumps as the number of points in the time discretization increases with the intensity. This means that it is not  efficient in the case of a large intensity, whereas the standard Euler scheme is not affected. But in the simulations that we consider in the following sections, we work with an average small number of jumps, making Algorithm \ref{alg3} much faster than Algorithm \ref{alg2}. Moreover, in this case we manage to always guarantee the predictability of the intensity function $(t,z)\mapsto \lambda^*_t(z)$, as the estimator $\hat{L}(t, x; \mathbf{x}_{i})$ is computed before updating the Euler scheme with the jump occurring at time $t$. Finally, we assume that no jumps occur exactly at the observation times $\{t_i\}_{i=1, \ldots, N}$ as jumps are generated only within each interval $(t_{i}, t_{i+1})$. This is consistent with the classical Euler scheme, where it is natural to treat jumps as occurring strictly between discretization points.

\section{Hyperparameter tuning}\label{section_calibration}
We outline the methodology for selecting appropriate values for the hyperparameters of our generative model. We describe a systematic approach that balances model performance and computational efficiency, ensuring that the chosen hyperparameters lead to accurate generative behavior. We present our calibration procedure in the one-dimensional setting. Specifically, we consider the optimal process \eqref{SDE_under P*} taking values in $\mathbb{R}$, with hyperparameters $\sigma>0$, $\nu^0=\mathcal{N}(c, \gamma^2)$, where $c\in \mathbb{R}$, $\gamma>0$, and $\lambda^0>0$. Although our analysis is carried out in dimension one, the same procedure will be applied componentwise in the multidimensional case (see Section \ref{section_real}).

\subsection{Selection of the kernel bandwidth and Markovianity order} \label{section_MBtest}
At each time, the kernel estimators presented above require to take into account the whole past values of the time series in order to compute the drift and intensity. In the presence of long time series, this results in an increasing number of factors $\mathbf{K}^i_{\mathbf{h}_i}(\mathbf{x}_{i}-\mathbf{X}^{(m)}_{t_i})$ equal to zero as we look at further dates. In order to solve this problem, we can force a shorter memory in the estimators, asking that the drift and intensity depend on $k<N$ past values. We follow the approach of \cite{alouadi2025robust}, where the authors propose a test to fix the window of past values dependence, also called Markovianity order. In addition, this test allows to calibrate also the bandwidth $h$ of the kernel \eqref{kernel_def}: this parameter has a direct impact on the amount of data that are involved in the kernel estimation, hence it is important to perform a bias-variance tradeoff to achieve a good performance in the generation of time series. 

Following the cross-validation test proposed in \cite{alouadi2025robust}, fix a train set $X = (X^{(m)}_{t_1},\ldots, X^{(m)}_{t_N})_{m=1,\ldots,M}$ and a test set $Y = (Y^{(q)}_{t_1},\ldots, Y^{(q)}_{t_N})_{q=1,\ldots,Q}$ of real data. For each $q$, we take the time series $Y^{(q)}$ of the test set at the first $N-1$ dates $(Y^{(q)}_{t_1},\ldots, Y^{(q)}_{t_{N-1}})$ and we generate $L$ realizations of the last value $\hat{Y}^{(q),l}_{t_{N}}$ using the Schrödinger bridge generative model and the train set $X$. We do this for each couple $(h,k)$ varying in a grid of values, and we compute 
\begin{equation}\label{Markovianity_bandwidth_test}
    MSE_{h,k} = \frac{1}{Q} \sum_{q=1}^Q \left|\frac{1}{L}\sum_{l=1}^L \hat{Y}^{(q),l}_{t_{N}} - Y^{(q)}_{t_{N}} \right|^2
\end{equation}
In this way we can select the optimal $(h^*,k^*)$ which gives the minimum value of the previous mean square error function. In practice, at time $t_i$, with $i>k^*$, we will work with the kernel product $\bar{K}^{(m)}_i = \prod_{j=i-k^*+1}^i K_{h^*}(x_j-X^{(m)}_{t_{j}})$. In the following numerical simulations, we use this test when we generate synthetic time series with both the SBTS algorithm \cite{pham_generative} and our SBJTS algorithm.

\subsection{Selection of the parameters $\sigma$, $\lambda^0$, $c$ and $\gamma$} \label{hyper_selection}
We propose two possible tests that we can use to determine appropriate values of $\sigma$, $\lambda^0$, $c$ and $\gamma$. Each method is implemented after generating synthetic time series: for a given set of parameters, synthetic series are produced, a suitable loss function or metric is computed, and the results are compared with the corresponding statistics of the original data to identify the parameters that achieve the best performance. This procedure implies that we need to fix a grid of pre-specified parameter values over which the test is conducted. We proceed with the following two tests:

\begin{itemize}
    \item \textbf{Test on the distribution of the quadratic variation}: we consider the empirical distribution of the quadratic variation of the time series, computed summing the squared increments for each path, both for the initial time series (giving the distribution $p_\text{orig}$) and for the synthetic time series ($p_\text{synt}$). We aim to minimize the Wasserstein-2 distance of these two distributions numerically estimated:
    \begin{equation*} \label{qv_metric}
        \mathcal{W}_2(p_\text{orig}, p_\text{synt}) = \left( \inf_{\pi \in \Pi (p_\text{orig}, p_\text{synt})} \int |x-y|^2 d\pi (x,y) \right)^\frac{1}{2}
    \end{equation*}
    For each set of parameter values, we select the one that minimizes this loss function. 
    \item \textbf{Test on the discriminative score}: we consider the metric of the discriminative score to assess how distinguishable generated time series are from real ones. It works by training a recurrent neural network discriminator that learns to classify sequences as real or synthetic. Then the final score is defined as the absolute difference between 0.5 and the accuracy of classification. If the discriminator can easily separate the two classes, its accuracy is close to 1, and the discriminative score is close to 0.5. Conversely, if the generated data closely mimics the real distribution, the discriminator performs no better than random guessing ($\approx 50\%$ accuracy), yielding a score near zero. Thus, the test provides a quantitative way to evaluate the realism of generated time series. For each set of parameter values, we select the one that minimizes the score.
\end{itemize}
In the following section we explain how we can combine these two tests to choose the parameters. 

\subsection{Calibration procedure} \label{procedure_calib}
Our first step concerns the choice of the parameter $c$, the mean of the Gaussian distribution $\nu^0$. In all our numerical experiments, we set $c=0$: this is motivated by the fact that the considered datasets do not exhibit a pronounced upward or downward trend, and alternative values of $c$ did not appear necessary to improve the performance of the generative model. To determine suitable values for the other parameters, we first look at the quadratic variation of the time series which we aim to simulate under the measure $\mathbb{P}^*$. The quadratic variation is a natural quantity to consider, as it captures an important aspect of the variability of the time series and is closely related to the volatility of the underlying process. Moreover, it is straightforward to estimate from the data and provides an informative statistic of the path behavior. Hence, in order to calibrate the model parameters, we compare the empirical quadratic variation of the original data with that of the generated samples and seek parameter values that reproduce this feature as accurately as possible. However, as under $\mathbb{P}^*$ we have all expressions depending on the estimators, it is more practical to consider first the optimal process under $\mathbb{P}^0$ to obtain a preliminary selection of appropriate parameter values for $\sigma$, $\lambda^0$, and $\gamma$. In particular, we look at the following expression, for $i=0, \dots, N-1$, 
\begin{equation} \label{var_relation}
    \text{Var}_\mu(X_{t_{i+1}}-X_{t_{i}}) = \Delta t_i (\sigma^2 + \lambda^0\gamma^2)
\end{equation}
where on the left hand side we have the empirical variance under the measure $\mu$ of the increments of the time series data, and on the right hand side the theoretical variance of the increment of the optimal process under the reference probability $\mathbb{P}^0$ in the time interval $[t_i,t_{i+1})$ of length $\Delta t_i=t_{i+1}-t_i$. In practice, in our setting all time increments have a fixed length $\Delta t$, and on the left hand side we consider the average of the empirical variances computed over these intervals. In this way, we can select a reference measure $\mathbb{P}^0$ under which the increments of the optimal process have already a variance close to the one of the initial data. Moreover, we can easily find some upper bounds for $\sigma$ and $\gamma$, using \eqref{var_relation}, and excluding values which cause instability of the algorithm. For each set of parameters satisfying condition \eqref{var_relation}, we then generate synthetic time series using the SBJTS generative model. 

However, considering only the variance under $\mathbb{P}^0$ does not ensure that the variance of the increments of the time series generated under $\mathbb{P}^*$ is close to the one of the data. To address this limitation, we can calibrate better the parameters directly under the measure $\mathbb{P}^*$. For each couple $(\sigma, \gamma)$ fixed, we tune the parameter $\lambda^0$ to minimize the Wasserstein-2 distance between the empirical distribution of the quadratic variation of the simulated time series under $\mathbb{P}^*$ and the one of the initial data. This refinement results in a more effective calibration procedure for the generative model. Finally, among the candidate sets of parameters, we choose the set minimizing the discriminative score computed between the real and the synthetic time series. We describe the calibration procedure in detail below. \\

\noindent \textbf{Calibration procedure for the parameters $h$, $k$, $\sigma$, $\lambda^0$ and $\gamma$:}
\begin{itemize}
    \item Fix initial kernel bandwidth $h$ and Markovianity order $k$: choose initial values for $h$ and $k$, allowing for a reasonable bias-variance tradeoff in the kernel estimators. 
    \item Preliminary parameter selection under $\mathbb{P}^0$: using relation \eqref{var_relation}, determine suitable values of $(\sigma, \gamma)$ which do not cause instability of the algorithm. 
    \item Parameter selection under $\mathbb{P}^*$: for each couple ($\sigma$, $\gamma$) determined at the previous point, tune the value of $\lambda^0$ to minimize the Wasserstein-2 distance of the empirical distributions of the time series quadratic variation.
    \item Minimization of the discriminative score: for each candidate $(\sigma, \gamma, \lambda^0)$,  generate synthetic time series and select the triplet that minimizes the discriminative score.    
    \item Validate the initial parameters: run again the Markovianity-bandwidth test with the selected $(\sigma, \gamma, \lambda^0)$ to confirm that the initial choice of $h$ and $k$ remains appropriate.
    \item Final generation: use the finalized parameter set to perform the generation of synthetic time series. 
\end{itemize}

\begin{remark}
We emphasize that, in theory, the SBJTS generative model can operate for any choice of parameters, as no specific restrictions are imposed on their values. However, in practice, the model’s performance varies significantly across different parameter configurations. It is therefore important to devote attention to the preliminary calibration phase in order to identify suitable parameters. This aspect becomes particularly crucial when dealing with time series derived from real datasets rather than those simulated from parametric models. Nevertheless, this calibration step is not intended to determine the optimal set of parameters, as it is always performed as a grid search on some fixed candidates, but rather to identify reasonable configurations that ensure good performance. 
\end{remark}

\section{Numerical tests}\label{section_numerical_test}
\subsection{Test on simulated data: Merton model}\label{section_merton}
\begin{figure}[t]
    \centering
    \includegraphics[width=0.7\textwidth]{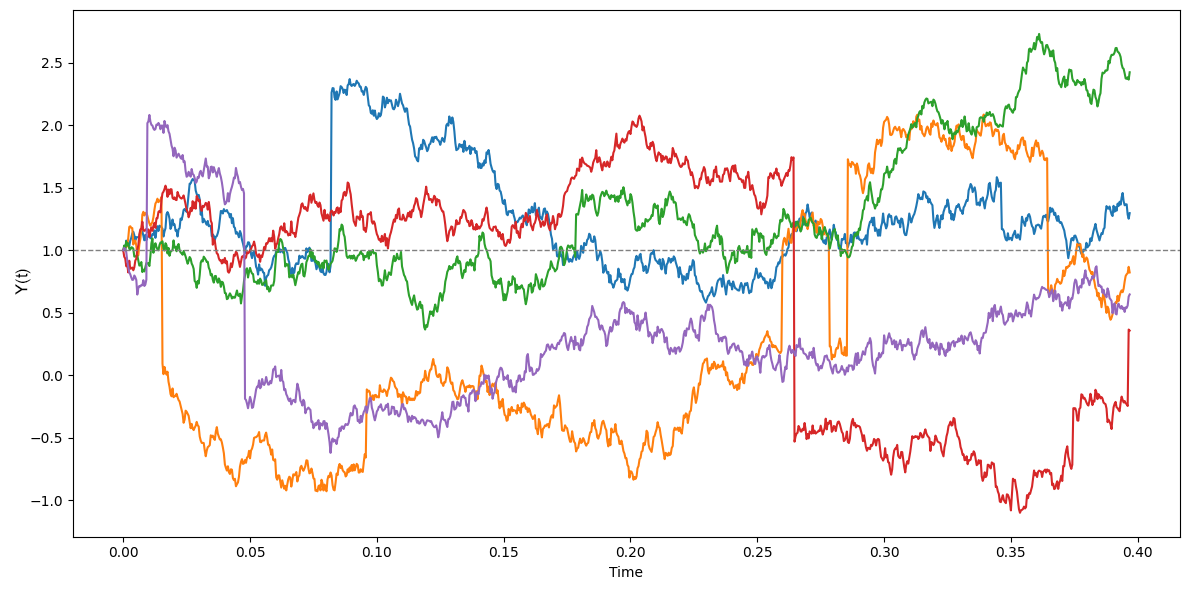}
    \caption{Simulated trajectories of Merton process \eqref{Merton_sde} with parameters $a=0$, $b=2$, $\lambda_\eta =10$, $m_J=0$ and $v_J =0.8$.}
    \label{Merton_sample_paths}
\end{figure}
As a preliminary illustrative example, we propose to test our generative model starting from time series obtained by the simulation of trajectories of a Merton-type process in dimension 1. Let $(\eta_t)_{t \ge 0}$ be a Poisson process of constant intensity $\lambda_\eta$, and denote by $\{{\tau_i}\}_{i \ge 1}$ its jump times. Consider the process defined by
\begin{equation} \label{Merton_sde}
    Y_t = Y_0 + a \,t + b\, W_t + \sum_{i=1}^{\eta_t} J_i(Y_{\tau_i^-}),
\quad t \in [0,T].
\end{equation}
with $Y_0 = 1$, where $a$ and $b$ are fixed parameters and $W=(W_t)_{t\geq 0}$ is a standard Brownian motion. The jump sizes are sampled with the following rule:
\begin{align*}
    J_i(Y_{\tau_i^-}) = 
    \begin{cases}
        |Z_i| &\text{if } Y_{\tau_i^-} < Y_0 \\
        -|Z_i| &\text{if } Y_{\tau_i^-} \geq Y_0
    \end{cases}
\end{align*}
where $Z_i \sim \mathcal{N}(m_J, v_J^2)$. This state-dependent sign rule induces a mean-reverting effect around the level $Y_0$, as upward (resp. downward) jumps are sampled when the process is below (resp. above) the initial value. We adopt this mechanism to have an initial dataset which allows numerical stability of the algorithm, without the need for additional processing or correction steps. We present the case where the volatility $b$ is moderate and the variance of the jump term permits the occurrence of relatively large jumps. Specifically, we fix the parameters $a=0$, $b=2$, $\lambda_\eta =10$, $m_J=0$ and $v_J =0.8$. Figure \ref{Merton_sample_paths} displays the sample paths obtained from five simulated realizations of the process. This configuration is intended to represent a realistic scenario in which the average jump frequency remains relatively low. We simulate $M=1000$ sample paths and then discretize the trajectories on the uniform grid $t_1, \ldots, t_N$ with $t_{i+1}-t_i=\frac{1}{252}$, for $i=0,\ldots,N-1$, $T=t_N$ and $N=100$; to run the simulation scheme, we discretize again each interval $[t_i, t_{i+1}]$ with $N_\pi=100$ steps. We test both the SBTS model (taken from \cite{pham_generative}) and our SBJTS model using Algorithm \ref{alg3}, to compare the resulting metrics on 500 synthetic time series.

\subsubsection{For comparison: SBTS generation using the algorithm from [Hamdouche, Henry-Labordère and Pham, 2023]}
\begin{figure}[t]
    \centering
    \includegraphics[width=0.7\textwidth]{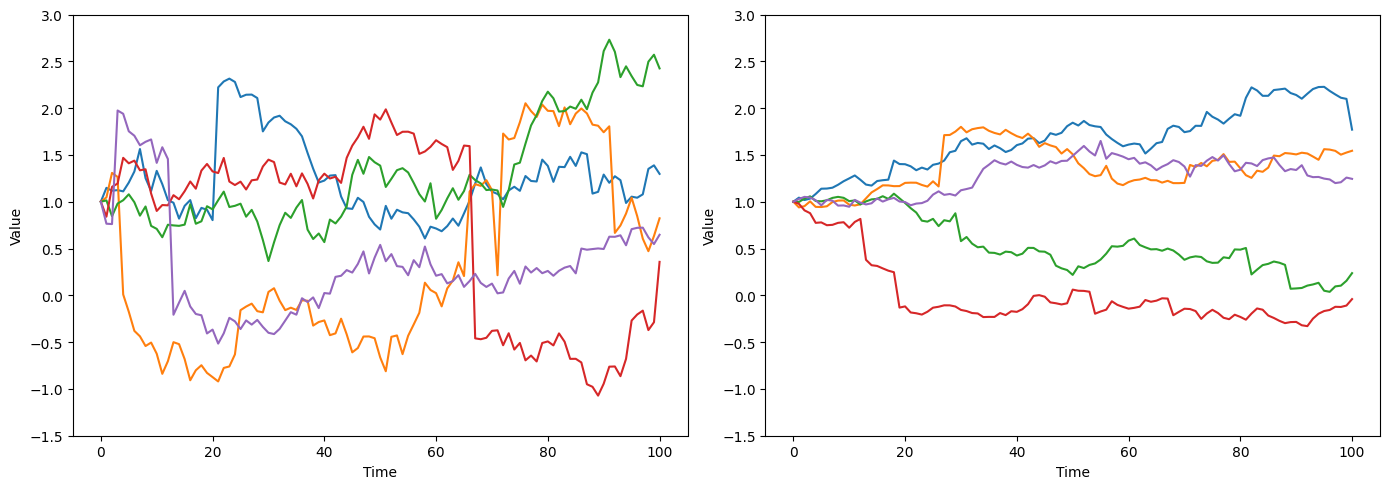}
    \caption{Generation of synthetic time series by SBTS model: sample paths of real time series (left) and sample paths of synthetic time series with $h=0.1$ and $k=1$ (right).}
    \label{cont_paths_3}
\end{figure}

To evaluate the performance of the SBTS algorithm on Merton-generated time series, we adopt the procedure described in \cite{alouadi2025robust, pham_generative}. Starting from the original time series $(X_{t_i})_{i=0, \ldots, N}$, we compute the increments $R_{t_i} = X_{t_{i}} - X_{t_{i-1}}$ for $i=1, \ldots, N$, and apply the rescaling
\begin{equation}\label{normalization_std}
    \tilde{R}_{t_1:t_N} = R_{t_1:t_N} \cdot \frac{\sqrt{\Delta t}}{\sigma(R_{t_1:t_N})}
\end{equation}
where $\Delta t=\frac{1}{252}$ is the interval between two consecutive dates, $R_{t_1:t_N}=(R_{t_1}, \ldots, R_{t_N})$ denotes the entire series of increments, and $\sigma(R_{t_1:t_N})$ their empirical standard deviation. This normalization ensures that the rescaled time series have increments with empirical variance approximately equal to $\Delta t$. We consider the diffusion process defined by 
\begin{equation*}
    \d X_t = \alpha^*_t \d t + \d W_t, \quad t\in [0,T],
\end{equation*}
where the drift $\alpha^*_t$ is defined as in \eqref{def_of_drift} and \eqref{a general}, but in this case the expression of \eqref{funct_Fi} involves only the Gaussian densities given by the increments of the Brownian motion, without Poisson terms. To generate synthetic time series we use the Algorithm 1 of \cite{pham_generative}, then it is sufficient to multiply by $\frac{\sigma(R_{t_1:t_N})}{\sqrt{\Delta t}}$ to get to the initial scale. We remark that in the aforementioned works, the generative model is performed using the log-return series $(R_{t_i})_{i=1, \ldots, N}$, while 
here we perform the test taking only the increments of the initial time series, as the result of the generative model does not particularly change in this example. 

The only hyperparameters required are the kernel bandwidth $h$ and the Markovianity order $k$ of the time series. Using the test described in Section \ref{section_MBtest} combined with the SBTS model, we get $k=1$, consistent with the assumption of a Markovian model, and $h=0.1$. We stress that when the time series have relatively low variance, very small values of $h$ may lead to the generation of trajectories that nearly replicate the original data. To prevent this issue, we select a range of $h$ values for the bandwidth test that ensures an adequate number of observations within the kernel estimators.  

However, this generative model does not succeed in reproducing time series comparable to the original dataset. In Figure \ref{cont_paths_3} we display the first five simulated trajectories to have a qualitative representation of the behaviour of the synthetic time series: it is evident that working with a diffusion process, the generative model struggles to reproduce the large increments of the initial data attributed to the jump term. To illustrate this more clearly, we look at the distribution of $X_{t_i}$, for some $t_i$ fixed. Indeed, even if the constraint of the Schrödinger bridge problem is on the joint distribution of the time series, we can check whether the generative model captures also the marginal distributions. We compare the real and synthetic time series by examining the empirical quantiles of $X_{t_{50}}$ and representing them in a QQ-plot in Figure \ref{QQ_on_time_series}. The deviation of the points from the reference diagonal clearly indicates that the SBTS model fails to reproduce the distribution at this time: the empirical quantiles of the synthetic data are systematically different from those of the real data, revealing a significant mismatch between the two distributions.

\subsubsection{SBJTS generation}
\begin{figure}[t]
    \centering
    \includegraphics[width=0.5\textwidth]{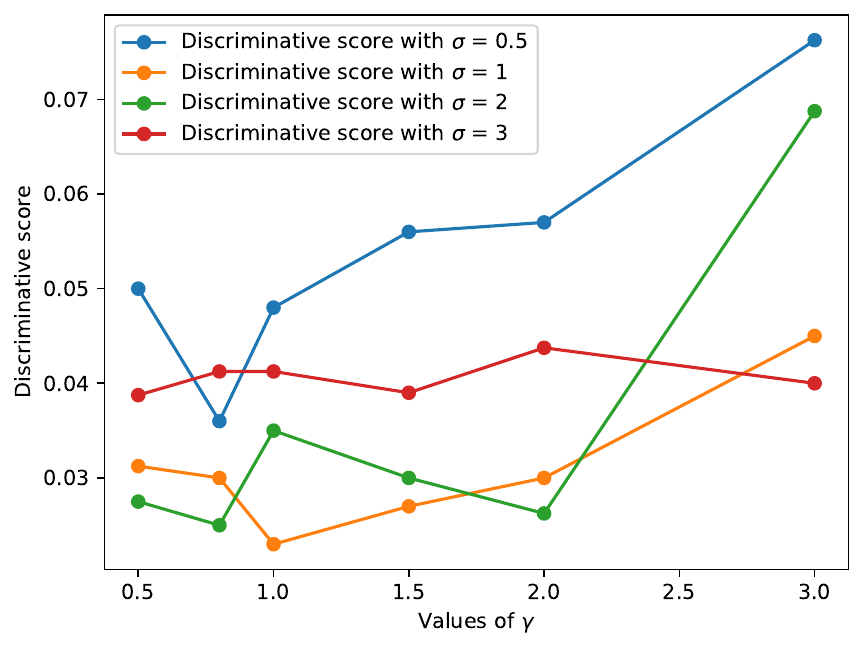}
    \caption{Discriminative scores tested on real time series and synthetic time series generated with different values of the hyperparameters $(\sigma, \gamma, \lambda^0)$.}
    \label{calibration_Merton}
\end{figure}
To assess the performance of the SBJTS model, we directly use the time series sampled from the Merton process without any rescaling. Instead, we calibrate the hyperparameters of the generative model to reproduce the key statistical characteristics of the original data. In this case, we consider the optimal process \eqref{SDE_under P*} in dimension 1, hence $\sigma\in \R^+$, and $\nu^0=\mathcal{N}(c, \gamma^2)$ is a Gaussian distribution with $c\in\R$, $\gamma\in\R^+$. We generate synthetic time series and present numerical experiments evaluating the algorithm’s performance. \\

\begin{figure}[t]
    \centering
    \begin{subfigure}[b]{0.3\textwidth}
        \includegraphics[width=\textwidth]{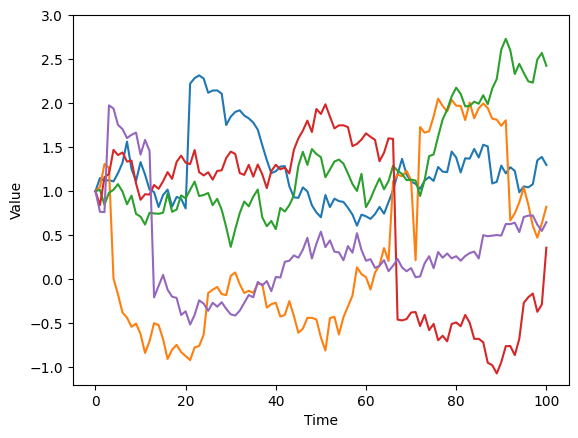}
    \end{subfigure}
    \hfill
    \begin{subfigure}[b]{0.3\textwidth}
        \includegraphics[width=\textwidth]{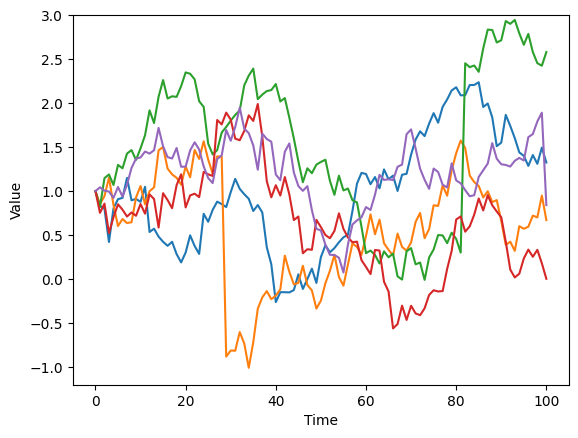}
    \end{subfigure}
    \hfill
    \begin{subfigure}[b]{0.3\textwidth}
        \includegraphics[width=\textwidth]{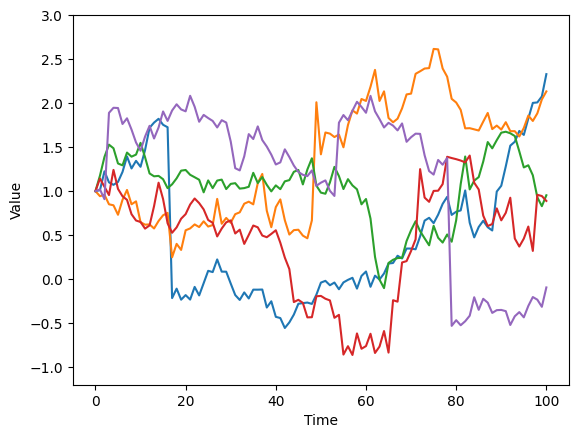}
    \end{subfigure}
    \caption{Generation of synthetic time series by SBJTS model: sample paths of real time series (left), sample paths of synthetic time series with the choice $h=0.3$, $k=1$, $\sigma=2$, $\lambda^0=5$, $c=0$, and $\gamma =0.8$ (case (\textit{i}) - middle), sample paths of synthetic time series with the choice $h=0.3$, $k=1$, $\sigma=1$, $\lambda^0=70$, $c=0$, and $\gamma =1$ (case (\textit{ii}) - right).}
    \label{paths_data_and_synth}
\end{figure}
\begin{figure}[h]
    \centering
    \begin{subfigure}[b]{0.49\textwidth}
        \includegraphics[width=\textwidth]{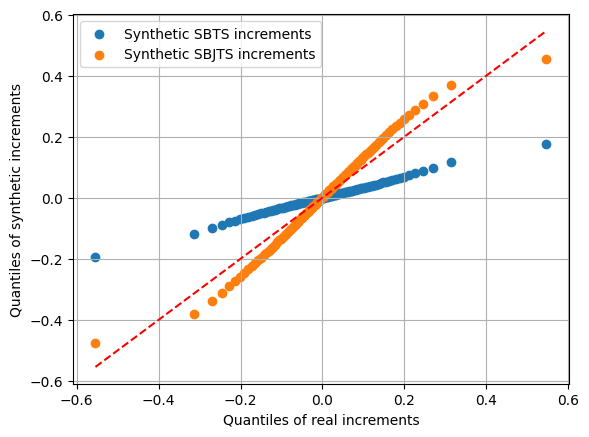}
    \end{subfigure}
    \hfill
    \begin{subfigure}[b]{0.49\textwidth}
        \includegraphics[width=\textwidth]{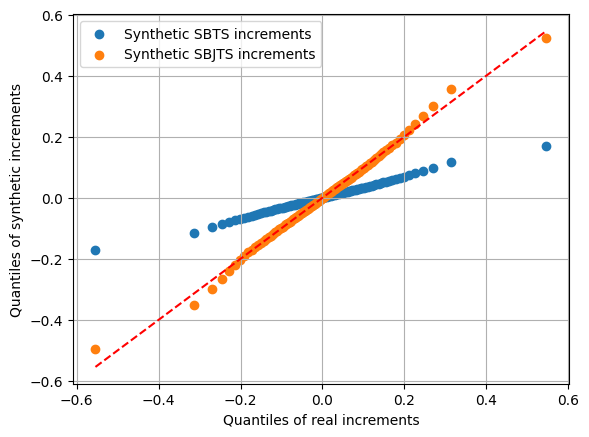}
    \end{subfigure}
    \caption{Generation of synthetic time series by SBJTS model: QQ-plot between the quantiles of the increments of initial time series and synthetic time series generated by SBTS and SBJTS in case (\textit{i}) (left) and in case (\textit{ii}) (right).}
    \label{QQ_ex1}
\end{figure}

\noindent \textbf{Calibration of the hyperparameters.}
Following the calibration procedure outlined in Section \ref{procedure_calib}, we aim to determine suitable values for the parameters $h$, $k$, $\sigma$, $\lambda^0$, $c$ and $\gamma$. In this example using Merton time series, we begin by setting $k=1$, consistent with the Markovian nature of the Merton model, and choose $h\in [0.1, 0.3]$. We fix $c=0$. Empirically, values of $\sigma<0.5$ or $\sigma>3$ fail to reproduce the empirical quadratic variation of the increments and result in synthetic time series with high discriminative scores. Accordingly, we test $\sigma \in \{  0.5, 1, 2, 3\}$. For each value of $\sigma$, using relation \eqref{var_relation}, we identify suitable values of $\gamma$ (and a corresponding $\lambda^0$); moreover, we find that $\gamma<1$ leads to numerical instability, while $\gamma>1.5$ produces synthetic series with excessively large jumps, causing the kernel estimators to vanish. Therefore, we test $\gamma\in \{0.5, 0.8, 1, 1.5, 2, 3\}$. For each pair $(\sigma, \gamma)$, we tune the value $\lambda^0$ via the test on the distribution of the quadratic variation (see Section \ref{hyper_selection}): in this way, we identify parameter triples $(\sigma, \gamma, \lambda^0)$ for which the synthetic time series have quadratic variation distribution close to the same distribution on real time series. Finally, for each choice of $(\sigma, \gamma, \lambda^0)$ among these candidates, we run the generative model and compute the discriminative score comparing the synthetic data to the real data. Figure \ref{calibration_Merton} reports the scores obtained for the different tested parameters in a calibration test: we observe that the parameter combinations $(\sigma, \gamma, \lambda^0)$ yielding the lowest scores are those close to the true Merton parameters $a, b, \lambda_\eta, m_J, v_J$ used to generate the initial time series. This behaviour is expected, and we indeed use this toy example to validate the calibration procedure. However, the dynamics \eqref{SDE_under P*} induced by our generative model differ from the original Merton process: in particular, the Schrödinger-bridge construction introduces a generally non-zero drift, typically adding extra variability, and a data-dependent jump term that does not coincide with the one in \eqref{Merton_sde}. Consequently, the fact that we can find a slight discrepancy between the score minimizers and the true parameter values is not contradictory. In this simulation, the lowest discriminative score is obtained testing the generative model with the following choice of hyperparameters: $h=0.3$, $k=1$, $\sigma=1$, $\lambda^0=70$, $c=0$, $\gamma =1$.\\

\noindent \textbf{Numerical results.} We present now the results of the SBJTS generative model. We fix $h=0.3$, $k=1$ and the two following set of hyperparameters:
\begin{enumerate}[label=(\textit{\roman*})]
    \item $\sigma=2$, $\lambda^0=5$, $c=0$, and $\gamma =0.8$: this choice corresponds exactly to the same volatility, mean and variance of jump sizes of the initial Merton process defined in \eqref{Merton_sde}, with the intensity $\lambda^0$ fixed to match the quadratic variation distribution;
    \item $\sigma=1$, $\lambda^0=70$, $c=0$, and $\gamma =1$: this choice is set to match the distribution of the quadratic variation of the initial time series, and it corresponds to the lowest value of discriminative score in Figure \ref{calibration_Merton} ($\approx 0.023$). 
\end{enumerate}
As a first visual check of the trajectories, in Figure \ref{paths_data_and_synth} we can see that the generated time series by SBJTS model have a general behaviour that is closed to the initial dataset and, in contrast with the simulation presented in Figure \ref{cont_paths_3} for the SBTS case, now it is evident the presence of jumps. The QQ-plots in Figure \ref{QQ_ex1} is built starting from the computation of the quantiles on all the increments of the time series: we can see that the distribution of the increments generated by SBJTS model has quantiles that closely match those of the original data, while the same quantiles on synthetic series produced by SBTS model exhibit significant discrepancy. 
\begin{figure}[t]
    \centering
    \includegraphics[width=0.8\textwidth]{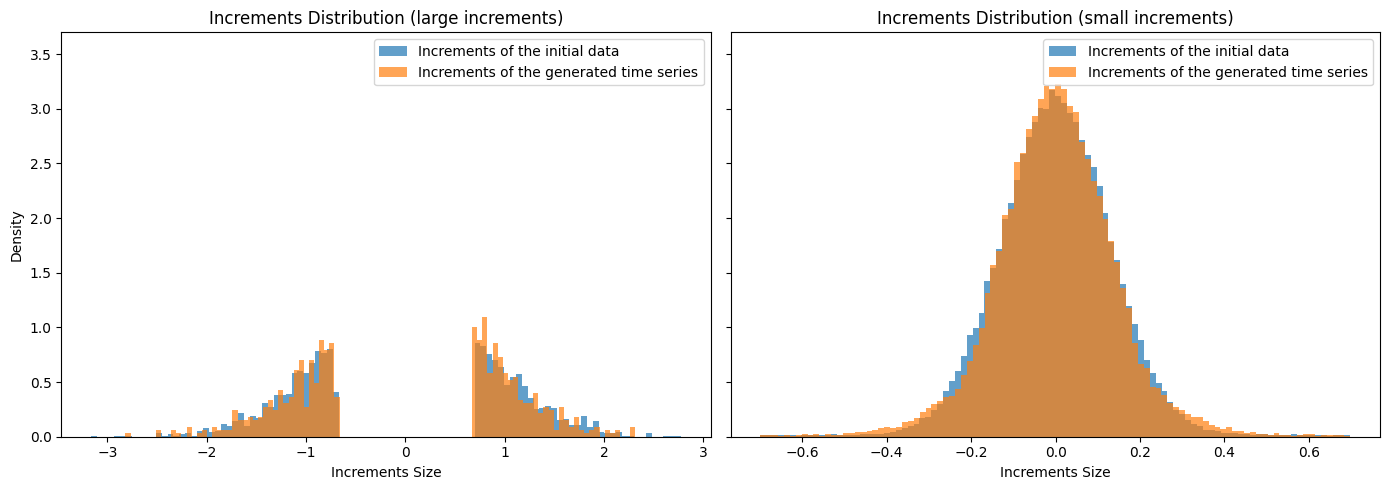}
    \caption{Comparison between the empirical distribution of the increments of real data and synthetic data generated via SBJTS model in case (\textit{ii}): distribution of increments larger than 0.7 (left) and smaller than 0.7 (right).}
    \label{dist_increments}
\end{figure}

\begin{figure}[t]
    \centering
    \begin{subfigure}[b]{0.45\textwidth}
        \includegraphics[width=\textwidth]{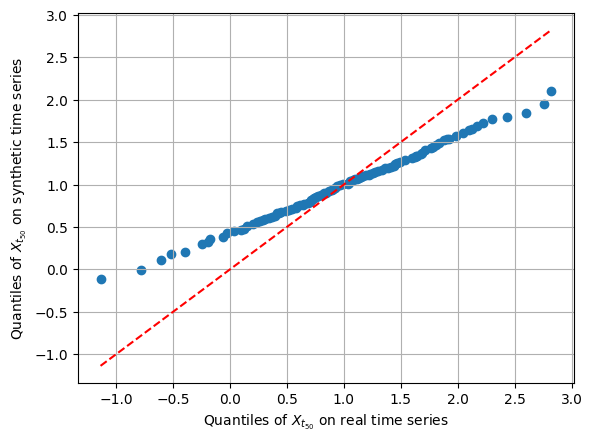}
        \caption{Generation via SBTS model.}
        \label{QQ_on_time_series}
    \end{subfigure}
    \hfill
    \begin{subfigure}[b]{0.45\textwidth}
        \includegraphics[width=\textwidth]{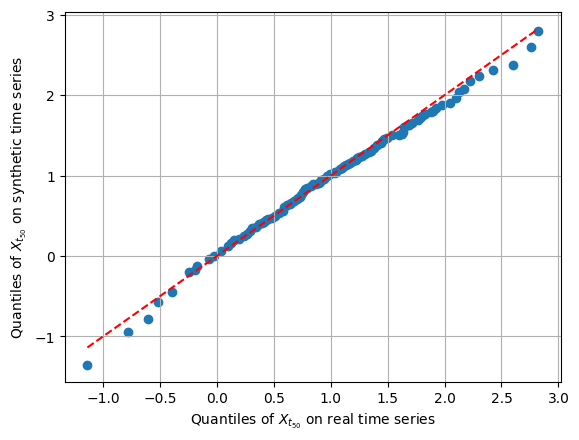}
        \caption{Generation via SBJTS model.}
        \label{QQ_of_timeseries}
    \end{subfigure}
    \caption{QQ-plot between the quantiles of the empirical distributions of $X_{t_{50}}$ on real and synthetic time series generated via SBTS model (left) and SBJTS model in case (\textit{ii}) (right).}
\end{figure}

To provide additional metrics in case (\textit{ii}), we plot in Figure \ref{dist_increments} the overall distribution of increments, where we fix a threshold of 0.7 to distinguish between small and large increments in both the original and generated time series. In this way, we can see that our generative model not only accurately captures the distribution of the small increments, but also the tails of the distribution, i.e.\ the large increments. This is particularly important as the jump term primarily affect this part of the distribution: at discrete sampling times, small jumps are essentially indistinguishable from diffusion-driven increments, whereas large jumps are much easier to identify. Examining the tail behaviour shows that the jump component is indeed well reproduced by our generative model. In Figure \ref{QQ_of_timeseries}, we focus on the marginal distribution of $X_{t_{50}}$: once again we can notice the improvement in the QQ-plot with respect to the SBTS model. In Appendix \ref{appendix_numerical_test} we present additional tests on the Merton time series.

\subsection{Pure jump case} \label{section_OU}
In this section, we provide an example in dimension one dropping the assumption on $\sigma$ to be non-degenerate, as we consider $\sigma=0$. In this way we recover the pure jump case, as we consider the solution $\mathbb{P}^*$ of the SBJTS problem to be the law of the process $X$ with dynamics
\begin{equation*}
    \begin{cases}
        \d X_t = \int_{\R} z N(\d t,\d z), \quad t\in [0,T],\\
        X_0 = 1
    \end{cases}
\end{equation*}
where $N$ has intensity measure $\lambda^*_t(z)\nu^0(\d z)\d t$. In this case, the expressions \eqref{a general} and \eqref{lambda general} can be derived working directly with the measures $\mu^0_\mathcal{T}(\d x_1,\ldots,\d x_N)$ and $\mu(\d x_1,\ldots,\d x_N)$, using the decomposition 
\begin{equation} 
    \mu^0_\mathcal{T}(\d x_1,\ldots,\d x_N) = \prod_{i=0}^{N-1} \mu^0_{i+1|i}(\d x_{i+1}),
\end{equation}
where $\mu^0_{i+1|i}(\d x_{i+1})$ denotes the conditional law of $X_{t_{i+1}}$ under $\mathbb{P}^0$, given the value of $X_{t_i}$. Making the choice $\nu^0 = \mathcal{N}(c, \gamma)$ with $c\in\R$ and $\gamma\in\R^+$, we have
\begin{align*}
    \mu^0_{i+1|i}(\d z) = e^{-\lambda^0 (t_{i+1}-t_{i})} \delta_0(\d z) + \sum_{k\geq 1} \frac{e^{-\lambda^0 (t_{i+1}-t_{i})}(\lambda^0 (t_{i+1}-t_{i}))^k}{k! \sqrt{2\pi k \gamma^2}} \exp{\left(-\frac{|z-k c|^2}{2 k \gamma^2}\right)} \d z.
\end{align*}
Hence, if we consider 
\begin{equation*}
    f^0_{t_{i+1}-t}(x_{i+1}-x) = 
    \begin{cases}
        e^{-\lambda^0 (t_{i+1}-t)}, & \text{if } x_{i+1}-x=0, \\
        \sum_{k\geq 1} \frac{(\lambda^0 (t_{i+1}-t))^k}{k!} \frac{e^{-\lambda^0 (t_{i+1}-t)}}{\sqrt{2\pi k \gamma^2}} \exp{\left(-\frac{|x_{i+1}-x-kc|^2}{2k \gamma^2}\right)}, &\text{otherwise},
    \end{cases}
\end{equation*}
for $t\in [t_i,t_{i+1})$, $i=0,\ldots,N_1$, then we get the expressions of $F_i(t,x_i,x,x_{i+1})$ defined in \eqref{funct_Fi} and the estimators of the drift and jump intensity. Moreover, also in this case we can use the Gaussian mixture model to simulate jump sizes: when we have a jump $J$ at time $t \in (t_i, t_{i+1})$ and $X_{t^-}=x$, we sample the amplitude of jump from the following distribution
\begin{equation*}
    J \sim \sum_{m=1}^M \sum_{j\geq 0} w_{j,m} \, \mathcal{N}_{j,m}
\end{equation*}
where
\begin{align*}
    w_{j,m} = 
    \frac{ \frac{(\lambda^0 (t_{i+1}-t))^j}{j!} \frac{e^{-\lambda^0 (t_{i+1}-t)}}{\sqrt{2\pi (j+1)\gamma^2}} \exp{\left( -\frac{|X^{(m)}_{t_{i+1}}-x-(j+1)c|^2}{2(j+1)\gamma^2}  \right)}}
    { e^{-\lambda^0 (t_{i+1}-t_{i})} \mathbbm{1}_{\{x_i = X^{(m)}_{t_{i+1}} \}} + \sum_{k\geq1} \frac{(\lambda^0 (t_{i+1}-t_{i}))^k}{k!} \frac{e^{-\lambda^0 (t_{i+1}-t_{i})}}{\sqrt{2\pi k \gamma^2}} \exp{\left(-\frac{|X^{(m)}_{t_{i+1}}-x_{i}-kc|^2}{2 k \gamma^2}\right)}} K_h(\mathbf{x}_{i}-\mathbf{X}^{(m)}_{t_i})  
\end{align*}
and 
\begin{align*}
    \mathcal{N}_{0,m} &= \delta_{\{0\}}(X^{(m)}_{t_{i+1}}-x),  \quad \\
    \mathcal{N}_{j,m} &= \mathcal{N}\left(\frac{X^{(m)}_{t_{i+1}}-x}{j+1}, \frac{j \gamma^2}{j+1}\right), \quad \text{if } j\geq 1.   
\end{align*}
\begin{figure}[t]
    \centering
    \includegraphics[width=0.7\textwidth]{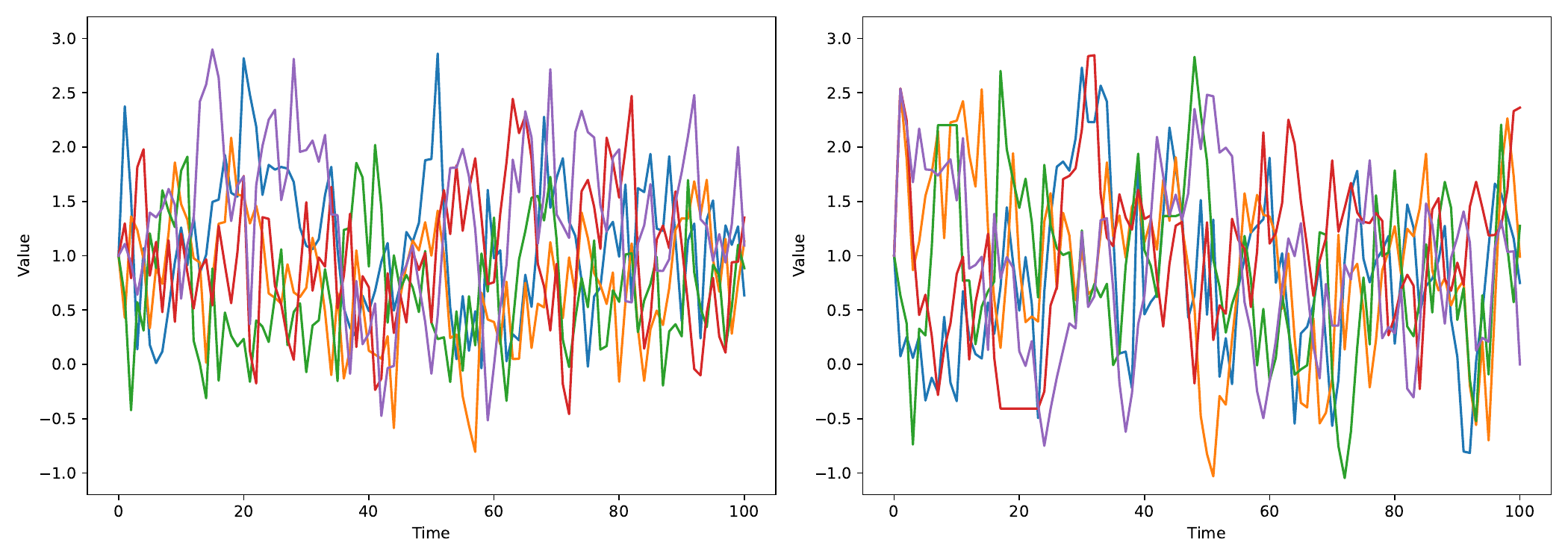}
    \caption{Paths of initial time series sampled by the trajectories of the OU process (left), and generated time series by SBJTS model in the case of pure jump generative process (right).}
    \label{purely_jump_1}
\end{figure}
\begin{figure}[t]
    \centering
    \begin{subfigure}[b]{0.47\textwidth}
        \includegraphics[width=\textwidth]{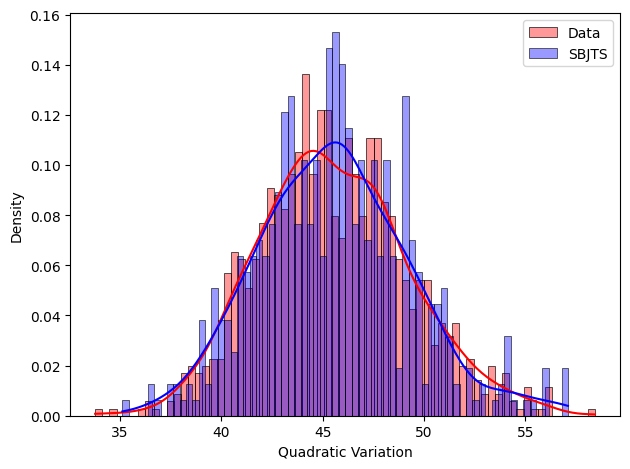}
        \caption{Quadratic variation distributions.}
        \label{fig:qv_pure_jump}
    \end{subfigure}
    \hfill
    \begin{subfigure}[b]{0.47\textwidth}
        \includegraphics[width=\textwidth]{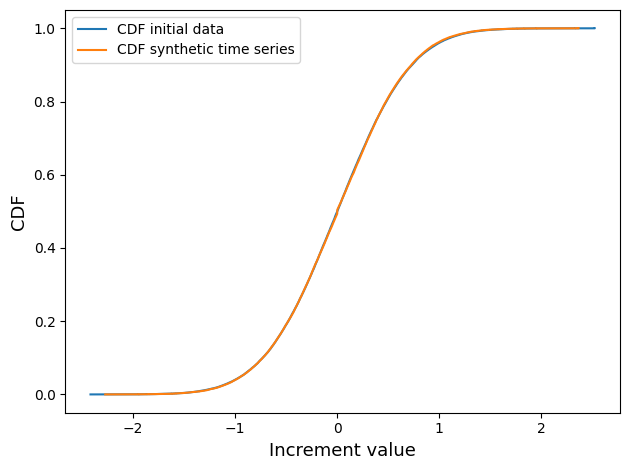}
        \caption{CDFs of increment distributions.}
        \label{fig:cdf_pure_jump}
    \end{subfigure}
    \caption{Metrics to compare real and synthetic datasets generated by the SBJTS model in the pure jump case.}
    \label{fig:pure_jump_metrics}
\end{figure}
In this pure jump case, we test the SBJTS model starting from simulated time series built as discretization of the sample paths of Ornstein–Uhlenbeck process. We consider the process in dimension $d=1$ defined by
\begin{equation}\label{sde_ou}
    \begin{cases}
        \d Y_t = \theta(a - Y_t)\,\d t + b\, \d W_t, \quad t\in [0,T], \quad \\
        Y_0 = 1
    \end{cases}
\end{equation}
and we fix $\theta=100$, $a=1$, $b=10$, $T=\frac{100}{252}$ in order to have strongly mean reverting trajectories also when we look at time series sampled at low frequency. We want to show that working with our state-dependent jump process properly calibrated we can still generate synthetic time series sampled from the trajectories of the pure jump process $X$ with metrics close to the ones of the initial time series. With respect to the previous example where there was also the diffusion term (and in particular the drift) which contributes to replicate the correct behaviour of the time series, here we focus on the jump term. In this case we use Algorithm \ref{alg2} to generate synthetic time series. 

\paragraph{Numerical results.} 
Starting from the trajectories of the process \eqref{sde_ou}, we fix $\Delta t_i = \frac{1}{252}$ to get initial time series of length $N=100$. We then generate 500 synthetic time series, choosing as hyperparameters of the generative model $h=0.3$, $k=1$, $\lambda^0=0.001$, $c=0$ and $\gamma =0.1$: see Figure \ref{purely_jump_1} to compare the trajectories of 5 realizations of real and synthetic time series. In Figure \ref{fig:qv_pure_jump} we first plot the empirical distribution of the quadratic variation computed along the real and synthetic time series to show that the chosen hyperparameters yield closely matching distributions. Then we compare the empirical cumulative distribution functions of the increments for the real and synthetic time series. In the pure-jump setting, this provides a global assessment of the model's ability to reproduce the distribution of both small and large jumps over the entire time series. As shown in Figure \ref{fig:cdf_pure_jump}, the empirical CDFs of the real and synthetic increments closely overlap, indicating that the proposed generative model accurately captures the jump dynamics. In Appendix \ref{appendix_ou}, we present the same test taking a finer time discretization: in this way we look at the time series obtained from the same trajectories but working at high frequency.

\subsection{Test on real data: Stock and Energy dataset} \label{section_real}
In this section, we apply the SBJTS generative model to the case of real datasets which store historical observations of multiple variables related to specific domains, hence we adopt the multidimensional formulation of the problem. From the historical observations we construct time series samples by applying a sliding-window procedure with overlapping windows. In this way we construct time series $X_{t_0 : t_N}=(X_{t_0}, \ldots,X_{t_N})$, and then we use the following two normalizations:
\begin{itemize}
    \item Base one normalization: $\left(\frac{X_{t_0 : t_N, p}}{X_{t_0,p}}\right)$, for ${p=1,\ldots,d}$;
    \item Standard normalization: $\left(\frac{X_{t_1 : t_N,p} -m(X_{t_1 : t_N,p})}{v(X_{t_1 : t_N,p})}\right)$, for ${p=1,\ldots,d}$, where we denote the empirical mean of the $p$-th component of the entire time series by $m(X_{t_1 : t_N, p})$ and the empirical standard deviation by $\sigma(X_{t_1 : t_N,p})$. 
\end{itemize}
Hence we proceed with the calibration of the hyperparameters through a grid–search procedure. To reduce the computational complexity, we first tune the parameters component-wise in dimension 1, following the methodology described in Section \ref{section_calibration}. Once suitable parameter candidates have been identified, we can perform a full $d$-dimensional calibration, where synthetic multivariate time series are generated for each candidate combination, and the set achieving the lowest discriminative score is selected. The parameters $h$, $k$ and $\lambda^0$ are finally fixed globally through $d$-dimensional generation. We run all the tests using Algorithm \ref{alg3}. After obtaining new time series, we invert the standard normalization to return to the base-one scale, in order to compare them to the original data. 

We test our generative model on two distinct datasets, and we compare our results with several state-of-the-art generative models for time series that are evaluated in the literature in terms of both generation quality and predictive performance. In particular, we consider models based on GAN architectures \cite{yoon2019time}, flow-matching approaches \cite{hu2024fm}, optimal transport methods \cite{xu2020cot}, and diffusion-based generative models employing score-matching techniques \cite{lim2025tsgm, naiman2024utilizing}. For consistency with the existing benchmark, we compare real and synthetic time series of fixed length $N=24$. 
\begin{itemize}
    \item \textbf{Stock dataset:} daily historical Google stock data from 2004 to 2019, including six features: open, high, low, close, adjusted close prices, and volume, hence $d=6$. We choose the following parameters: $\sigma_p = 0.7$ for $p=1, \ldots, d-1$, $\sigma_d=1$, $\lambda^0=0.2$, $c_p = 0$ for $p=1, \ldots, d$, $\gamma_p = 0.1$ for $p=1, \ldots, d-1$, $\gamma_d = 0.6$. We generate 2000 synthetic time series.
    \item \textbf{Energy dataset:} electricity consumption and environmental conditions of a low-energy residential building over time, taken from \cite{candanedo2017data}. Its main target variable is the total energy consumed by household appliances. The dataset includes a variety of environmental and indoor features like the temperature readings from multiple rooms, relative humidity levels in the same areas, and outdoor measurements such as temperature, relative humidity, wind speed, atmospheric pressure. The data were collected every 10 minutes for 137 days. This dataset has 28 features, hence $d=28$, and we select the parameters $\sigma_1= \sigma_2 =1.2$, $\sigma_p =1$ for $p=3, \ldots, d$, $\lambda^0=0.5$, $c_p = 0$ for $p=1, \ldots, d$, $\gamma_1=\gamma_2=0.3$ and $\gamma_p = 0.1$, $p=3, \ldots, d$. We generate 1000 synthetic time series.
\end{itemize}
\begin{table}[t]
  \centering
  \footnotesize
  \begin{tabular}{lcccc}
    \toprule
    & \multicolumn{2}{c}{Stock} & \multicolumn{2}{c}{Energy} \\
    \cmidrule(lr){2-3} \cmidrule(lr){4-5}
    Model & Disc. $\downarrow$ & Pred. $\downarrow$ & Disc. $\downarrow$ & Pred. $\downarrow$ \\
    \midrule
    TSGM-VP        & 0.022 $\pm$ 0.005 & 0.037 $\pm$ 0.000 & 0.221 $\pm$ 0.025 & 0.257 $\pm$ 0.000 \\
    TSGM-subVP     & 0.021 $\pm$ 0.008 & 0.037 $\pm$ 0.000 & 0.198 $\pm$ 0.025 & 0.252 $\pm$ 0.000 \\
    ImagenTime     & 0.037 $\pm$ 0.006 & 0.036 $\pm$ 0.000 & \textbf{0.040 $\pm$ 0.004} & 0.250 $\pm$ 0.000 \\
    T-Forcing      & 0.226 $\pm$ 0.035 & 0.038 $\pm$ 0.001 & 0.483 $\pm$ 0.004 & 0.315 $\pm$ 0.005 \\
    P-Forcing      & 0.257 $\pm$ 0.026 & 0.043 $\pm$ 0.001 & 0.412 $\pm$ 0.006 & 0.303 $\pm$ 0.006 \\
    TimeGAN        & 0.102 $\pm$ 0.031 & 0.038 $\pm$ 0.001 & 0.236 $\pm$ 0.012 & 0.273 $\pm$ 0.004 \\
    RCGAN          & 0.196 $\pm$ 0.027 & 0.040 $\pm$ 0.001 & 0.336 $\pm$ 0.017 & 0.292 $\pm$ 0.005 \\
    C-RNN-GAN      & 0.399 $\pm$ 0.028 & 0.038 $\pm$ 0.000 & 0.499 $\pm$ 0.001 & 0.483 $\pm$ 0.005 \\
    TimeVAE        & 0.175 $\pm$ 0.031 & 0.042 $\pm$ 0.002 & 0.498 $\pm$ 0.006 & 0.268 $\pm$ 0.004 \\
    WaveGAN        & 0.217 $\pm$ 0.022 & 0.041 $\pm$ 0.001 & 0.363 $\pm$ 0.012 & 0.307 $\pm$ 0.007 \\
    COT-GAN        & 0.285 $\pm$ 0.030 & 0.044 $\pm$ 0.000 & 0.498 $\pm$ 0.000 & 0.260 $\pm$ 0.000 \\
    FlowTS         & 0.019 $\pm$ 0.013 & 0.036 $\pm$ 0.000 & 0.053 $\pm$ 0.010 & 0.250 $\pm$ 0.000 \\
    SBTS           & \textbf{0.010 $\pm$ 0.008} & \textbf{0.017 $\pm$ 0.000} & 0.356 $\pm$ 0.020 & \textbf{0.072 $\pm$ 0.001} \\
    \textbf{SBJTS} & \textbf{0.036 $\pm$ 0.031} & \textbf{0.018 $\pm$ 0.005} & \textbf{0.065 $\pm$ 0.031} & \textbf{0.080 $\pm$ 0.011} \\
    \bottomrule
  \end{tabular}
  \caption{Benchmark on discriminative and predictive scores for Stock and Energy datasets: comparison of the performance of the SBJTS model with other generative models. In bold: the lowest discriminative and predictive score among the benchmark and the scores for the SBJTS model. }
  \label{table_scores}
\end{table}
Once that we perform the generation of synthetic time series, we measure the quality of the synthetic data using the following two metrics:
\begin{itemize}
    \item \textbf{Discriminative score:} see Section \ref{hyper_selection}.
    \item \textbf{Predictive score:} we aim to determine whether the synthetic data preserve the predictive structure of the real sequences. The idea is that, if the generative model is well trained, the synthetic samples should capture the temporal dependencies and conditional distributions present in the original data. Concretely, we train on the synthetic dataset a recurrent neural network, which learns to predict the $d$-th component from time steps $t_2$ to $t_N$, given the first $d-1$ components from $t_1$ to $t_{N-1}$, minimizing the mean absolute error. Then the trained predictor is tested on the real dataset, and the global mean absolute error gives the final score. 
\end{itemize}
Following the benchmark literature, we compute the scores using post-hoc recurrent neural networks (GRU type) with batch size 128 and hidden dimension $\max(\frac{d}{2},1)$, using a single layer for the predictive score and two layers for the discriminative score. To determine a value of the scores, we compare the same number of synthetic time series and randomly selected real samples; for each dataset, the final score is the mean of the results of 10 runs of the test, the error is the standard deviation. In Table \ref{table_scores} we report the discriminative and predictive scores obtained using the SBJTS generative model on the Stock and Energy datasets compared to the scores of the other generative models taken from the aforementioned papers. In both cases, we underline in bold the scores obtained with our generative model and the lowest scores among the state-of-the-art benchmark. We can conclude that the SBJTS has a very competitive performance on the two datasets. Indeed, in the case of the Stock dataset, the scores are really close to the ones obtained for the SBTS model, which has already a very good performance compared to the other generative models. On the other hand, in the case of the Energy dataset we manage to definitely improve the discriminative score obtained with the SBTS generative model, without deteriorating the predictive score. This is motivated by the fact that the dataset presents some components whose behaviour can be better reproduced adding the jump term in the generative process, leading to higher accuracy on the generation of synthetic time series using our generative model, and therefore a low discriminative score.

\section{Conclusion}

In this work, we addressed the problem of generating realistic time series by extending Schr\"odinger bridge–based generative models to stochastic dynamics with jumps. Our main contribution is  the generalization of the continuous-time framework of Hamdouche et al.\ \cite{pham_generative} to probability measures defined on the space of càdlàg paths, allowing the reference process to be a jump–diffusion. Under a finite Kullback-Leibler divergence assumption, we characterized the solution of the resulting Schr\"odinger bridge problem, and explicitly identified the dynamics of the optimal controlled process which, sampled at the observation dates, generates synthetic time series with joint distribution matching a prescribed target distribution. We further proposed a systematic hyperparameter calibration strategy that improves the numerical efficiency and robustness of the generative model. 

The proposed framework nevertheless presents some limitations. In particular, the choice of the reference measure remains restrictive, and while the model learns jump sizes from the data, it does not explicitly aim to reproduce the empirical distribution of jump times. Addressing these aspects — by allowing for more flexible reference dynamics or by directly modeling jump-time distributions — constitutes a natural direction for future research. Despite these limitations, the extension of Schr\"odinger bridge generative models to jump–diffusion dynamics leads to consistent improvements in numerical experiments on both simulated and real-world datasets. In particular, the proposed approach shows competitive performance compared with state-of-the-art generative models, notably in capturing heavy-tailed behavior, abrupt variations, and regime changes in realistic time series.



\begin{appendices}

\section{Predictability and Random Measures} \label{app_A}
In this appendix, we recall the foundational definitions and structural results concerning predictability and random measures operating on a filtered probability space $(\Omega, \mathcal{F}, (\mathcal{F}_t)_{t \ge 0}, \mathbb{P})$ satisfying the usual conditions. These notions are essential for setting up the rigorous framework of Schrödinger bridges for jump-diffusions. For a comprehensive treatment, we refer the reader to \cite{applebaum2009levy, jacod2013limit}.

\paragraph{Predictable functions.} Let $(E, \mathcal{E})$ be a measurable space representing the mark space (typically a Borel subset of $\mathbb{R}^d$). The predictable $\sigma$-algebra $\mathcal{P}_{\mathrm{pred}}$ on $\Omega \times \mathbb{R}^+$ is defined as the $\sigma$-algebra generated by all left-continuous, $(\mathcal{F}_t)$-adapted processes. Equivalently, $\mathcal{P}_{\mathrm{pred}}$ is generated by sets of the form $A\times\{0\}$ for $A \in \mathcal{F}_0$ and $A\times(s, t]$ for $0 \le s < t$ and $A \in \mathcal{F}_s$. A stochastic process defined on $\Omega \times \mathbb{R}^+ \times E$ is called predictable if it is measurable with respect to the product $\sigma$-algebra $\mathcal{P}_{\mathrm{pred}} \otimes \mathcal{E}$. An integer-valued random measure $N$ on
$\mathbb{R}^+\times E$ is a random measure such that
$N(\omega,B)\in\mathbb{N}\cup\{+\infty\}$ for every
$B\in\mathcal{B}(\mathbb{R}^+)\otimes\mathcal{E}$.
We write it as $N(\mathrm{d}t,\mathrm{d}z)$.

\paragraph{Case 1: Poisson Random Measures (Deterministic Intensity).} We first introduce the concept of Poisson random measure. An integer-valued random measure $N$ is a Poisson random measure with a $\sigma$-finite, deterministic intensity measure $\mu$ on $\mathcal{B}(\mathbb{R}^+) \times \mathcal{E}$ if:
\begin{itemize}
    \item for any measurable set $B \in \mathcal{B}(\mathbb{R}^+) \otimes \mathcal{E}$ with $\mu(B) < \infty$, the random variable $N(B)$ follows a Poisson distribution with mean $\mu(B)$;
    \item for any finite collection of disjoint measurable sets $B_1, \dots, B_k \in \mathcal{B}(\mathbb{R}^+) \otimes \mathcal{E}$, the random variables $N(B_1), \dots, N(B_k)$ are mutually independent.
\end{itemize}
In the standard time-homogeneous setting, the deterministic intensity factorizes as $\mu(\mathrm{d}t, \mathrm{d}z) = \nu(\mathrm{d}z)\mathrm{d}t$, where $\nu$ is a $\sigma$-finite measure on $(E, \mathcal{E})$. With respect to the natural filtration of $N$, the
increments of $N$ are independent of the past. Hence, whenever
$\nu(A)<\infty$,
\[
\mathbb{E}\left[
N((s,t]\times A)\mid\mathcal{F}_s
\right]
=
(t-s)\nu(A),
\qquad 0\leq s<t.
\]

\paragraph{Case 2: Integer-Valued Random Measures with Stochastic Intensity.} In the context of jump-diffusions or marked point processes where jump rates change dynamically based on the state or history, the increments of $N$ need not be independent of the past. Instead, the behavior of $N$ is governed by a predictable compensator. Under the stochastic-intensity framework, we assume the compensator $\mu(\omega; \mathrm{d}t, \mathrm{d}z)$ is absolutely continuous with respect to a deterministic measure $\nu(\mathrm{d}z)\mathrm{d}t$, taking the form:
$$\mu(\omega; \mathrm{d}t, \mathrm{d}z) = \lambda_t(\omega, z) \nu(\mathrm{d}z) \mathrm{d}t,$$
where $\lambda: \Omega \times \mathbb{R}^+ \times E \to \mathbb{R}^+$ is a non-negative, $\mathcal{P}_{\mathrm{pred}} \otimes \mathcal{E}$-measurable process. Because $\mu$ is the predictable compensator of $N$, the relation to the filtration $(\mathcal{F}_t)_{t \ge 0}$ is explicitly captured via conditional expectations over increments:
$$\mathbb{E}\left[N((s,t] \times A) \mid \mathcal{F}_s\right] = \mathbb{E}\left[\int_s^t \int_A \lambda_u(z) \nu(\mathrm{d}z) \mathrm{d}u \;\Bigg|\; \mathcal{F}_s\right],$$
for any $0 \le s < t$ and $A \in \mathcal{E}$.

\paragraph{Compensated Measures and Stochastic Integration.} For both settings described above, we define the compensated random measure $\tilde{N}$ by subtracting the predictable compensator from the raw counting measure:
$$\tilde{N}(\mathrm{d}t, \mathrm{d}z) = N(\mathrm{d}t, \mathrm{d}z) - \lambda_t(z)\nu(\mathrm{d}z)\mathrm{d}t,$$
where $\lambda_t(z)$ is deterministic in Case 1, and a predictable random process in Case 2. Let $F: \Omega\times \mathbb{R}^+ \times E \to \mathbb{R}$ be a $\mathcal{P}_{\mathrm{pred}} \otimes \mathcal{E}$-measurable function satisfying the square-integrability condition on a finite horizon $[0, T]$:
$$\mathbb{E} \left[ \int_0^T \int_E |F(s,z)|^2 \lambda_s(z)\nu(\mathrm{d}z)\mathrm{d}s \right] < \infty.$$ 
Then, the stochastic integral process $M = (M_t)_{t \in [0,T]}$ defined by:
$$M_t = \int_0^t \int_E F(s,z) \tilde{N}(\mathrm{d}s, \mathrm{d}z)$$
is a well-defined, square-integrable martingale with respect to $(\mathcal{F}_t)_{t \ge 0}$.

\section{Additional numerical tests} \label{app_B}
\subsection{Merton time series: distribution of jumps and simulated trajectories}\label{appendix_numerical_test}
\begin{figure}[h]
    \centering
    \begin{subfigure}[b]{0.45\textwidth}
        \includegraphics[width=\textwidth]{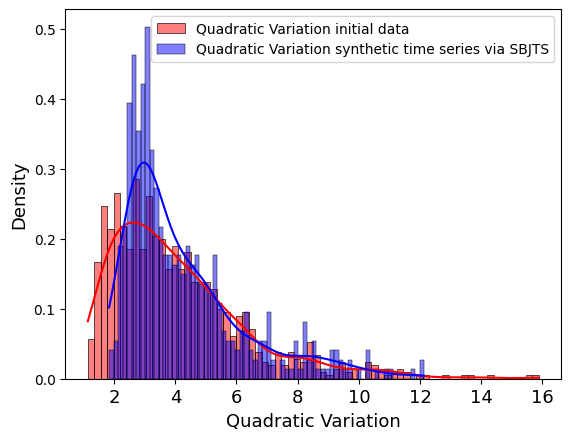}
    \end{subfigure}
    \hfill
    \begin{subfigure}[b]{0.45\textwidth}
        \includegraphics[width=\textwidth]{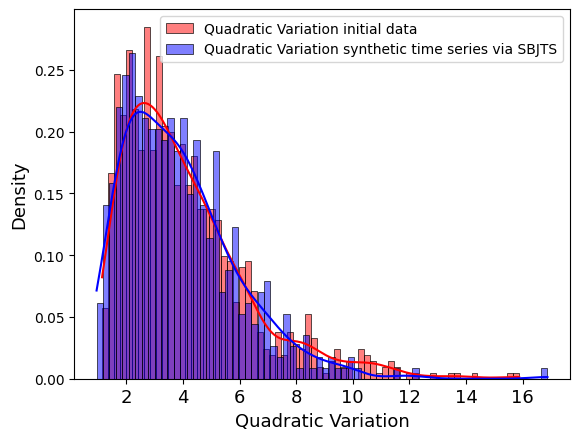}
    \end{subfigure}
    \caption{Comparison between the quadratic variation distribution of the initial time series and the generated time series in case (\textit{i}) (left) and in case (\textit{ii}) (right).}
    \label{quad_var_app}
\end{figure}
\begin{figure}[h]
    \centering
    \begin{subfigure}[b]{0.45\textwidth}
        \includegraphics[width=\textwidth]{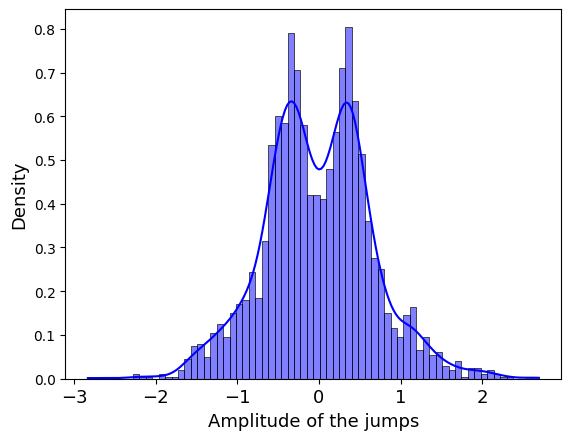}
    \end{subfigure}
    \hfill
    \begin{subfigure}[b]{0.45\textwidth}
        \includegraphics[width=\textwidth]{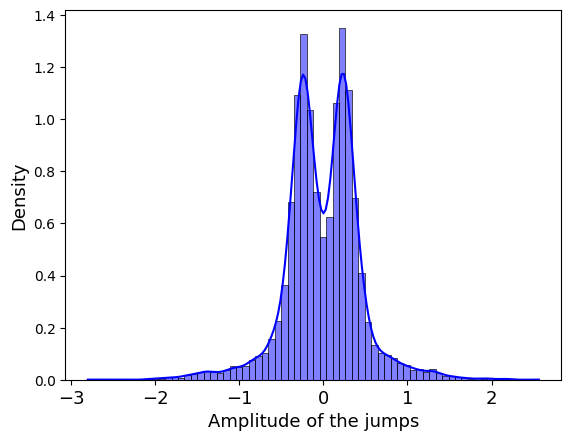}
    \end{subfigure}
    \caption{Empirical distribution of the generated jump sizes in case (\textit{i}) (left) and in case (\textit{ii}) (right) estimated over 500 synthetic time series.}
    \label{jump_ex1}
\end{figure}
\noindent We report here additional plots obtained in the generation of time series starting from the Merton model. In Figure \ref{quad_var_app} we can see that the two choices of parameters presented in Section \ref{section_merton} (denoted always with case (\textit{i}) and (\textit{ii})) allow to generate time series whose quadratic variation distribution matches the quadratic variation distribution of the initial data. Figure \ref{jump_ex1} displays the empirical distribution of the generated jump sizes in case (\textit{i}) and (\textit{ii}). In both cases the distribution of jump sizes is bimodal. This suggests that the generative model captures small increments of the original data primarily with the diffusion component, while the larger jumps are captured by the jump mechanism.

\subsection{Pure jump generative model} \label{appendix_ou}
In this appendix we provide further simulations using Algorithm \ref{alg2} for our generative model: first we look at Ornstein-Uhlenbeck time series in a high-frequency case, then we consider time series from the price process of a Heston model. 

\paragraph{Ornstein-Uhlenbeck process.}
Starting from the same trajectories of the Ornstein-Uhlenbeck process \eqref{sde_ou} sampled for the numerical test in Section \ref{section_OU}, we change the time discretization of the interval $[0,T]$ decreasing the time step $\Delta t_i=t_{i+1}-t_i$: we consider $\Delta t_i = \frac{1}{252} \cdot\frac{1}{10}$, hence we have time series with length $N=1000$ (see Figure \ref{purely_jump_2}). We fix $h=0.15$, $k=1$ and, for the jump term, $\lambda^0= 0.005$, $c=0$ and $\gamma =0.1$.

To assess the fidelity of the synthetic time series, we first compare the empirical marginal distributions by considering the quantiles at different time points. Figure \ref{metric_dist_jump_purely} presents QQ-plots of the distributions of $X_{t_{100}}, X_{t_{600}}, X_{t_{800}}$ for the real and synthetic data. We then examine the temporal dependence structure through the empirical autocorrelation function (ACF). Specifically, Figure \ref{ACF_high_freq} reports the average ACF computed over the real and synthetic time series for lags in $\{1, \ldots, 40\}$. The synthetic data successfully reproduce the main autocorrelation patterns observed in the real data, showing good agreement in the magnitude of the correlations across the different lags. For completeness, Figure \ref{ACF_high_freq} also includes the ACF corresponding to the experiment presented in Section \ref{section_OU}.
 
\begin{figure}[h]
    \centering
    \includegraphics[width=0.7\textwidth]{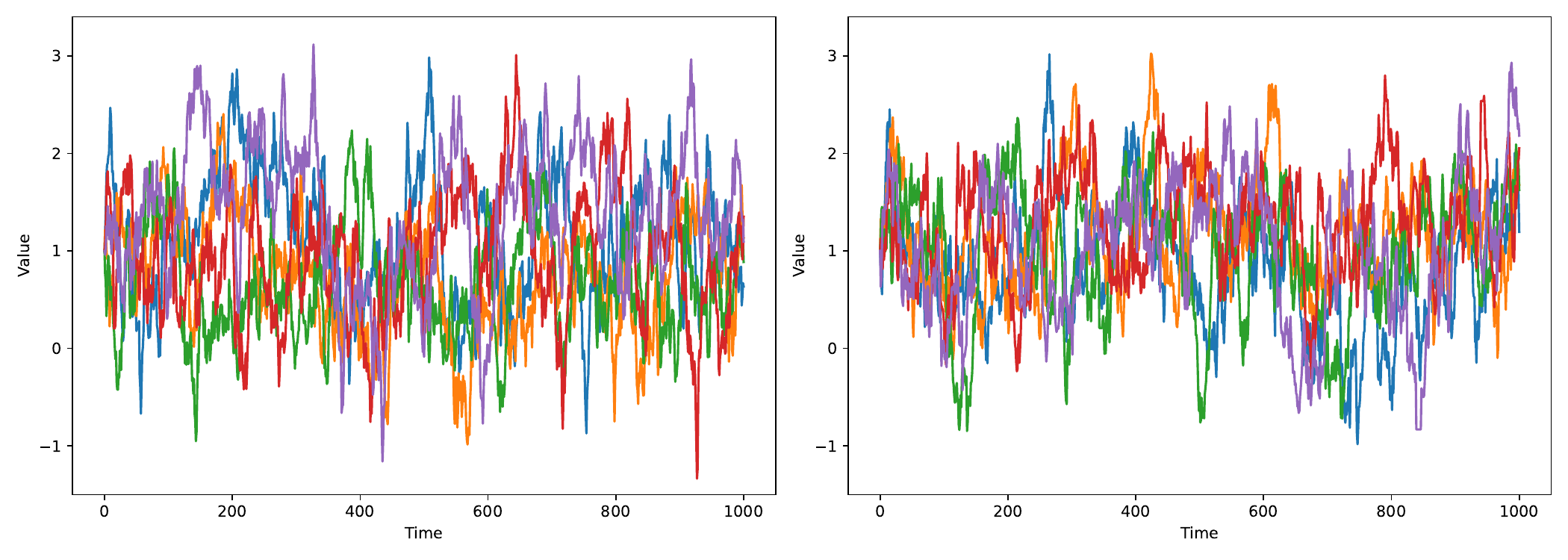}
    \caption{Sample paths of initial time series from trajectories of Ornstein–Uhlenbeck process (left) and generated time series by SBJTS (right) using the pure jump model.}
    \label{purely_jump_2}
\end{figure}

\begin{figure}[htb]
    \centering
    \begin{subfigure}[b]{0.3\textwidth}
        \includegraphics[width=\textwidth]{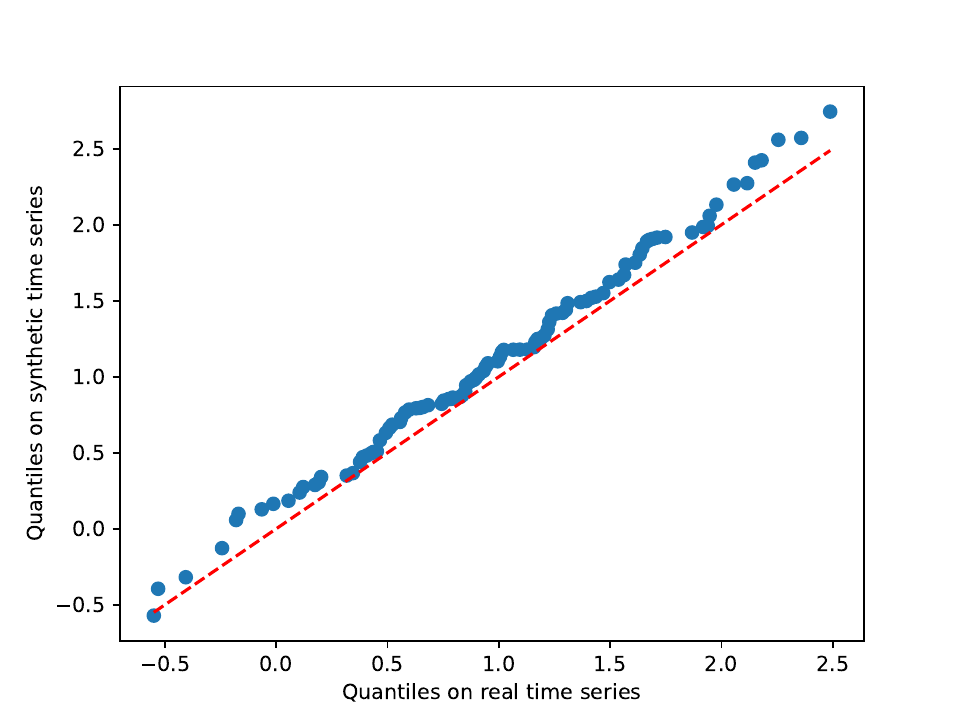}
    \end{subfigure}
    \hfill
    \begin{subfigure}[b]{0.3\textwidth}
        \includegraphics[width=\textwidth]{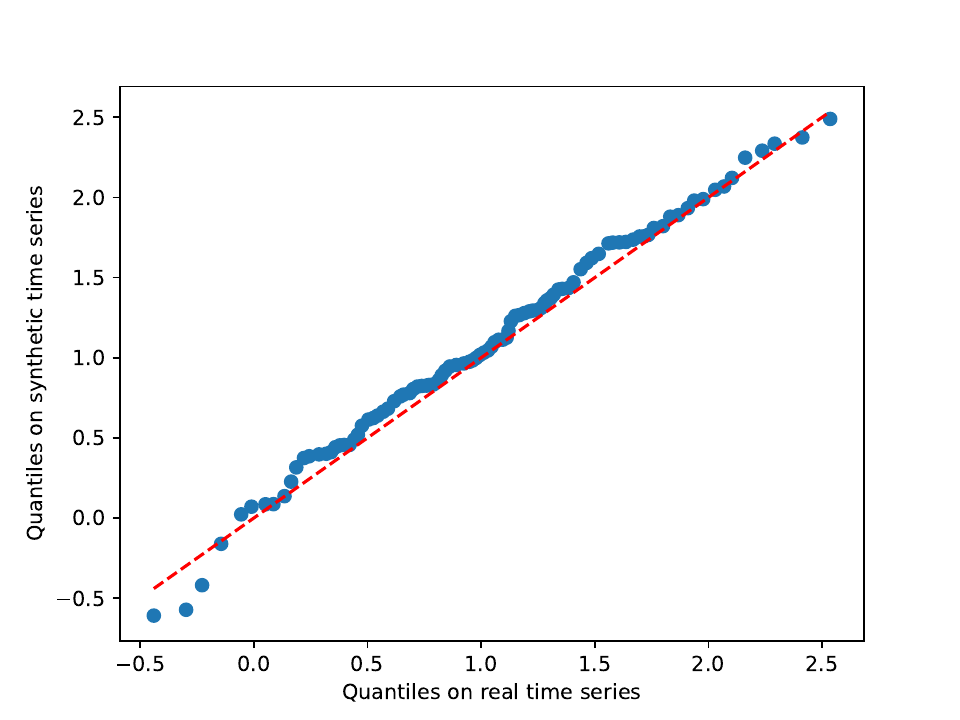}
    \end{subfigure}
    \hfill
    \begin{subfigure}[b]{0.3\textwidth}
        \includegraphics[width=\textwidth]{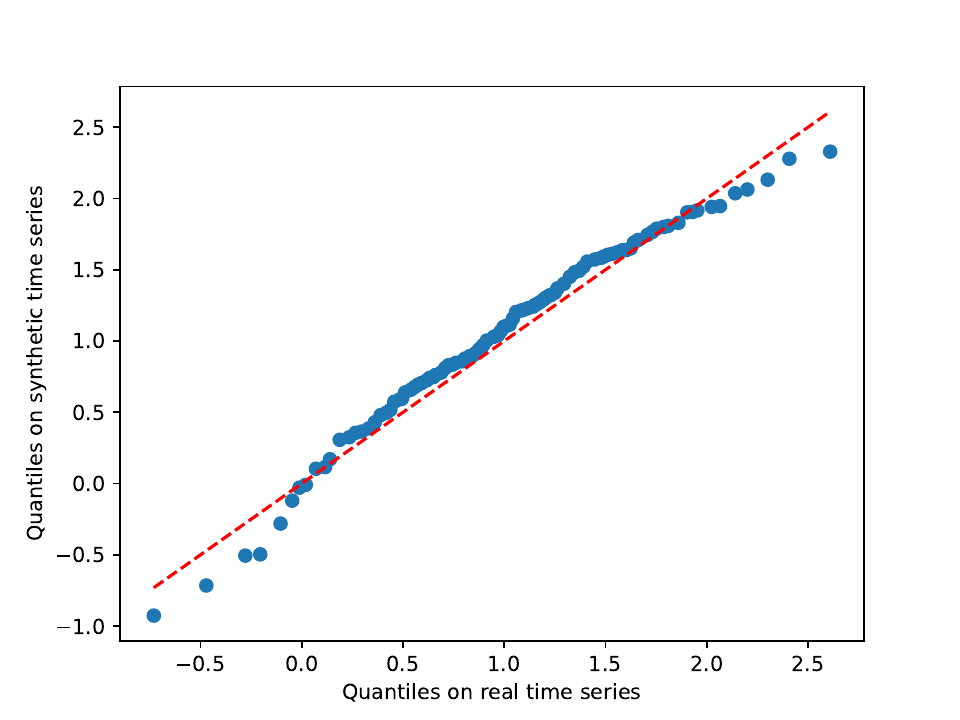}
    \end{subfigure}
    \caption{QQ-plots: comparison between the quantiles of marginal distributions at times $t_{100}, t_{600}, t_{800}$ for real and synthetic data, in the case of OU time series of length $N=1000$. }
    \label{metric_dist_jump_purely}
\end{figure}

\begin{figure}[h]
    \centering
    \begin{subfigure}[b]{0.45\textwidth}
        \includegraphics[width=\textwidth]{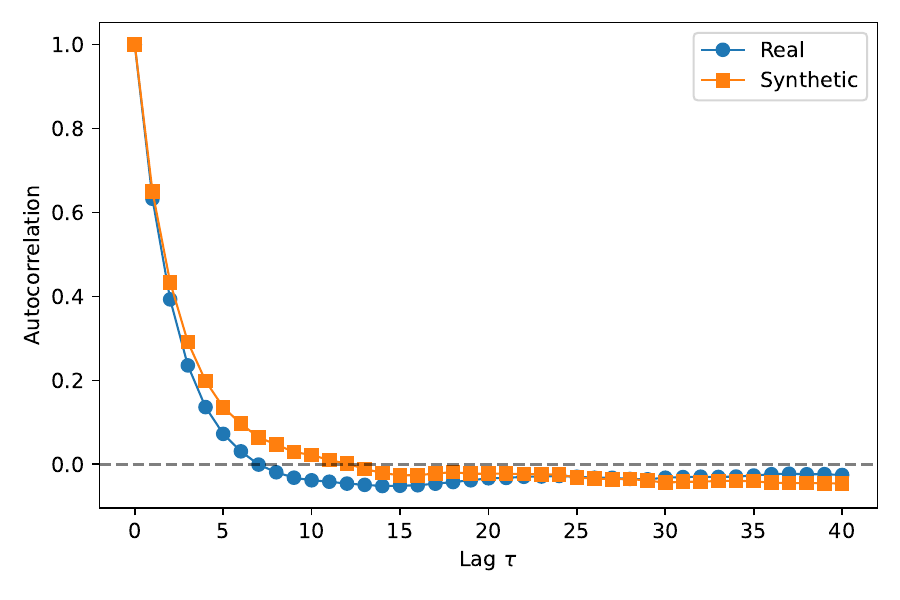}
    \end{subfigure}
    \hfill
    \begin{subfigure}[b]{0.45\textwidth}
        \includegraphics[width=\textwidth]{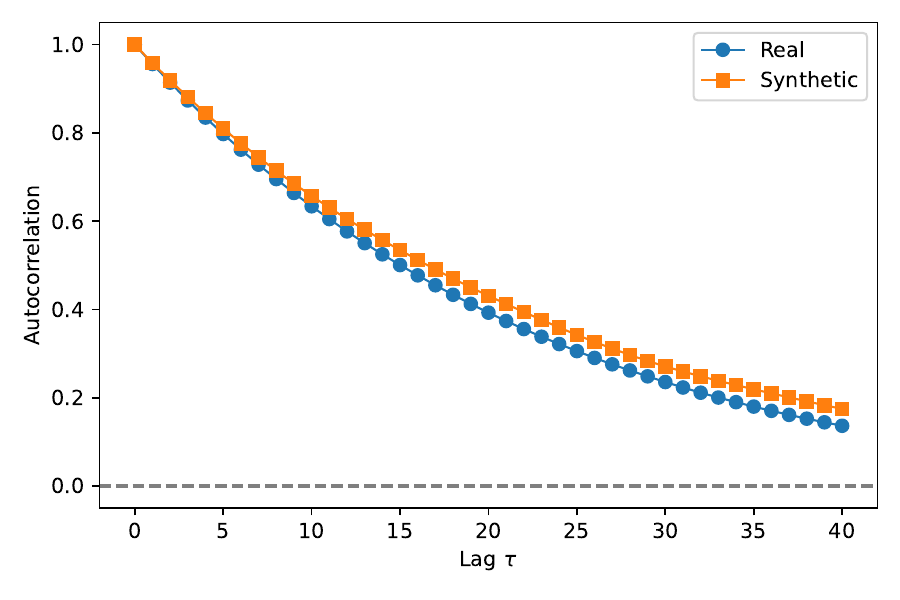}
    \end{subfigure}
    \caption{Empirical autocorrelation function (ACF) for the real and synthetic data (OU time series), in the low-frequency case (introduced in Section \ref{section_OU}) (left) and high-frequency case (right).}
    \label{ACF_high_freq}
\end{figure}

\paragraph{Heston process.} We evaluate the proposed pure jump model on time series generated from the trajectories of a Heston price process $(Y_t)_{t\in[0,T]}$ given by:
\begin{align*}
    \begin{cases}
    \d Y_t = r Y_t \d t + \sqrt{v_t} Y_t \d W_t^Y  \\
    \d v_t = \kappa(\theta - v_t)\d t + \xi \sqrt{v_t} \d W_t^v
    \end{cases}
\end{align*}
where we set $Y_0 = 100$, $v_0 = 0.06$, $r = 0.06$, $\kappa = 1.5$, $\theta = 0.06$, and $\xi = 0.4$. The correlation between the two Brownian motions is set to $\rho = -0.8$. Using the simulated paths of the process $Y$, we construct the initial time series and we generate synthetic time series via the pure jump model (see Figure~\ref{fig:heston_traj}). In this case, our choice of the parameters is $\lambda^0= 0.08$, $c=0$, $\gamma =1$. 

In this case, we analyze the autocorrelation functions (ACF) of both the raw returns and the squared returns. As illustrated in Figure~\ref{fig:heston_acf}, both the reference data and the synthetic outputs successfully replicate the return ACFs fluctuating closely around zero across all evaluated lags. However, a discrepancy emerges in the ACF of the squared returns, which serves as a standard metric for volatility clustering. In the reference Heston dataset, the squared returns ACF exhibits an initial correlation of approximately $0.09$, followed by a slow decay over a $40$-day horizon. Conversely, the synthetic trajectories generated by our model display a significantly muted squared returns ACF, starting at roughly $0.03$ and decaying to zero. We attribute this discrepancy to the fundamental structural difference between the two systems: the baseline data is driven by a stochastic volatility process, whereas the synthetic data is generated by a pure jump model where the memory feedback is localized within the past-dependent kernel estimators.

\begin{figure}[h]
    \centering
    \includegraphics[width=0.7\linewidth]{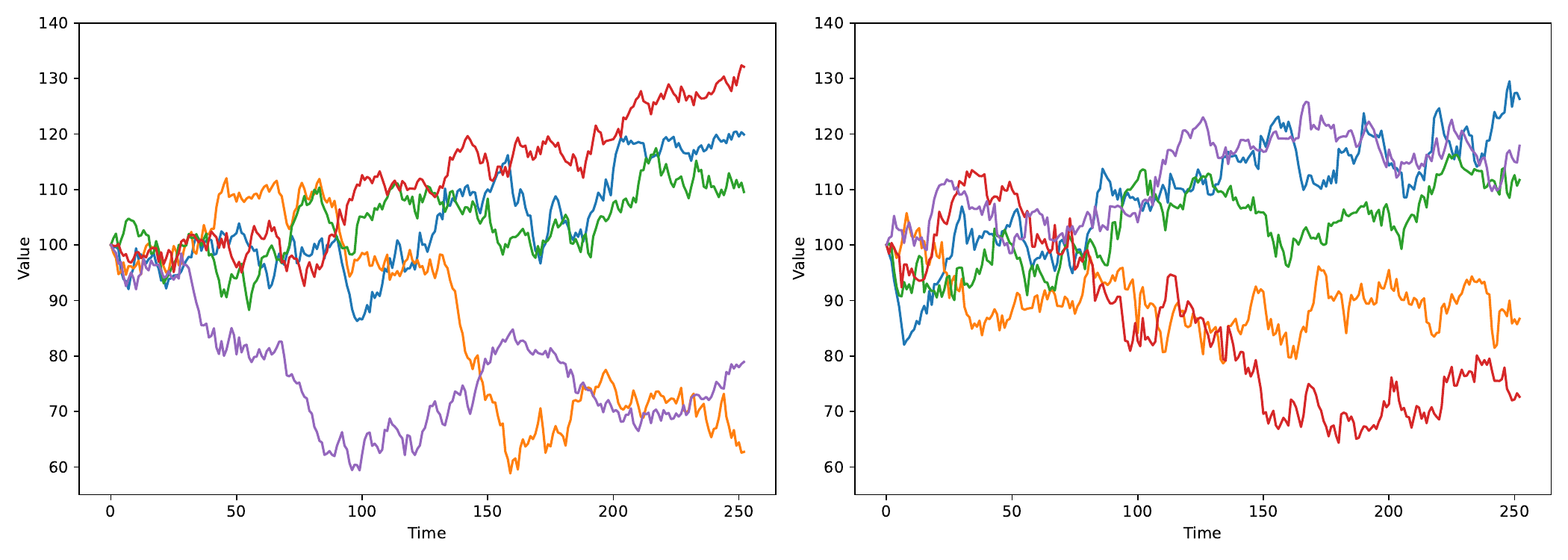}
    \caption{Sample paths of initial time series from trajectories of Heston price process (left) and generated time series by SBJTS (right) using the pure jump model.}
    \label{fig:heston_traj}
\end{figure}

\begin{figure}[h]
    \centering
    \includegraphics[width=0.7\linewidth]{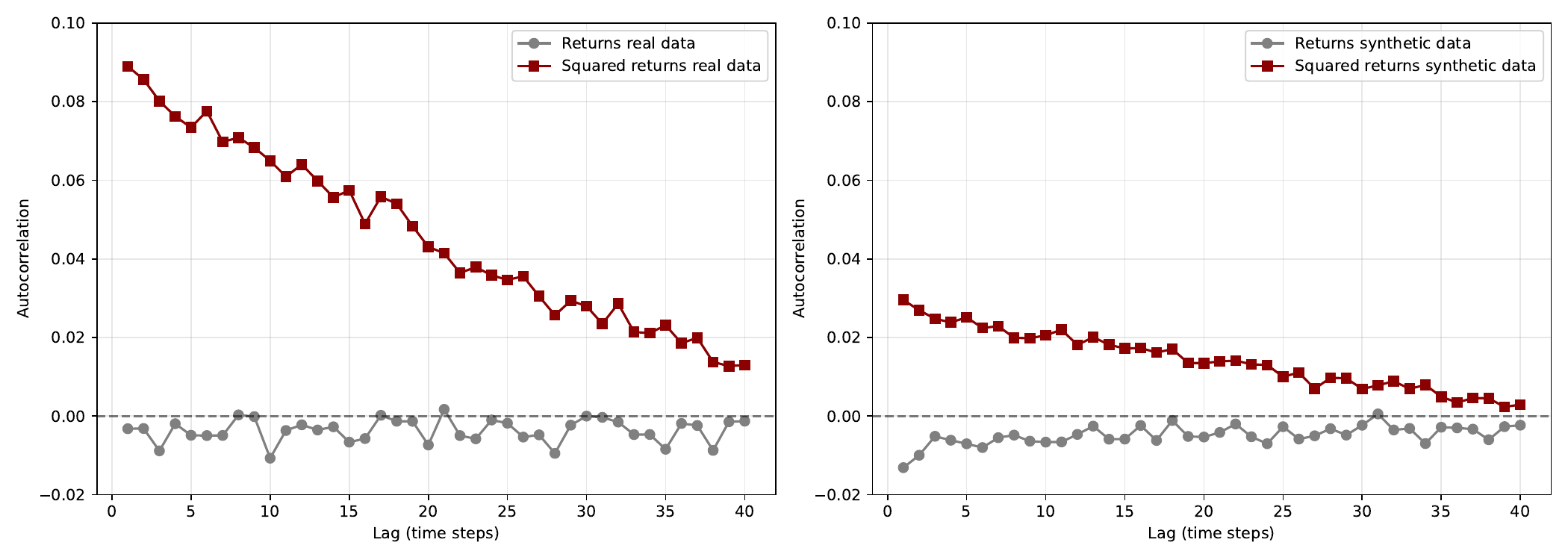}
    \caption{Empirical autocorrelation function (ACF) for the real (left) and synthetic (right) data, in the case of time series of a Heston price process, for both the returns and squared returns.}
    \label{fig:heston_acf}
\end{figure}

\newpage
\subsection{Real Datasets: Qualitative and Quantitative Results}
In addition to the tests reported in Section \ref{section_real}, this appendix presents further quantitative metrics assessing the generative fidelity of the SBJTS framework. We provide additional comparisons for the real datasets: for the Energy dataset we use the synthetic time series generated in Section \ref{section_real}. For the Google dataset we perform the generation using the return series $R_{t_i}=\frac{X_{t_{i+1}}-X_{t_{i}}}{X_{t_{i}}}$, which allows us to better assess statistics specifically related to returns; in this setting, the discriminative score is 0.019 $\pm$ 0.013, while the predictive score is 0.040 $\pm$ 0.001, indicating that the proposed framework continues to generate high-quality synthetic data.

\paragraph{Marginal Distributions.}
To see if the synthesized paths accurately reconstruct the target distribution, we first evaluate the alignment between the marginal distribution of the real data and that of the generated trajectories. In Table \ref{quant_stock}, we report the 5\% and 95\% quantiles of the empirical distribution of the random variable $X_{t_{12}}$ computed on real values and on synthetic values, in the case of the Stock dataset. Overall, the synthetic data closely match the empirical quantiles across all components, also in the case of the sixth component, which presents large fluctuations.
In Table \ref{tab:quantiles_increments} we compare the increments distribution between real (R) and synthetic (S) time series for the Energy dataset for three components exhibiting different dynamical regimes, namely components 2 (with large fluctuations), 15 and 28 (with mean-reverting behaviour). Notably, both small increments (central quantiles) and large increments (tail quantiles) are well reproduced across all cases, highlighting the model’s ability to capture both diffusive and jump-driven dynamics.
\begin{table}[h]
\centering
\begin{tabular}{c|cccc}
\toprule
Component & 5\% Data & 5\% SBJTS & 95\% Data & 95\% SBJTS \\
\midrule
1 & 0.960 & 0.958 & 1.039 & 1.039 \\
2 & 0.968 & 0.967 & 1.034 & 1.032 \\
3 & 0.964 & 0.961 & 1.038 & 1.036 \\
4 & 0.962 & 0.962 & 1.040 & 1.039 \\
5 & 0.962 & 0.960 & 1.040 & 1.039 \\
6 & 0.433 & 0.454 & 2.261 & 2.288 \\
\bottomrule
\end{tabular}
\caption{Quantiles of the marginal distribution at time $t_{12}$ for the Stock dataset over its 6 components: comparison between real and synthetic data.}
\label{quant_stock}
\end{table}
\begin{table}[h]
\centering
\begin{tabular}{c|cc|cc|cc}
\toprule
\multirow{2}{*}{Quantile} 
& \multicolumn{2}{c|}{Component 2} 
& \multicolumn{2}{c|}{Component 15} 
& \multicolumn{2}{c}{Component 28} \\
\cmidrule(lr){2-3} \cmidrule(lr){4-5} \cmidrule(lr){6-7}
& R & S & R & S & R & S \\
\midrule
1\%  & -3.62 & -3.53 & -2.07 & -2.13 & -2.98 & -3.00 \\
5\%  & -1.73 & -1.68 & -1.04 & -1.12 & -2.40 & -2.40 \\
25\% & 0.00  & -0.20 & 0.00  & -0.20 & -1.04 & -1.03 \\
50\% & 0.00  & -0.01 & 0.00  & -0.01 & 0.01  & 0.02 \\
75\% & 0.00  & 0.18  & 0.00  & 0.17  & 1.04  & 1.04 \\
95\% & 1.77  & 1.77  & 1.05  & 1.12  & 2.37  & 2.34 \\
99\% & 3.62  & 3.63  & 2.14  & 2.18  & 2.99  & 2.99 \\
\bottomrule
\end{tabular}

\caption{Quantiles of the increments distributions (R = real, S = synthetic) for components 2, 15, 28 of the Energy dataset.}
\label{tab:quantiles_increments}

\end{table}

\paragraph{Correlation Matrices.}
Figures \ref{correl_matrix_stock} and \ref{correl_matrix} illustrate the empirical correlation matrices evaluated across the components of the Stock dataset and the first 5 components of the Energy dataset, respectively. As discussed in Section \ref{section_real}, the generative process $X$ under the reference measure $\mathbb{P}^0$ assumes independent Brownian and jump components across dimensions. This architectural choice is primarily motivated by the simplification of the calibration procedure, as it drastically reduces the number of parameters requiring initial estimation. However, this choice does not limit the performance of the generative model in capturing the final correlations. Indeed, under the optimal probability measure $\mathbb{P}^*$, the dynamics of the process $X$ are governed by a state-dependent drift vector $\alpha^*_t$ and a jump intensity kernel $\lambda^*_t$, where each component $i \in \{1, \ldots, d\}$ is explicitly defined as a function of the entire joint state vector. Because our non-parametric kernel estimators are evaluated over the full $d$-dimensional historical state space, the estimated drift corrects the trajectories during path simulation, allowing the model to accurately synthesize off-diagonal cross-correlations without the need for a coupled reference covariance matrix.

\paragraph{Autocorrelation function for Google returns.} 
In Figure \ref{ACF}, we report the empirical autocorrelation function of the open, high, low, and close returns for Google stock. Specifically, we compute the autocorrelation between $R_{t_i}$ and $R_{t_{i+k}}$ for lags $k\in\{1, \ldots, 10\}$, averaging the estimates over the time series.
The synthetic data are able to reproduce the main autocorrelation patterns observed in the real data, with a good agreement in the magnitude of the correlations across the different lags.

\paragraph{Empirical joint distributions in dimension 2.}
To analyze the cross-temporal structure of the generated data, we plot the empirical joint distributions of two random variables $(X_{t_i}, X_{t_j})$, with $i, j \in \{1, 12, 24\}$, for component 1 of the Stock dataset in Figure \ref{joint_dist_empirical_Stock}, and component 2 of the Energy dataset in Figure \ref{joint_dist_empirical_Energy}. 

\paragraph{Trajectories of Energy time series.}
For illustrative purposes, in Figure \ref{Trajectories_Energy} we present the trajectories of 20 real and synthetic time series from the Energy dataset for variables with indices 1, 7, and 27, to provide an intuitive check of the generation accuracy more than a rigorous validation test. These plots highlight how the SBJTS generative model is able to reproduce the qualitative behavior of the different variables: it adapts to time series with stronger mean reversion, it captures varying levels of volatility or the presence of peaks, that can be attributed to a jump dynamics. 

\begin{figure}[h]
\centering
\begin{subfigure}{0.45\columnwidth}
  \centering
  \includegraphics[width=\linewidth]{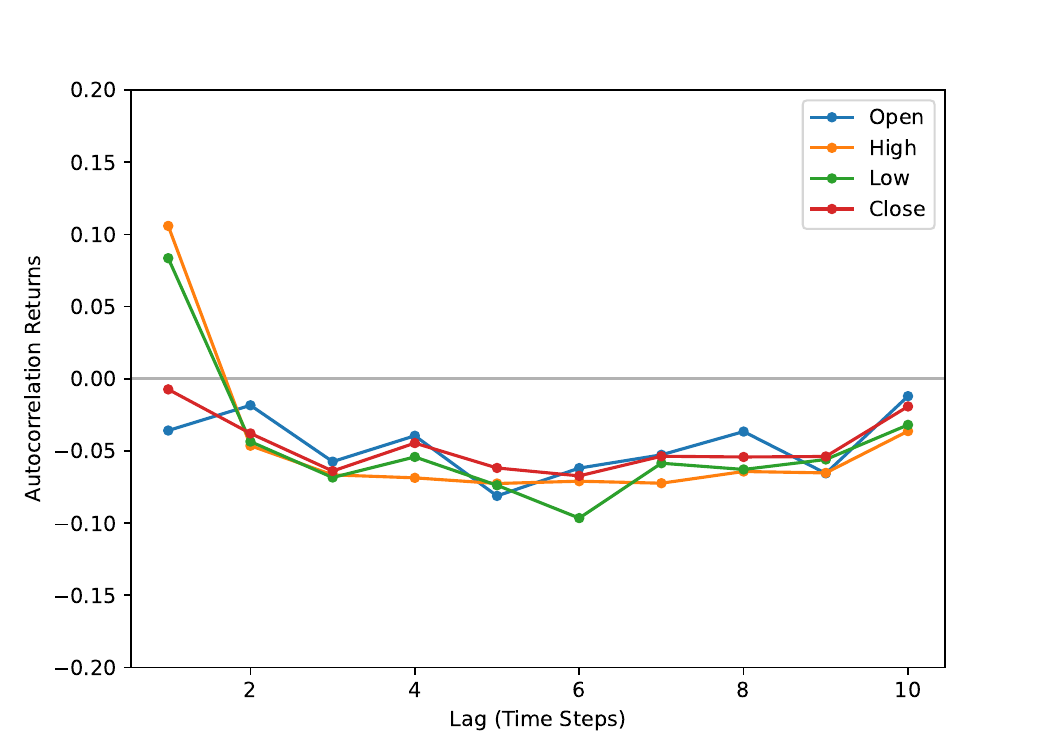}
\end{subfigure}
\hfill
\begin{subfigure}{0.45\columnwidth}
  \centering
  \includegraphics[width=\linewidth]{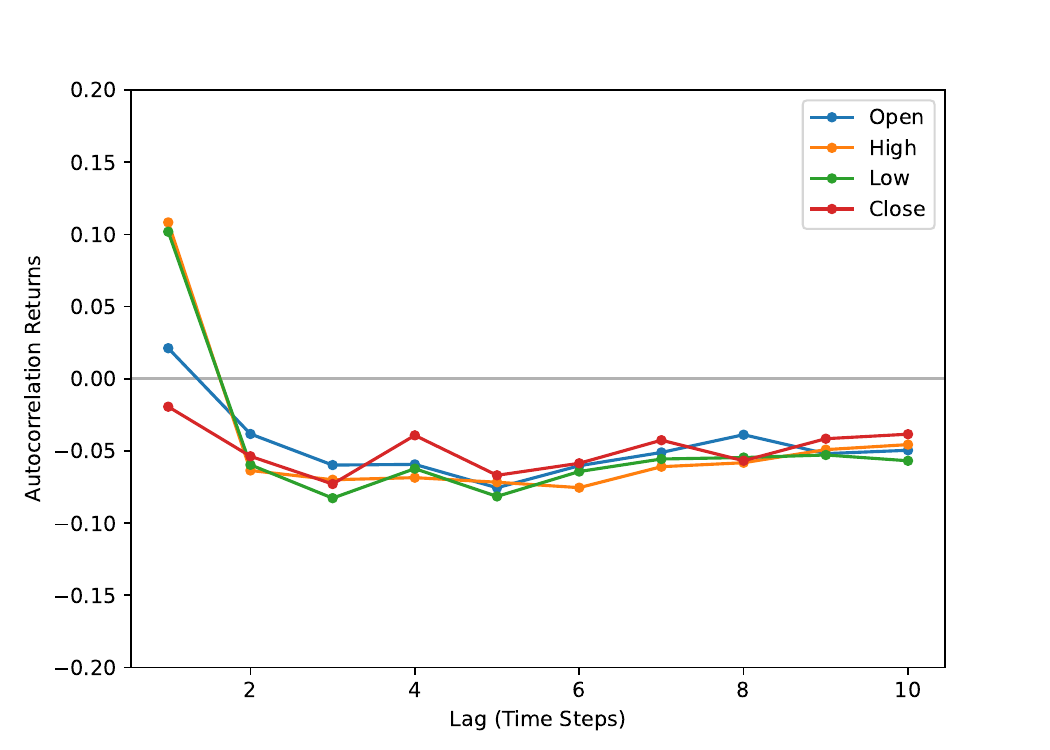}
\end{subfigure}
\caption{Autocorrelation function of Google returns: comparison between real and synthetic data for open, high, low, close prices.}
\label{ACF}
\end{figure}

\begin{figure}[h]
    \centering
    \includegraphics[width=0.8\textwidth]{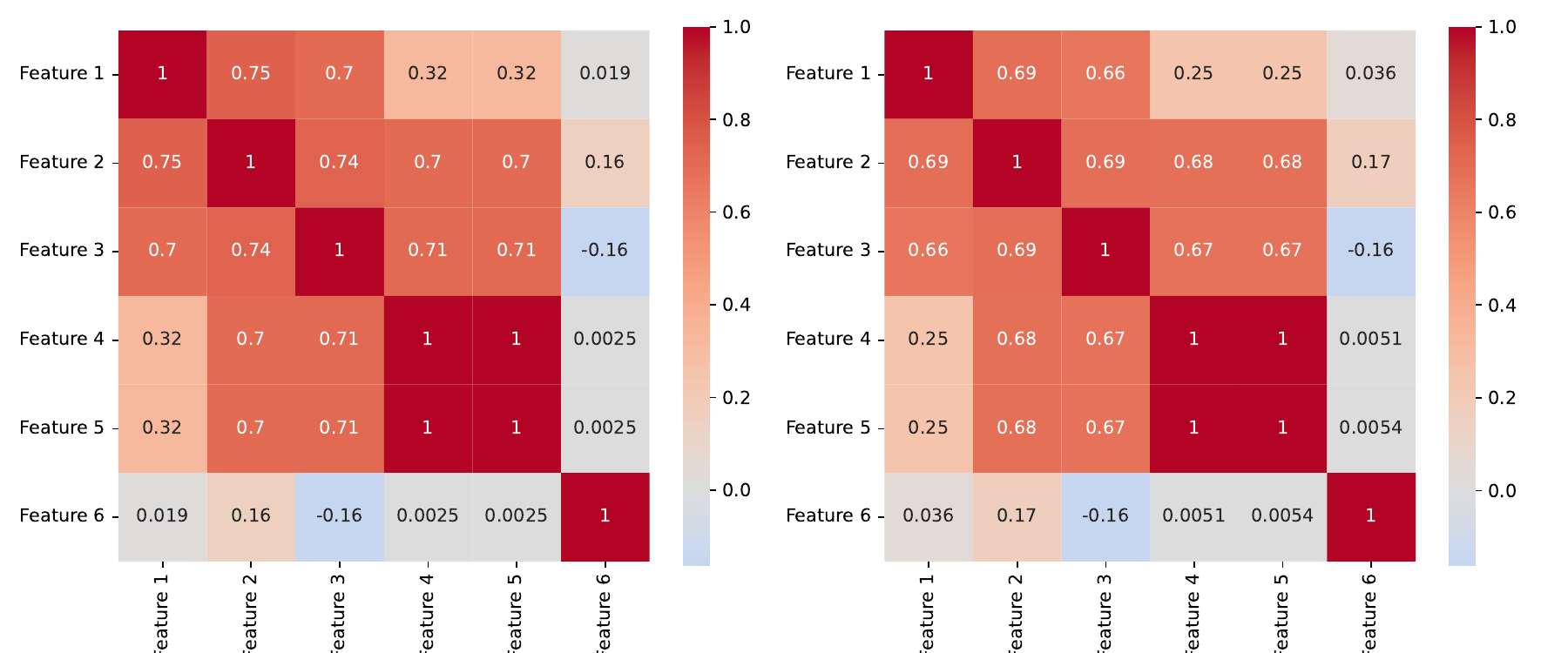}
    \caption{Correlation matrix among the components of the real dataset (left) and synthetic dataset (right) in the case of Stock time series.}
    \label{correl_matrix_stock}
\end{figure}
\begin{figure}[h]
    \centering
    \includegraphics[width=0.8\textwidth]{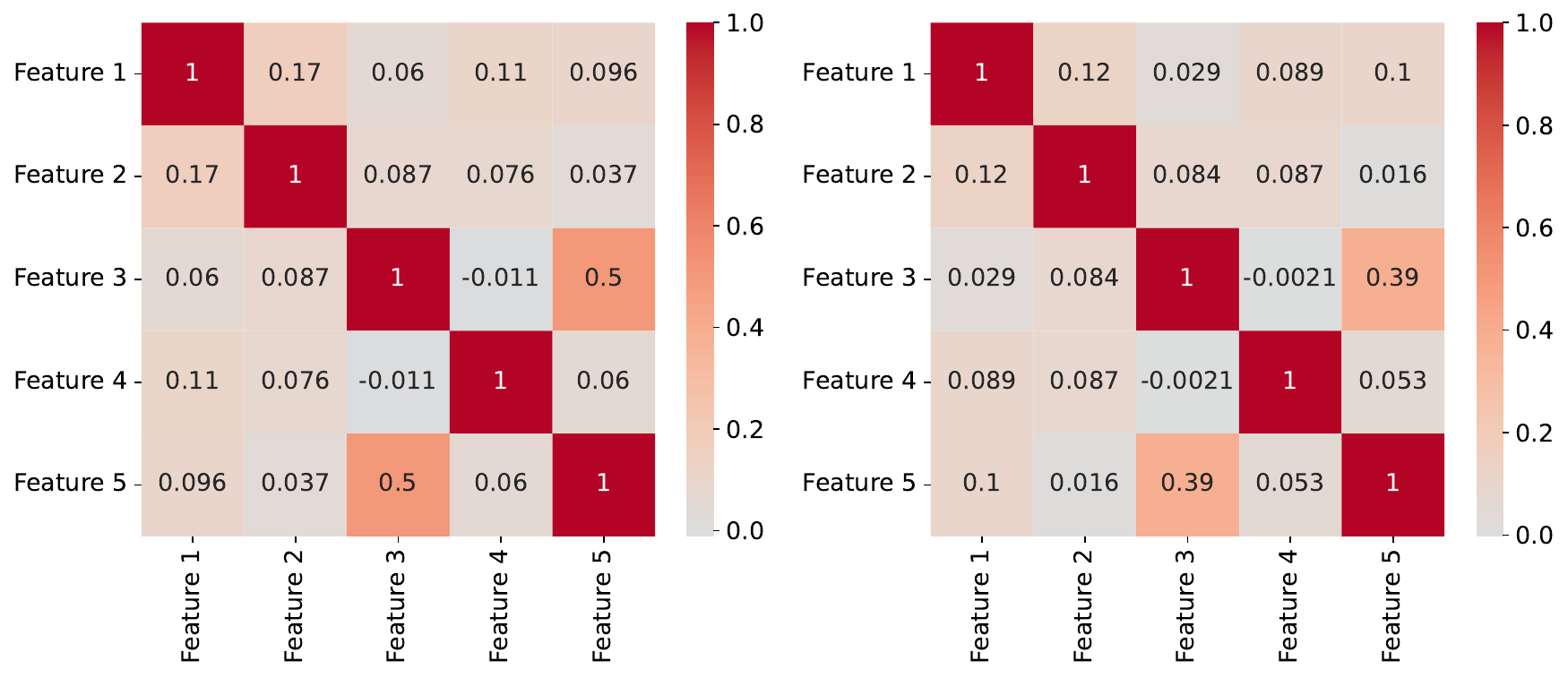}
    \caption{Correlation matrix among 5 components of the real dataset (left) and synthetic dataset (right) in the case of Energy time series.}
    \label{correl_matrix}
\end{figure}

\begin{figure}[t]
\centering
\begin{subfigure}{0.3\columnwidth}
  \centering
  \includegraphics[width=\linewidth]{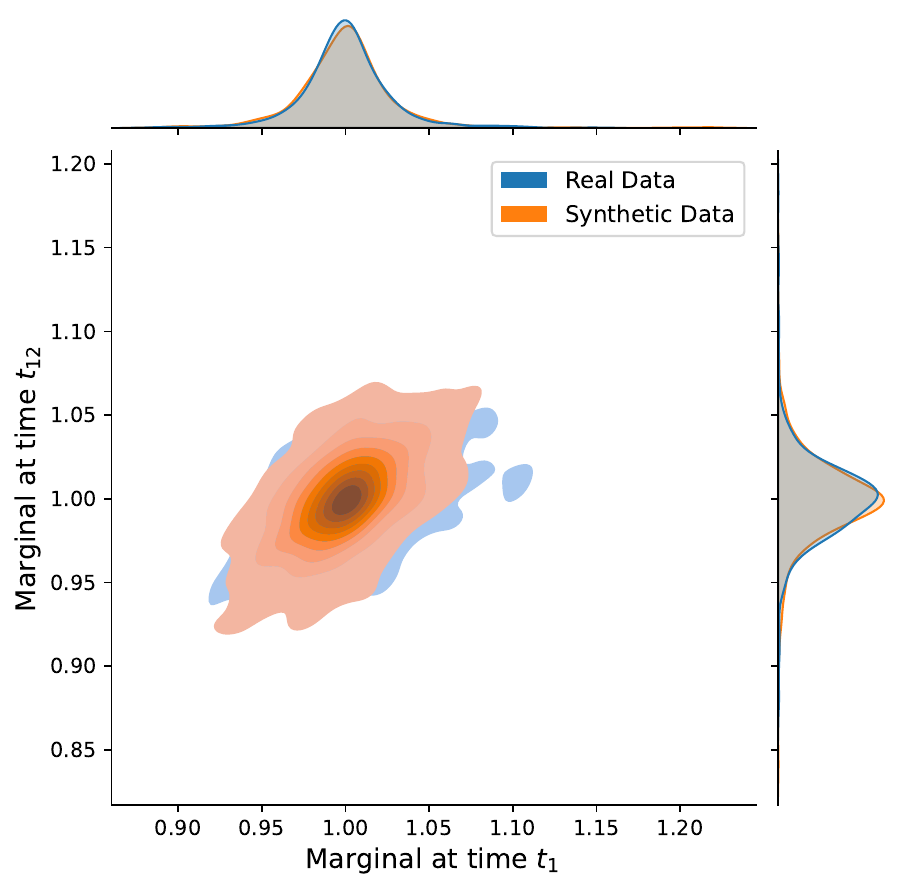}
\end{subfigure}
\hfill
\begin{subfigure}{0.3\columnwidth}
  \centering
  \includegraphics[width=\linewidth]{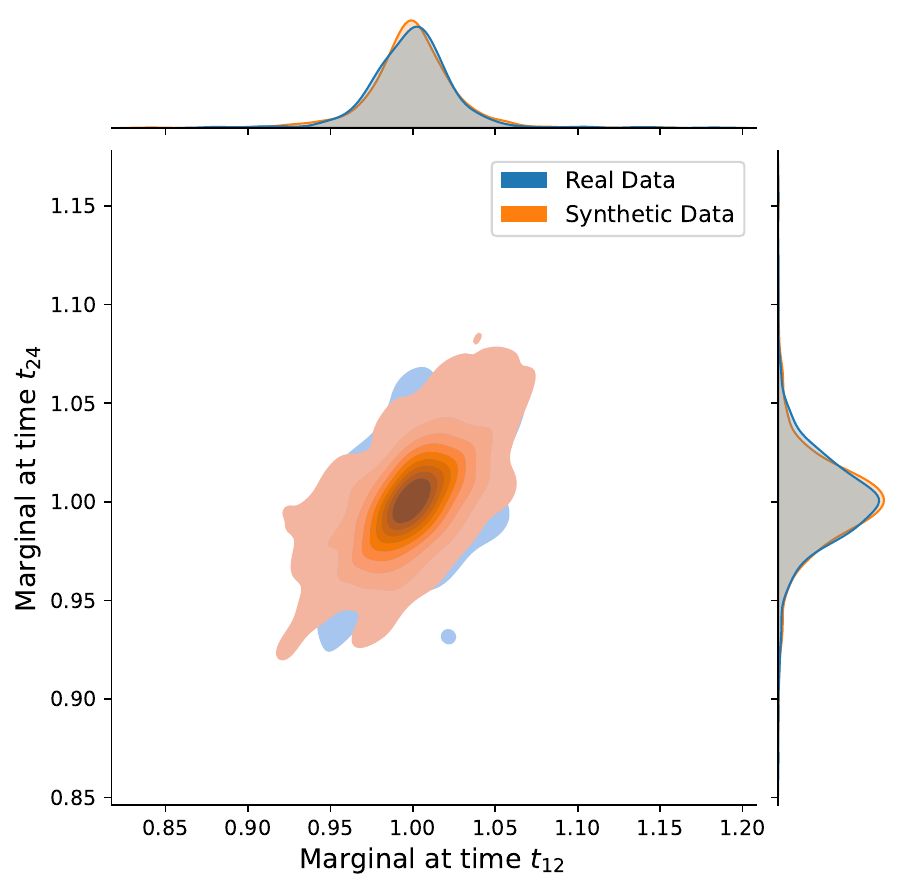}
\end{subfigure}
\hfill
\begin{subfigure}{0.3\columnwidth}
  \centering
  \includegraphics[width=\linewidth]{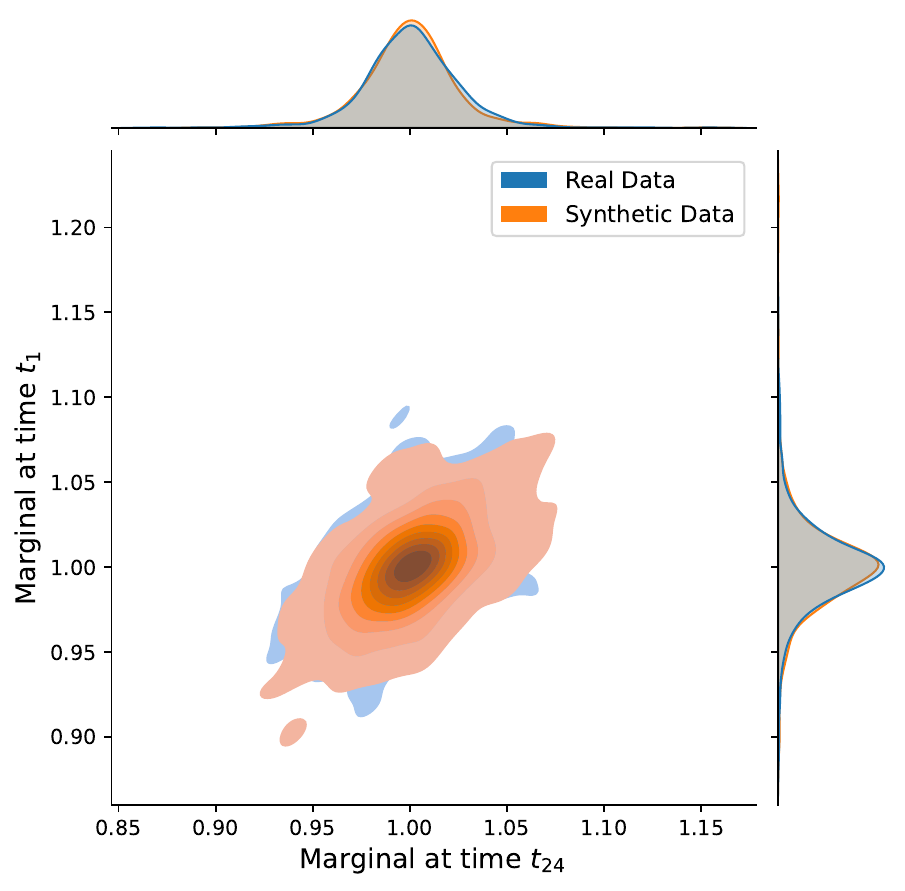}
\end{subfigure}
\caption{Empirical joint distributions $(X_{t_i}, X_{t_j})$ for component 1 of the Stock dataset, with $i, j \in \{1, 12, 24\}$: comparison between real and synthetic distributions.}
\label{joint_dist_empirical_Stock}
\end{figure}

\begin{figure}[t]
\centering
\begin{subfigure}{0.3\columnwidth}
  \centering
  \includegraphics[width=\linewidth]{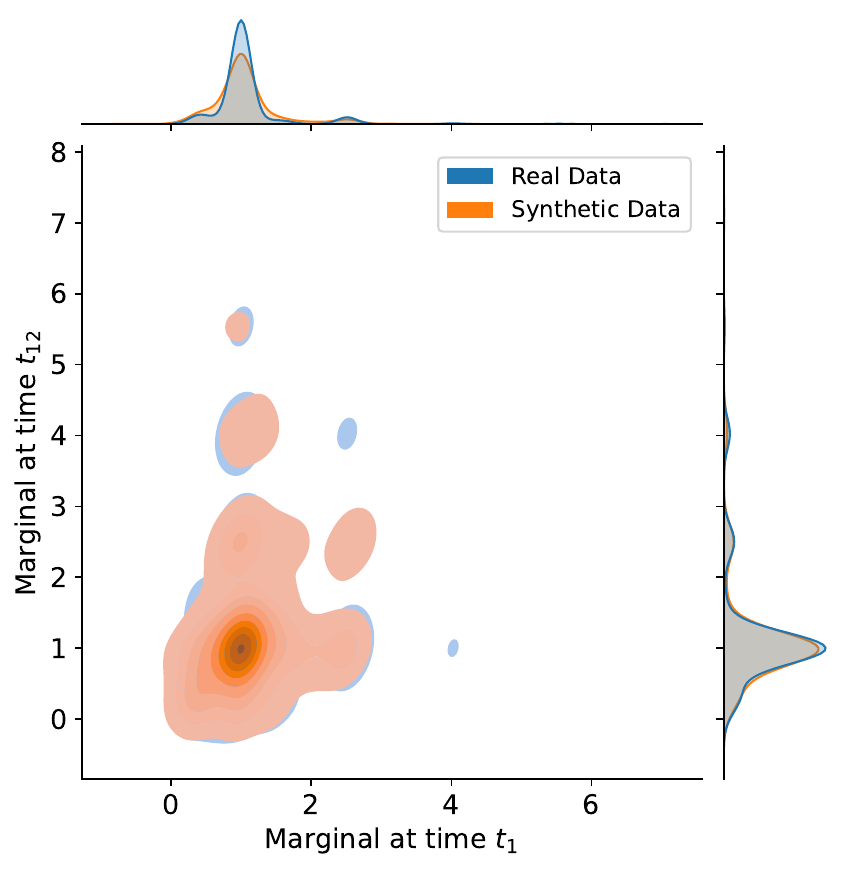}
\end{subfigure}
\hfill
\begin{subfigure}{0.3\columnwidth}
  \centering
  \includegraphics[width=\linewidth]{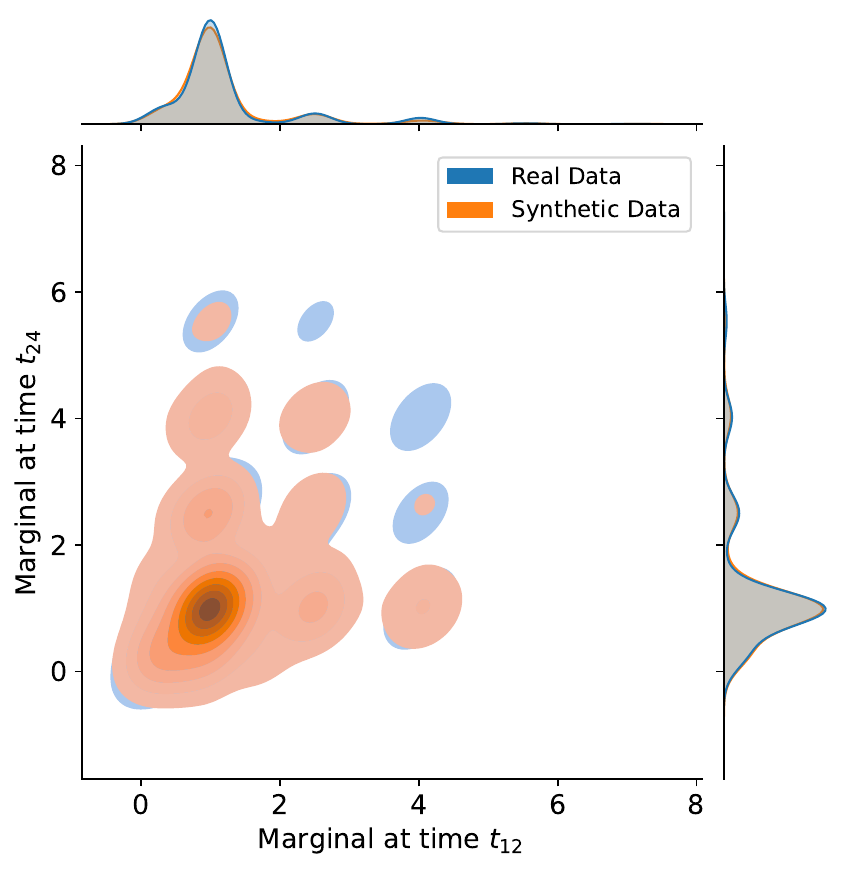}
\end{subfigure}
\hfill
\begin{subfigure}{0.3\columnwidth}
  \centering
  \includegraphics[width=\linewidth]{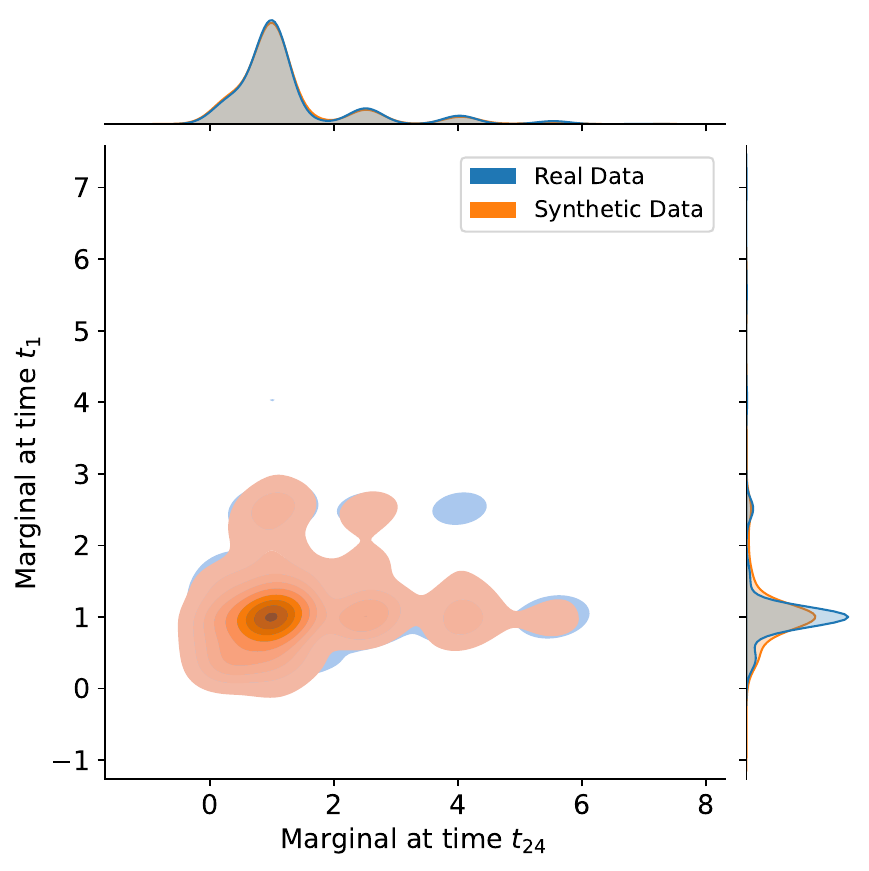}
\end{subfigure}
\caption{Empirical joint distributions $(X_{t_i}, X_{t_j})$ for component 2 of the Energy dataset, with $i, j \in \{1, 12, 24\}$: comparison between real and synthetic distributions.}
\label{joint_dist_empirical_Energy}
\end{figure}

\begin{figure}[t]
    \centering
    \begin{subfigure}{0.8\textwidth}
        \centering
        \includegraphics[width=\textwidth]{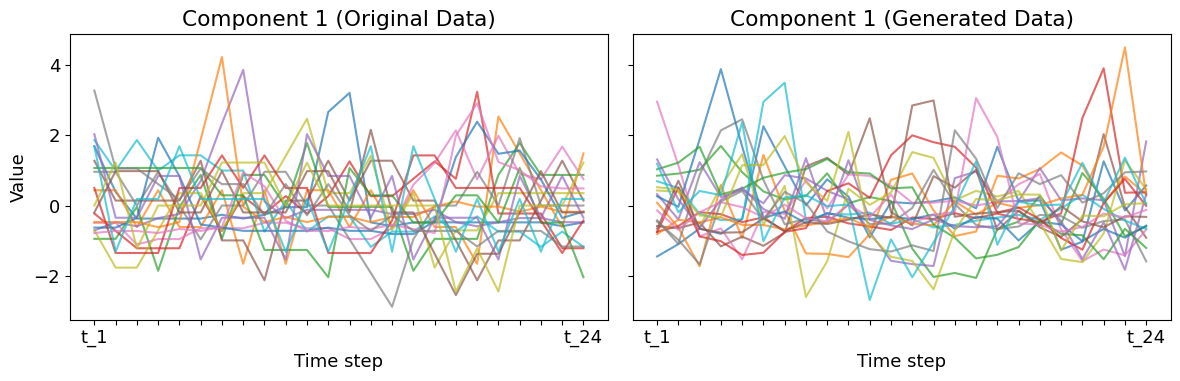}
        \caption{Trajectories of 20 real time series (left) and synthetic time series (right) for the variable 1 (appliances energy consumption) of the dataset Energy.}
        \label{comp_1}
    \end{subfigure}
    \vspace{3mm}
    \par\smallskip
    \begin{subfigure}{0.8\textwidth}
        \centering
        \includegraphics[width=\textwidth]{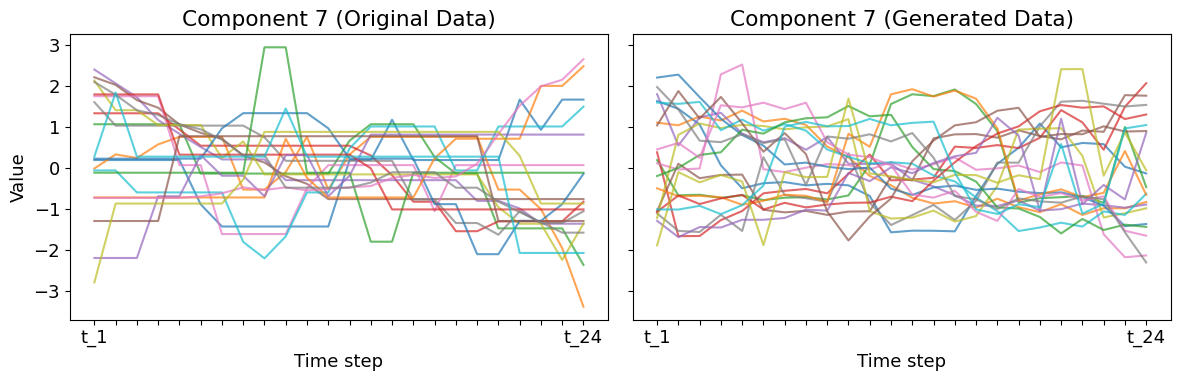}
        \caption{Trajectories of 20 real time series (left) and synthetic time series (right) for the variable 7 (temperature in laundry room area) of the dataset Energy.}
        \label{comp_4}
    \end{subfigure}
    \vspace{3mm}
    \par\smallskip
    \begin{subfigure}{0.8\textwidth}
        \centering
        \includegraphics[width=\textwidth]{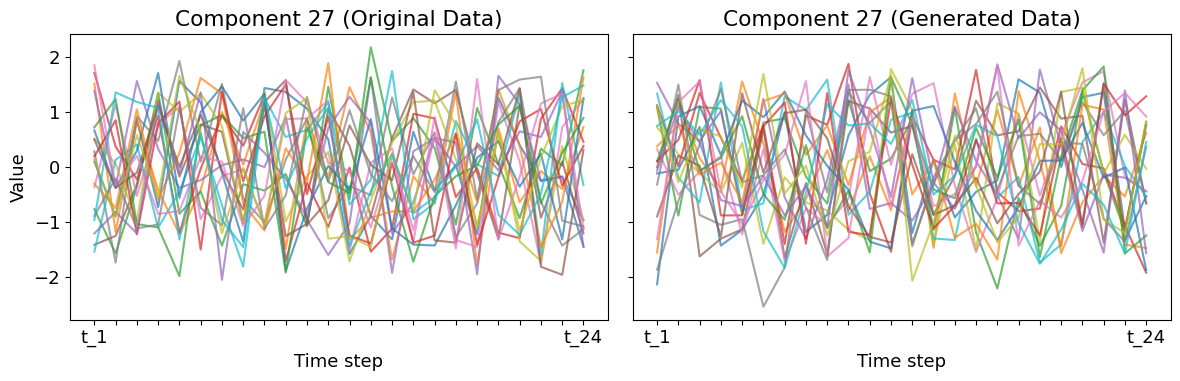}
        \caption{Trajectories of 20 real time series (left) and synthetic time series (right) for the variable 27 (random variable 1) of the dataset Energy.}
        \label{comp_27}   
    \end{subfigure}
    \caption{Trajectories of time series from the Energy dataset.}
    \label{Trajectories_Energy}
\end{figure}

\end{appendices}

\clearpage
\bibliographystyle{abbrv}

\bibliography{bib}

\end{document}